\documentclass[oneside,english]{amsart}
\usepackage[T1]{fontenc}
\usepackage[latin9]{inputenc}
\usepackage{geometry}
\usepackage{color}
\usepackage{babel}
\usepackage{float}
\usepackage{endnotes}
\usepackage{amstext}
\usepackage{amsthm}
\usepackage{amssymb}
\usepackage{graphicx}
\usepackage[unicode=true,pdfusetitle,
 bookmarks=true,bookmarksnumbered=false,bookmarksopen=false,
 breaklinks=false,pdfborder={0 0 1},backref=false,colorlinks=false]
 {hyperref}

\makeatletter
\numberwithin{equation}{section}
\numberwithin{figure}{section}
\let\footnote=\endnote
\theoremstyle{plain}
\newtheorem{thm}{\protect\theoremname}
\theoremstyle{definition}
\newtheorem{defn}[thm]{\protect\definitionname}
\theoremstyle{remark}
\newtheorem{claim}[thm]{\protect\claimname}
\theoremstyle{remark}
\newtheorem{rem}[thm]{\protect\remarkname}
\theoremstyle{plain}
\newtheorem{prop}[thm]{\protect\propositionname}
\theoremstyle{plain}
\newtheorem{lem}[thm]{\protect\lemmaname}

\usepackage{xcolor,pifont}
\usepackage{epstopdf}
\usepackage{tikz}

\makeatother

\providecommand{\claimname}{Claim}
\providecommand{\definitionname}{Definition}
\providecommand{\lemmaname}{Lemma}
\providecommand{\propositionname}{Proposition}
\providecommand{\remarkname}{Remark}
\providecommand{\theoremname}{Theorem}

\begin{document}
\title{Smart Proofs via Smart Contracts: \\
Succinct and Informative Mathematical Derivations \\
via Decentralized Markets }
\author{Sylvain Carré$^{\spadesuit}$$^{*}$}
\author{Franck Gabriel$^{\dagger}$$^{*}$}
\author{Clément Hongler$^{\dagger}$$^{*}$}
\author{Gustavo Lacerda}
\author{Gloria Capano}
\date{12 October 2021}
\thanks{$^{*}$Equal Contribution}
\thanks{$^{\spadesuit}$Université Paris-Dauphine, Department of Economics}
\thanks{$^{\dagger}$École Polytechnique Fédérale de Lausanne, Institute of
Mathematics, Chair of Statistical Field Theory}
\begin{abstract}
Modern mathematics is built on the idea that a proof should be translatable
into a formal proof, whose validity is an objective question, decidable
by a computer. In practice, however, proofs are informal, succinct,
and omit numerous uninteresting details: their goal is to share insight
among a community of agents. An agent considers a proof valid if they
trust that it could (in principle) be expanded into a machine-verifiable
proof. A proof\textquoteright s validity can thus become a subjective
matter, possibly leading to a debate; if agents\textquoteright{} incentives
are not aligned, it may be hard to reach a consensus. Hence, while
the concept of valid proof is well-defined in principle, the process
to establish a proof\textquoteright s validity is itself a complex
multi-agent problem. 

In this paper, we introduce the SPRIG (Smart Proofs via Recursive
Information Gathering) protocol, which allows agents to propose and
verify succinct and informative proofs in a decentralized fashion;
the trust is established by agents being able to request more details
at steps where they feel there could be problems; debates, if they
arise, need to isolate specific details of proofs; if they persist,
they must go down to machine-level details, where they can be settled
automatically. A structure of fees, bounties, and stakes is set to
incentivize the agents to act in good faith, i.e. to not publish problematic
proofs and to not ask for trivial details. 

We propose a game-theoretic discussion of SPRIG, illustrating how
agents with different types of information interact, leading to a
verification tree with an appropriate level of detail, and to the
invalidation of problematic proofs, and we discuss resilience against
various attacks. We then provide an in-depth treatment of a simplified
model, characterize its equilibria and analytically compute the agents\textquoteright{}
level of trust. 

The SPRIG protocol is designed so that it can run fully autonomously
as a smart contract on a decentralized blockchain platform, without
a need for a central trusted institution. This allows agents to participate
anonymously in the verification debate, being incentivized to contribute
with their information. The smart contract mediates all the interactions
between the agents, and settles debates on the validity of proofs,
and guarantees that bounties and stakes are paid as specified by the
protocol. 

SPRIG also allows for a number of other applications, in particular
the issuance of bounties for solving open problems, and the creation
of derivatives markets, enabling agents to inject more information
pertaining to mathematical proofs.

In addition to presenting the key novelties of the SPRIG protocol,
this whitepaper is also designed as a brief survey. In order to reach
a broader audience, we recall a variety of historical and technical
details.
\end{abstract}

\maketitle

\section{\label{sec:introduction}Introduction}

\subsection{\label{subsec:mathematical-derivation}Mathematical Proofs}

Mathematical derivation, also sometimes called logical reasoning,
rigorous derivation, formal rational reasoning, or mathematical proof
is a process that allows one to derive mathematical statements from
other mathematical statements. By relying on a collection of statements
accepted to be fundamentally true, called axioms, this mechanism allows
one to derive new mathematical truths, called proven statements. Depending
on the context, such proven statements are also called propositions,
theorems (when they are deemed interesting), or lemmas (when they
are ancillary in the derivation of theorems); the derivation leading
to a statement (starting from another statement, assumed or already
established to be true) is called its proof. 

This form of reasoning is at the heart of rational thinking (in mathematics,
all the sciences, and way beyond), crucially leading to:
\begin{itemize}
\item One's trust in the truth of statements derived.
\item One's insight into the reasons why such statements hold true.
\end{itemize}
These two aspects of the question are discussed in Sections \ref{subsec:proofs-for-trust}
and \ref{subsec:proofs-for-explanation} below. 

\subsubsection{\label{subsec:proofs-for-trust}Proofs as a Means of Trust}

In many ways, mathematically proven statements achieve the highest
possible level of certainty one can have. The validity of an established
theorem (say, for instance, Euclid's theorem on the existence of an
infinite number of primes) does not  fluctuate with the evolution
of knowledge. Conversely, for statements for which no known proof
exists, the trust in their validity only grows progressively with
empirical or heuristic evidence support, and it never quite reaches
that of mathematically proven statements. 

Once the trust in a statement is established via a proof, the statement
can be used as a basis for establishing trust in new statements: the
proofs of all the statements obtained this way can (if needed) be
`unrolled' down to the axioms. Hence, by propagation, various agents
are able to build together a set of trusted statements relying upon
each other (a `tower of knowledge') without necessarily knowing all
the details of all the proofs. This is the way modern mathematics
is built. 

More recently, due to the development of computer technologies, proofs
have become fundamental not only for the construction of mathematics
and science, but they have also become objects manipulated by e.g.
cryptography or program verification systems. Most digital interactions
are now mediated by cryptographic primitives, which aim in particular
at establishing trust in confidentiality and authenticity: for instance,
a party can prove their identity by providing a digital signature
(more precisely, the mathematically proven statement is: either the
party knows a secret encryption key, or they are extremely lucky,
or a certain widely believed algorithmic hardness assumption is in
fact wrong). Similarly, for critical systems, there is often a need
for a proof that a piece of code meets some specification, such as
termination or type safety.

Of course, in any case, the trust in a mathematically proven statement
relies on the trust that the underlying proof is indeed correct, i.e.
that:
\begin{itemize}
\item The derivation rules are clearly defined, computable and consistent. 
\item They are applied uniformly throughout the proof of the statement,
as verified by computers or other mathematicians.
\item The premises and axioms underlying the reasoning can be trusted.
\end{itemize}
These issues and the underlying challenges are briefly discussed in
Section \ref{subsec:nature-of-mathematical-derivations} below.

\subsubsection{\label{subsec:proofs-for-explanation}Proofs as a Means of Explanation
and Insight }

Arguably more important to most mathematicians' minds than the idea
that a mathematical derivation answers the question `how do we know
this statement is true?', is the resulting insight that it provides
into the nature of the underlying problem (see e.g. \cite{thurston}). 

For instance, there are famous statements such as Goldbach's conjecture
that are already widely believed to be true (as supported by various
heuristics and numerical verifications), despite being unproven. Finding
a proof of Goldbach's conjecture would be considered a breakthrough
not so much because it would tell us that it is true (which would
hardly surprise anyone), but because it would tell us why, and because
it may give some new deep insight into the nature of prime numbers. 

The idea that proofs give insight is arguably a central driving force
behind mathematical teaching and exposition: it motivates the great
effort that goes into the presentation of theorems' proofs. Similarly,
the work towards finding new, simpler, more intuitive, or just different
proofs of existing theorems is greatly valued: multiple proofs of
the same statement can bring insights on various issues, a complementary
value \cite{aigner-ziegler}, or enhance one's intuition of a result. 

Conversely, certain proofs appear to bring little insight because
of their complexity (for instance, computer-generated proofs such
as the one of the four-color theorem \cite{appel-haken,thurston}).
While the trust in such proofs is very high, they are considered somewhat
unsatisfactory by many mathematicians, as they cannot be comprehended
as well as more `elegant' proofs. 

Often, insightful proofs appear to be centered around a limited number
of new ideas. In fact, this seems to be how educated agents should
convince each other: convey credence about statements through the
transmission of a limited amount of relevant information.

A balance between this idea of proof and that of Section \ref{subsec:proofs-for-trust},
i.e. between `a short collection of insightful statements' and `the
list of all the statements needed to establish perfect trust', is
in principle possible, though somewhat delicate, as discussed in Section
\ref{subsec:nature-of-mathematical-derivations} below. 

The SPRIG (Smart Proofs via Recursive Information Gathering) protocol
allows, for a system of agents with various levels of interest and
information, to reach this balance. 

\subsection{\label{subsec:nature-of-mathematical-derivations}Nature of Mathematical
Derivations}

In this subsection, we describe the modern view of the idea of mathematical
derivation, as it emerged at the beginning of the 20th century, that
now serves as the basis of all contemporary mathematics. This view
also lies at the heart of all computer-based proof systems (Section
\ref{subsec:recent-developments-in-computer-based-proofs}) and has,
in recent years, faced a number of practical challenges (Section \ref{subsec:challenges-in-modern-mathematical-derivations}
).

\subsubsection{Hilbert's Program and Logicism}

The clarification of the foundations of mathematics, as advocated
by Hilbert's program, progressed greatly in the early 20th century.
One key step was the idea of logicism \cite{korner}: that mathematical
statements should be written in a formal language (or unambiguously
translatable into one), and that a mathematical derivation ought to
consist of a sequence of well-defined manipulations of such statements.
This is naturally connected to the then-emerging field of computer
science: mathematical derivations, written in the appropriate language,
ought to be verifiable by a computer program. 

These new foundations led to two key developments: 
\begin{itemize}
\item On the one hand, the emergence of the study of valid formal mathematical
derivations as central objects, as a field in itself (namely proof
theory), primarily associated with mathematical logic and theoretical
computer science (see Section \ref{subsec:logic-view-of-mathematical-proofs}
below);
\item On the other hand, the emergence of modern mathematics, with more
rigorous and standardized definitions, theorems, and proofs, implicitly
relying on the new foundations (see Section \ref{subsec:modern-mathematical-derivations-in-practice}
below). 
\end{itemize}
Interestingly, despite immense progress in computer technology, these
two developments have only seen little interaction. Arguably, this
is due to practical (rather than fundamental) reasons, and recent
challenges and developments suggest that this has led to a somewhat
unfortunate situation (see Section \ref{subsec:recent-developments-in-computer-based-proofs}
below). The objective of the present paper is indeed to present a
new way to make this interaction practical and fruitful. 

\subsubsection{\label{subsec:logic-view-of-mathematical-proofs}Logic View of Mathematical
Derivations}
\begin{quote}
\textquotedbl The development of mathematics toward greater precision
has led, as is well known, to the formalization of large tracts of
it, so that one can prove any theorem using nothing but a few mechanical
rules\textquotedbl . -- K. Gödel.
\end{quote}
The controversies of the late 19th century led to convergence on the
foundational theory of mathematics, by which the disagreements could
be resolved. Ideas of Frege, Hilbert, Gödel, Turing, and others, led
to a definition of formal proof, connected with the notion of computer
(itself defined in terms of Turing machines; see e.g. \cite{widgerson}
for a modern account). 
\begin{defn}
A formal proof or formal derivation, or machine-level proof, or proof-object
is a finite sequence of sentences in a formal language, each of which
is an axiom, an assumption, or follows from the preceding sentences
in the sequence, by a rule of inference. The validity of the application
of the rules of inference can be checked by a computer. 
\end{defn}

In practice, most of modern mathematics, by convention, relies on
a standard set of axioms (based on the Zermelo-Fraenkel set theory,
with some version of the Axiom of Choice), and on higher-order logic.
A number of computer proof systems implement such a framework (see
Section \ref{subsec:automated-proofs} below). 

Besides the important clarifications that they bring, the strength
of formal proofs is that the verification of their validity is completely
mechanical. As a result, they can be checked reliably by computers.
In principle, computers can also then be used to try to produce proofs,
as was already suggested by Gödel in his Lost Letter to Von Neumann
\cite{lipton}. The use of a formal language also allows in principle
various areas of mathematics to communicate with each other, by allowing
them to inter-operate unambiguously.

Unfortunately, formal proofs are extremely long in practice \cite{wiedijk}
and hard to produce, even with modern computer proof assistants (see
Section \ref{subsec:automated-proofs} below); and most importantly,
they are often quite different from the way mathematicians think of
proofs (a formal proof may bring little insight to a mathematician).
Still, the stronger rigor of such proofs has influenced the shaping
of modern mathematical derivations as mathematicians use them today,
as described in Section \ref{subsec:modern-mathematical-derivations-in-practice}
below: they must stay `in the back of mathematicians' minds', as they
write and communicate their proofs, something that the SPRIG protocol
allows us to formalize. 

\subsubsection{\label{subsec:modern-mathematical-derivations-in-practice}Modern
Mathematical Derivations in Practice}

The core of mathematical activity has seen little change over the
last century, and the focus of mathematics has shifted away from formal
verification and foundational issues, in favor of introducing new
objects, discovering new ideas, and solving interesting problems.
At the same time, mathematics still emphasizes strict rigor: for instance,
heuristic arguments or numerical simulations, however convincing,
are not accepted as being parts of mathematical derivations and proofs,
while, for instance in theoretical physics derivations, they are often
deemed sufficient. 

The following defines what it means for a mathematician to know the
proof of a statement:
\begin{claim}
\label{claim:knowing-how-to-prove-a-statement}A mathematician knows
how to prove a statement rigorously, if they have the confidence in
the following: given access to a corpus of references, they would
be able, if pressed and given enough time, to give details at an arbitrarily
high level in the proof of each statement, down to a computer-checkable
formal level if needed.
\end{claim}

The working definition of `a proof' that a modern mathematical text
uses can be phrased as follows:
\begin{claim}
\label{claim:written-proof-definition}A written proof consists of
\begin{itemize}
\item A Proof Sketch $\mathbf{P}=\mathbf{D},\mathbf{S}_{1},\ldots,\mathbf{S}_{k}$,
consisting of a collection of definitions and references $\mathbf{D}$
and a list of statements (lemmas, propositions, theorems, remarks)
$\mathbf{S}_{1},\ldots,\mathbf{S}_{k}$ using symbols in $\mathbf{D}$,
where each $\mathbf{S}_{j}$ is allowed to assume that $\mathbf{S}_{i}$
holds true for $i<j$.
\item A text in free format $\mathbf{F}$ (including proof arguments, drawings,
informal explanations, etc.).
\end{itemize}
such that it is claimed that the proof is complete and valid in the
eyes of mathematicians (of the given audience), in the following sense:
\end{claim}

\begin{claim}
\label{claim:complete-and-valid}A written proof $\left(\mathbf{P},\mathbf{F}\right)$
of a statement $\mathbf{S}$, with $\mathbf{P}=\mathbf{D},\mathbf{S}_{1},\ldots,\mathbf{S}_{k}$
is considered complete and valid in the eye of a mathematician if,
by using using the text from $\mathbf{F}$ and standard mathematical
knowledge if needed, they know how to prove: 
\begin{itemize}
\item for each $j=1,\ldots,k$ the statement $\mathbf{S}_{j}$ assuming
(if needed) $\mathbf{S}_{1},\ldots,\mathbf{S}_{j-1}$;
\item the statement $\mathbf{S}$ from the statements $\mathbf{S}_{j}$,
where $j=1,\ldots,k$. 
\end{itemize}
\end{claim}

\begin{rem}
Another way to phrase the structure of the proof sketch $\mathbf{P}$
is to say that the statements $\mathbf{S}_{1},\ldots,\mathbf{S}_{k},\mathbf{S}$
form a directed acyclic graph of dependence, with root $\mathbf{S}$
(where a statement points to the statements it assumes). The order
in which the parts $\mathbf{P}$ are presented in a paper may not
follow the order here, but the vertices of any directed acyclic graph
can be ordered so the vertex $i\to j$ implies $i>j$. 
\end{rem}

\begin{rem}
In mathematical papers, definitions-statements may appear (for instance,
defining Riemann's $\zeta$ function may require a proof of convergence);
see Section \ref{subsec:claim-of-proof-format} for a discussion of
how such definitions-statements can be recast in the format of proof
sketches above. 
\end{rem}

\begin{rem}
Proofs by contradictions can be written in the proof sketch format
as above (see Section \ref{subsec:claim-of-proof-format}).
\end{rem}

The free part $\mathbf{F}$ of a proof is what is sometimes called
a Social Proof \cite{buss}, and the proof sketch part $\mathbf{P}$
should be directly translatable into a collection of formal statements,
sometimes called a Formal Proof Sketch \cite{wiedijk}. In this article,
we will consider that the proof sketches are always formal. It should
be noted that in mathematical articles, these two parts are sometimes
not clearly separated; also, in some cases, the entire proof is a
social proof. 

\subsubsection{\label{subsec:amount-of-details-in-a-proof}Amount of Detail in Proofs}

As explained in Section \ref{subsec:modern-mathematical-derivations-in-practice}
above, `being convinced by a proof given as a mathematical text' means:
knowing enough to be confident that one would be able to produce all
the details, if needed, while at the same time knowing that this most
likely will not be needed. Indeed, it would be so energy-consuming
and uninformative that there would be no point in doing it. In the
language of structured proofs above, every mathematician must build
their own structured proof at a level of detail that they deem satisfying. 

Being a mathematician hence crucially requires a great deal of self-discipline,
in order not to delude oneself into thinking that one knows how to
prove a statement. How to be confident in one's ability to perform
a task (providing machine language proof of new results), that one
will (in all likelihood) never perform?

Moreover, what constitutes a complete proof (in the sense of Section
\ref{subsec:modern-mathematical-derivations-in-practice}) becomes
as a result somewhat subjective, depending on the reader's standards:
the amount of detail required from a student at an exam will typically
be very different from the level of detail in a research paper. 

For research papers, it will depend on the subfield and the journal
and editorial standards: a debate between the authors and referees
can arise, in which the editor is the arbiter.

In any case, determining the relevant amount of detail to provide
in a mathematical paper is a difficult task, which requires a delicate
balance between the need to guarantee the validity of the reasoning,
to limit the pre-requisites and the work on the reader's side, to
stay within page limits, to avoid unnecessary clutter, and to keep
only the essential arguments. 

Within these constraints, there is a lot of room for subjective choices
and writing modern mathematical proofs is largely an art, as discussed
in e.g. \cite{lamport-i,lamport-ii}. The increasing complexity of
proofs involved in contemporary mathematics has led to a number of
challenges, as discussed in Section \ref{subsec:challenges-in-modern-mathematical-derivations}
below. At the same time, the development of computer-based proof systems
offers great promises to help tackle such challenges.

The goal of SPRIG protocol is to unify these two visions of formal
proofs, as verified by computers, and of informal proofs as done by
mathematicians, to leverage the advantages of both (trust and insight,
respectively): SPRIG will allow the agents to inject various levels
of information to reveal a subtree of the proof tree as in Section
\ref{subsec:claim-of-proof-format} below. 

\subsubsection{\label{subsec:structured-proofs}Structured Proofs}

An interactive way of viewing proofs as in Claim \ref{claim:written-proof-definition}
is the following: instead of asking to be able to write down a proof
of the formal statements in the machine-level language directly, one
can think of being able to answer requests for formal details for
each of the statements (assuming perhaps an audience that is more
and more curious into the details); and in our answer, if any requests
for more details arises, one should again be able to provide them.
This way, one should eventually be able to reach (if needed) the machine
level after a reasonable number of steps, using a reasonable amount
of space, keeping an informative structure (and abstracting away the
free part $\mathbf{F}$ of the proof). 

This leads to the following definition of structured proof (written
in a formal language), upon which the proof format of our protocol
is based (see Section \ref{subsec:claim-of-proof-format}). A structured
proof of a statement $\mathbf{S}_{*}$ of level $L\geq1$ consists
of a tree with the following structure:
\begin{itemize}
\item The root (`top-level') is the statement $\mathbf{S}_{*}:\mathbf{A}_{*}\implies\mathbf{C}_{*}$
itself, where $\mathbf{A}_{*}$ includes axioms and accepted statements
used to derive the conclusion $\mathbf{C}_{*}$.
\item For each non-leaf (`high-level') statement $\mathbf{S}:\mathbf{A}\implies\mathbf{C}$,
its children $\left(\mathbf{S}_{j}\right)_{j=1,\ldots,k}$ are statements
$\mathbf{A}_{j}\implies\mathbf{C}_{j}$, where $\mathbf{A}_{j}$ is
of the form
\[
\mathbf{A}_{j}=\mathbf{A}\cup\left\{ \mathbf{C}_{i}\text{ for }i\in\mathcal{I}_{j}\right\} \qquad\text{for some }\mathcal{I}_{j}\subset\left\{ 0,1,\ldots,j-1\right\} ,
\]
and where $\mathbf{C}_{k}=\mathbf{C}$. 
\item For each leaf (`low-level') statement $\mathbf{S}:\mathbf{A}\implies\mathbf{C}$,
a machine-verifiable proof is provided. 
\item The tree height (distance between the root and leaves) is at most
$L\geq1$. 
\end{itemize}
\begin{rem}
The idea of structured proof in our paper is very close to structured
proofs suggested by Lamport \cite{lamport-i,lamport-ii}; the difference
is that we make the level more explicit. Ideally, a proof tree should
be well balanced (not too deep, and at the same time with a moderately
large degree, with fairly short statements), and highest levels should
be the most interesting to experienced mathematicians.
\end{rem}

The SPRIG protocol is based on the assumption that valid known mathematical
derivations can be structured with well-balanced trees (in principle,
as the complete tree is as large as a machine-level proof), and that
the existence (or non-existence) of a complete structured proof tree
can be determined with high confidence by a mathematician only knowing
a small subset of the tree (which may depend on the level of information
of the mathematician). 

The verification of a proof in a debate (between a teacher and a student,
or a reviewer and an author \ref{subsec:challenges-in-modern-mathematical-derivations}
below) works largely with the idea of well-balanced tree: a teacher
may ask a student to produce the highest  levels of the tree, to reach
the confidence that the student would know how to provide the rest
of the tree (if given enough time). 

As discussed in Section \ref{subsec:amount-of-details-in-a-proof}
above, determining the relevant amount of details to provide in a
published proof is delicate and somewhat subjective. This paper aims
at explicitly taking this subjectivity into account, by considering
a system of agents with various levels of information and confidence
in their ability to fill in the details of a proof, and proposing
a protocol by which such agents can exchange information.

\subsection{\label{subsec:recent-developments-in-computer-based-proofs}Recent
Developments in Computer-Based Proofs}

The second half of the 20th century saw the explosion, both theoretical
and practical, of computer science. As discussed in Section \ref{subsec:logic-view-of-mathematical-proofs},
the modern notion of formal proof leads to the idea that such proofs
are machine-verifiable, and as a result, we use formal proofs and
machine-level proofs as synonyms. This has led to a desire to see
a formalization of mathematics in computer-checkable terms (see e.g.
the QED Manifesto \cite{anonymous}). At this point, this wish remains
largely unfulfilled, and most modern mathematics has not benefited
from the progress in computer-based proofs. 

In this subsection, we discuss key features of such systems and the
associated challenges. 

\subsubsection{\label{subsec:automated-proofs}Computer-Assisted Proofs}

The idea that computers could (and perhaps should) verify the validity
of proofs goes back at least to Gödel's Lost Letter to Von Neumann
\cite{lipton}. It is extremely natural: in principle, any modern
mathematical derivation can be translated into a sequence of formulae,
which are progressively derived by applying specific rules (which
we will call logic system); specifying the formulae and the applied
derivation rules thus constitutes a computer-checkable proof, sometimes
called proof object (see Section \ref{subsec:computer-proof-systems}
below). 

Computer-assisted derivations have yielded a number of successes in
mathematics, and have been instrumental for the proofs of celebrated
conjectures, which involve dealing with a large number of cases separately:
\begin{itemize}
\item Famously, computer-assisted proofs were instrumental to the first
proof of the four-color theorem \cite{appel-haken,appel-haken-koch};
a fully machine-verified proof was given later \cite{gonthier}. 
\item The Kepler conjecture about sphere packings was established by a machine-checked
proof \cite{hales}
\end{itemize}
In addition to enabling proofs that are too hard for humans to write
and check, computer-assisted proofs are also important in the field
of software verification, where they allow one to guarantee that functions
of a program will behave according to specification.

As a result of their appeal (as trusted elements of knowledge, as
objects that computers can sometimes produce better than humans),
a desire to formalize mathematics has grown over the years. SPRIG
is based on the idea that to convey succinct, informative, and trustable
proofs, entire proofs are not necessary: only a subset that is relevant
to the agents exchanging information is needed. As discussed in Section
\ref{subsec:computer-proof-systems}, many proof systems exist; the
rest of this paper is based, for concreteness, on one of them, which
is particularly readable by mathematicians. 

\subsubsection{\label{subsec:computer-proof-systems}Computer Proof Systems}

A number of computer-based systems such as Mizar, TLA+, Isabelle/Isar,
Coq, Metamath, HOL, or Lean enable and facilitate the writing of formal
proofs. These systems rely on some basic low-level language: proofs
written in this language are called proof objects, and they are what
the computer ultimately checks the validity of. At the same time,
the users of such systems usually work with a higher-level language,
called user-level proof, which yields proofs that are usually much
shorter (yet still very long compared to proofs used by mathematicians
\cite{wiedijk-ii}). 

This article is largely agnostic on the specific choice of computer-system
but could be implemented easily in so-called declarative systems such
as Mizar, Isar or Lean.

A user-level proof consists of a sequence of statements (which include
equivalent of high-level mathematical proofs, with definitions, proof
steps, justifications, sub-statements, cases, etc.) written in a formal
language. The proof is called complete when the system is able to
validate each of the justifications for the steps. While somewhat
tedious to write and to read, formal proof sketches, which are valid
proofs in which a number of steps have been removed, are particularly
easy to understand by mathematicians, as pointed out in \cite{wiedijk}
(see also Section \ref{subsec:modern-mathematical-derivations-in-practice}
above). 

In languages such as Isar, reading and writing simple statements or
definitions (as opposed to writing complete justifications) is relatively
easy. As a result, a mathematician can be expected to be able to determine
if a simple statement deemed to have a short proof does indeed have
one, without having to try to write it; this idea is at the heart
of SPRIG. 

\subsubsection{\label{subsec:interactive-proofs}Interactive Proofs}

We conclude this subsection by discussing the development of a completely
different take on proofs induced by the algorithmic reduction of proofs,
motivated in particular by cryptographic applications: that of interactive
proof protocols. An interactive proof protocol allows a prover to
demonstrate their knowledge of a proof to a verifier, through their
ability to answer the verifier's questions.

The cornerstone of interactive proof checking relies on computational
complexity theory and NP-completeness (see e.g. \cite{widgerson}
for a modern account): for any statement $\mathbf{S}$ in a formal
logic system $\Lambda$, the statement `$\mathbf{S}$ has a proof
of length $\leq L$ in $\Lambda$' can be reduced to the statement
`$\mathbf{G}$ is 3-colorable' for a certain graph $\mathbf{G}$ computable
in polynomial time in terms of $\mathbf{S},L$, and $\Lambda$; proofs
of length $\leq L$ of $\mathbf{S}$ in $\Lambda$ then are in one-to-one
correspondence with 3-colorings of $\mathbf{G}$ (i.e. assignments
of colors to the vertices of $\mathbf{G}$ such that any pair of adjacent
vertices has different colors). Note that we focus on 3-colorings
for concreteness, but any NP-complete problem would do the job as
well. Through the prism of computational complexity, a prover can
demonstrate her knowledge of a proof of length $\leq L$ of a statement
$\mathbf{S}$ by demonstrating her ability to color the corresponding
graph $\mathbf{G}$.

This view of proofs allows in particular for a number of interesting
applications (see e.g. \cite{widgerson}):
\begin{itemize}
\item Probabilistically checkable proofs: the prover may be able to map
her 3-coloring problem to a 3-coloring problem with an amplified gap,
and to use it to convince a skeptic of her knowledge of a proof (of
length $\leq L$) via a limited number of interactions (independent
on the proof length); in essence, if the prover bluffs in her answers
to questions, she will be caught with high probability. 
\item Zero-knowledge interactive proofs: by relying on cryptographic primitives,
the prover may be able to demonstrate her knowledge of a proof (of
length $\leq L$) without divulgating anything about the proof itself.
\end{itemize}
The view of proofs underlying the field of interactive proofs is at
odds with that of theorem proving in mathematics: the proofs in that
world yield little, if any, insight to mathematicians (as discussed
in Section \ref{subsec:proofs-for-explanation} and \ref{subsec:logic-view-of-mathematical-proofs}
above) into why proven theorems are true (and for instance, in the
case of zero-knowledge proofs, the goal is to give zero information
about the proof). 

Still, the idea of establishing trust via a few interactions is very
appealing for proof verification in mathematics. In light of this
question, SPRIG aims at allowing agents to communicate both trust
\textit{and} insight to each other, via a limited number of interactions. 

\subsection{\label{subsec:challenges-in-modern-mathematical-derivations}Challenges
in Modern Mathematical Derivations}

In mathematical practice, the amount of detail needed to assess a
proof's validity is usually decided by a peer-review refereeing process,
whose goal is also to determine how interesting the results and insights
are. Usually, this happens within the context of a publication by
a journal, in which a small number of independent experts assesses
the validity of a proof (i.e. whether enough details are provided
to transmit them the confidence that the proof is correct, in the
sense of Claim \ref{claim:knowing-how-to-prove-a-statement} above).
In the case of higher-profile results, validation also comes from
the larger community, where all experts of the field may discuss the
results, identify weaknesses, and exchange comments with the authors.

In this context of `high-level' proofs, the last decades have seen
a number of developments which created additional challenges, in particular:
\begin{itemize}
\item The inflation in the complexity of (published) mathematical proofs
makes their validation more difficult (see Section \ref{subsec:complexity-of-proof-validation}
below).
\item The boundary conditions of the process (see Section \ref{subsec:boundary-conditions-of-reviewing}). 
\item The alignment of various external incentives with the ones of the
validation process (see Section \ref{subsec:alignment-of-incentives}). 
\end{itemize}
In principle, as suggested in Section \ref{subsec:nature-of-mathematical-derivations}
above, all of these challenges are of a purely practical nature: given
enough competent, reliable and properly incentivized experts, these
challenges would not exist; or, alternatively, if all proofs were
easy to write down in a format verifiable by a computer (as discussed
in Section \ref{subsec:recent-developments-in-computer-based-proofs}),
there would be no need for expert verification. As discussed in the
following subsections, however, these challenges are very real today;
it is the aim of SPRIG protocol allowing to help tackle them. 

\subsubsection{\label{subsec:complexity-of-proof-validation}Complexity of Proof
Validation}

Checking the validity of mathematical derivations is a time-intensive
task, which naturally depends on the length of the result as published,
and on the time it takes a referee to check the details (i.e. explicitly
or implicitly filling in the blanks to convince oneself of its validity).
Over the last 100 years, the complexity of published mathematical
proofs has grown significantly:
\begin{itemize}
\item The average length of mathematical proofs has grown: for instance,
the average length of a paper in Annals of Mathematics in 1950 was
less than 17 pages, but in 2020 it was more than 58. 
\item The papers, in turn, usually rely on larger and larger bodies of work,
and proofs are now rarely self-contained.
\item Some proofs are split into many mathematical articles: for instance,
the classification of the finite simple groups consists of tens of
thousands of pages in several hundred journal articles, published
over a 50-year period. 
\end{itemize}
As a result of this inflation in complexity, the process of validating
a proof has increased in difficulty. Some documented, high-profile,
examples are:
\begin{itemize}
\item The Jacobian Conjecture: it has seen a large number of claims of proof
in the 20th century, which have survived for a number of years, before
being invalidated, and standing as an open problem \cite{wikipedia-jacobian}.
\item Poincaré's Conjecture: after a number of incorrect claimed proofs
were proposed throughout the 20th century, a collection of papers
was published by Perelman in 2002--2003, which led to a validation
by the Clay Institute in 2006 following many debates, including the
publication of a number of papers, some of which filled in details,
of which some claimed to be the first complete proof of the conjecture
\cite{nasar-gruber,szpiro}.
\item Hilbert's 16-th Problem: currently an open problem, for which many
attempted solutions have been proposed, some of which took decades
to be invalidated \cite{ilyashenko}.
\item The ABC Conjecture: a solution has been proposed, which is considered
wrong or incomplete by a significant number of experts, but at the
same time considered correct by a significant number of experts \cite{castelvecchi}.
\item The classification of finite simple groups was announced as completed
in 1983. Yet a number of gaps have been found over the years, which
were filled over the following decades \cite{solomon}. 
\end{itemize}
In a number of high-profile cases, the underlying debates have taken
years to settle. In principle, such debates should not last: it should
be enough to provide a computer-verifiable proof \ref{subsec:recent-developments-in-computer-based-proofs};
in practice, the sheer length of the relevant proofs in machine-level
language makes such a task daunting. As discussed in Section \ref{subsec:boundary-conditions-of-reviewing},
this poses a number of problems in terms of the boundary conditions
of the process. 

SPRIG aims at enabling various agents (including computers) with diverse
areas of expertise to collaborate in the reviewing process and in
the writing of the proofs' details.

\subsubsection{\label{subsec:boundary-conditions-of-reviewing}Boundary Conditions
of the Reviewing Process}

The reviewing process, as performed by a journal or a community as
a whole, involves a number of delicate boundary conditions, which
are usually not formalized:
\begin{itemize}
\item The reviewing process involves matching authors and expert reviewers.
Picking experts may be a difficult task for a journal editor; in the
case of public debates, for the community to decide whom to listen
to requires the build-up of a consensus. 
\item The interaction protocol between authors and reviewers is not formalized.
Should the authors or reviewers not act in good faith or have drastically
views of what a complete proof means, the process will stall:
\begin{itemize}
\item In principle, there is nothing that prevents reviewers from nitpicking
or claiming they don't understand some parts, and hence of deeming
a correct proof incomplete: in some sense, a proof is indeed incomplete
until it is completely written down in a machine-verifiable format. 
\item Dually, an author whose proof is too vague or incomplete (or possibly
void) may keep adding irrelevant details that do not address the heart
of the issue, or claim there is nothing important to add and that
the reviewers are nitpicking on trivialities.
\item Since going down to the machine level is not a feasible option, for
a journal, the editor ends up being the ultimate arbiter; when the
whole community discusses the issue, a consensus forms (or doesn't
form). 
\end{itemize}
\item The dual role of reviewers: they are expected to emit at the same
time, a judgement on the validity and the interest of the result. 
\item Ultimately, there is no specification as to which of the `wrong until
proven correct' (the machine-level proof standard) or `correct until
proven wrong' principles prevail in case of disagreement, and in which
timeframe, i.e. where the burden of proof lies.
\end{itemize}
While the above are theoretical weaknesses of the protocol of the
reviewing process, they are not necessarily problematic in practice
if the various agents work constructively towards aligned goals (such
as uncovering mathematical truths, acting in good faith, etc.). But
this is no longer the case as soon as conflicts of interest exist:
see Section \ref{subsec:alignment-of-incentives}.

In light of the boundary conditions problem, SPRIG allows one to rely
on computer-based proof verification algorithms (as in Section \ref{subsec:recent-developments-in-computer-based-proofs})
as the ultimate arbiters, and to enforce explicit and transparent
time constraints.

\subsubsection{\label{subsec:alignment-of-incentives} Alignment of Incentives}

The functioning of the reviewing process involves a number of agents,
whose identities may or may not be known; for high-profile proofs,
the whole community may end up being involved. As discussed in Section
\ref{subsec:boundary-conditions-of-reviewing}, the boundary conditions
can in principle be the source of problems; this can in particular
be the case if there is misalignment between the goal of thorough
and quick validation and the agents' objectives. 

Arguably, a large part of the incentives underlying the reviewing
process is implicit, rather than explicit (namely: desire to discover
the truth, to be intellectually honest, to participate in the good
functioning of the community, to be respected as an expert, etc.).
However, the agents' strategies may be also influenced by the presence
of various external incentives whose alignment with the reviewing
goals is unclear (e.g. funding, jobs, prizes, recognition). 

In terms of explicit incentives, a number of problems may arise:
\begin{itemize}
\item The reviewing process is rarely explicitly incentivized (in the case
of journals, reviewers are usually anonymous and not compensated),
in particular related to their ability to spot e.g. mistakes; and
there is an asymmetry of incentives: there are usually only negative
consequences for not finding errors and no significant downside to
rejecting a valid proof.
\item For high-profile problems, there is a problematic asymmetry: numerous
amateurs may see a lot of upside in submitting (mostly incorrect)
proofs of famous conjectures, while at the same time fewer experts
are available to spend their precious time on finding errors in these
proofs (with no upside). 
\item Authors may be incentivized to publish vague, incomplete proofs to
claim precedence over other authors. 
\item Independent experts are hard to find in very specialized fields, and
the incentives to disclose conflicts of interest are limited. If the
reviewers are competing with the authors (or conversely are interested
in seeing them succeed) they may stall the reviewing process (or conversely
be too lenient), as discussed in Section \ref{subsec:boundary-conditions-of-reviewing}
above.
\item In the case where proofs involve security issues (such as in cryptographic
contexts), there may be additional problems, with experts having possibly
incentives to keep discovered mistakes to themselves, as exploiting
them may be worth money. 
\end{itemize}
The above problems are illustrated by a number of high-profile examples
\cite{nasar-gruber,castelvecchi}, and even when aware of the existence
of alignment problems, it is hard for external agents to identify
in which instances of the above problems a situation falls \cite{castelvecchi,nasar-gruber}. 

In light of the alignment problem, the goal of SPRIG is to mediate
multi-agent interactions with explicit incentives, designed in such
a way as to align the agents' objectives with the ones of the reviewing
process. 

\subsection{\label{subsec:markets-information-and-games}Markets, Information,
and Games }

As discussed in Section \ref{subsec:challenges-in-modern-mathematical-derivations},
the validation process of proofs as performed by mathematicians involves
a number of agents, with different levels of information, interacting
in a variety of manners. These include publishing proofs, detecting
issues in papers, asking for more details, and providing them. 

Similarly, SPRIG invites repeated interactions between members of
the mathematical community with potentially variable levels of information
and degree of involvement. 

In both the current validation process and the SPRIG protocol, understanding
the set of involved agents and their interactions as an economic system
is of paramount importance. Indeed, the raw data of the protocol outcome
(say: number of refereeing rounds, final status: accepted/rejected)
is in general insufficient to determine precisely what credence the
community should have in the validity/invalidity of a proof. 

Understanding the motives, incentives, and beliefs of the relevant
agents allows one to assess what information we can actually extract
from a protocol outcome. The simplest example is perhaps the case
of an accepted paper, the author of which sits in the editorial board
of the publishing review. All other things being equal, and because
of an obvious incentive problem, it seems rational to (at least slightly)
decrease the credence in the validity of the published results. 

Section \ref{subsec:alignment-of-incentives} lists a number of other
incentive issues. For completeness, we now briefly review the core
concepts of the modern microeconomics toolbox and the main economic
theories pertaining to the discussion above. A deeper game-theoretic
treatment of the validation process is provided in Section \ref{sec:a-simplified-equilibrium-analysis}. 

\subsubsection{\label{subsec:agents-bayesian-views-markets}Agents, Bayesian Views,
and Markets}

The raison d'être of incentives problems is that each agent follows
their own agenda, being driven by specific motives or preferences.
These preferences can be represented by a utility function, a notion
that traces back at least to Bentham and J. S. Mill and has become
commonplace since the emergence of neoclassical economics. Von Neumann
and Morgenstern \cite{von-neumann-morgensten} established conditions
on agents\textquoteright{} preferences such that those can be ordered
by an expected utility calculation, providing the economics profession
with a key tool for dealing with decision-making under uncertainty.
However, other concepts were needed in order to think and make predictions
about the way economic agents (inter)act. 

In particular, a proper microeconomic treatment of incentive problems
and their consequences was virtually impossible before the advent
of two intellectual revolutions. 
\begin{itemize}
\item The first one, largely initiated by Nash \cite{nash}, is game theory:
it provided economics with a much-needed tool to model situations
where uncertainty arises from one's imperfect knowledge, not about
the state of Nature, but rather from other agents\textquoteright{}
actions. 
\item The second one, initiated by Muth \cite{muth} and later supported,
consolidated, and popularized by Lucas, is rational expectations:
rational agents not only maximize their own utility but have the same
knowledge as the economic modeler and are able to correctly compute
the model\textquoteright s outcome. This requires making assumptions
about other agents' behavior, but at the same time, allows them to
predict this behavior (given that agents will take utility-maximizing
decisions). Rationality requires that the predictions coincide with
the assumptions. Key to the rational anticipation process is the ability
to correctly compute expectations; in particular, rational agents
are Bayesian updaters. 
\end{itemize}
Several contexts, including the issues we investigate in the current
paper, called for an extension of these tools to the case of asymmetric
information: situations in which agents operate under different information
sets and can extract some information from others' actions. Akerlof
\cite{akerlof}, Spence \cite{spence}, and the various works of Stiglitz
and his co-authors in the seventies closed this gap. Contributions
such as \cite{cho-kreps} and \cite{fudenberg-tirole} provided refinements
of the equilibrium concept for strategic interactions under asymmetric
information. In this literature, agents have a `type', i.e. a characteristic
that is not directly observable but partially or fully inferred given
a history of actions. As we shall see, in SPRIG, this type is the
subjective probability that a proposed proof can be unrolled up to
machine language level (itself a function of variables such as personal
skill, amount of work, and carefulness which are not fully observable
by outsiders). 

Townsend's model \cite{townsend} with costly state verification initiated
a large literature on optimal contracting under asymmetric information.
In SPRIG, the goal is, indeed, to estimate the status (i.e. valid/invalid)
of a proof. But the very design of our validation process implies
that verification (while potentially costly in terms of time and intellectual
energy) might in fact be beneficial to the verifiers, who can collect
bounties. 

The system formed by SPRIG and its users can be seen as a market for
proofs, although not quite in the sense of a stock market. However,
just as the (semi-strong) efficient market hypothesis \cite{fama}
stipulates that the utility-maximizing behavior of agents will lead
stock prices to reflect any available public information, we expect
that with properly chosen parameters, SPRIG will aggregate information
and eventually disclose the actual status of a proof. 

The dark markets reviewed by Duffie \cite{duffie}, while different
from our market in a variety of regards, also share similarities with
SPRIG. These markets are dealer networks in which connected agents
conduct bilateral negotiations to trade financial assets (there is
no `market price'). Each transaction reveals part of the private information
that dealers have, and therefore one can expect information to `percolate'
through the network. 

While scoring rules can be used to aggregate the credences of various
agents on statements (see e.g. \cite{hanson-i,hanson-ii}), SPRIG
features two additional key characteristics. First, it provides a
built-in termination date at which the validity of the proof/question
is decided and thus ``bets'' can be settled. Second, the dynamic
verification process generates explicit information about the strengths
and weaknesses of a claim. In fact, SPRIG can be viewed as a multi-round
security game (see e.g. \cite{blocki-et-al}), in which agents `debate'
the validity of claims (\cite{irving-christiano-amodei}), with automatically
set boundary conditions.

\subsubsection{\label{subsec:economic-markets-and-mathematical-truth}Economic Markets
and Mathematical Truth }

The idea of agents with a Bayesian view on the truth of mathematical
questions dates back at least to the works of Solomonoff \cite{solomonoff}.
Recent works on systems of such agents interacting through a market
\cite{garrabrant-benson-tilsen-critch-soares-taylor-logical-induction}
have shed light on how such systems may be viewed as (decentralized)
algorithms that estimate and refine probabilities of truths for mathematical
statements. \cite{garrabrant-benson-tilsen-critch-soares-taylor-logical-induction}'s
algorithm, a \textit{logical inductor,} dynamically assigns probabilities
to mathematical statements and the belief system thus produced is
shown to be consistent asymptotically. This consistency as well as
other desirable properties of their algorithm derives from a \textit{logical
induction criterion, }which is essentially a no-arbitrage condition
on a market defined as follows. Each mathematical statement $\phi$
is associated with a derivative that pays \$1 if is true and zero
otherwise; the market price of this derivative can then be seen as
the current belief about the truth of $\phi$. A trader can observe
the history of prices up to time $t-1$, make some computations of
their own, and post market orders at time $t$, adjusting their portfolio
of derivatives written on statements $\phi_{1},...,\phi_{n(t)}.$
Their trading strategy is adapted, i.e. a function of past prices.
Importantly, this ``past'' includes time $t$: naturally, the demand
function depends on the price that will eventually be set. A \textit{market
maker }(a subroutine of the logical inductor) then sets prices in
such a way that the demand of the trader is (approximately) zero for
all derivatives. By construction then, a ``fair price'' obtains,
which captures the probability of the statements underlying the derivatives.

While appealing in many aspects, these ideas remain unfortunately
extremely theoretical. In the words of \cite{garrabrant-benson-tilsen-critch-soares-taylor-logical-induction},
`logical inductors are not intended for practical use.'
\begin{itemize}
\item The required computation times and spaces are unreasonably large.
\item There is an important distinction between being true and being provable:
the latter is the focus of mathematics research, and accumulating
evidence for the veracity of result may not result in any progress
towards proving it (for instance, verifying Goldbach's conjecture
up to a large $N$ may increase one's trust in the truth of the statement,
but not bring any insight into how to prove it).
\item There is no focus on the amount of insight associated with the agents'
discoveries.
\item There is no obvious way to incorporate agents seeking information
about specific statements, i.e. to shift the attention of the market
towards a set of problems currently of interest to these agents. 
\end{itemize}
SPRIG leverages on the view that the combination of Bayesian updating
and individual profit-seeking behaviour leads the market to reveal
information about fundamentals (in our context: the validity of mathematical
claims). However, rather than aiming at constructing an ``exhaustive
encyclopedia'' of mathematical propositions, our framework incentivizes
agents to inject (or induce injection of) information about \textit{specific
}statements, which are relevant for the community either because they
are mathematically interesting or because they correspond to critical
points in a proof. Furthermore, our focus is on the effective provability
of statements rather than on the credence that they are, in an abstract
way, true. That is, SPRIG not only invites its users to focus on important
mathematical statements, but also induces them to discuss/prove/refute
those in a way that provides intuition and insights about \textit{why
}they are true or false.

\subsection{\label{subsec:blockchain}Blockchain and Related Technologies}

The last decade saw the emergence of fully decentralized computing
systems, in particular in the context of public databases made of
public, immutable records, called blockchains. These systems, running
on a decentralized network of computers (typically connected to the
internet), have grown out of the desire to build transparent, trustless
consensus systems, relying on no central authority, or specific machine,
following a secure, time-stamped, and easily auditable behavior. 

While initially restricted to the context of the storage of digital
assets into accounts (such as for Bitcoin), blockchains have grown
in terms of applications and features. The introduction of smart contracts,
allowing the blockchain to perform complex operations and transactions
conditioned on taking various inputs, has opened new possibilities.
A particularly prominent development is the creation of inexpensive,
open, efficient, and trusted general-purpose markets. 

SPRIG provides a way to construct such markets aimed at decentralized,
public proof verification: its very design makes it perfectly suitable
for running on the blockchain. 

In this subsection, we review a number of key principles and mechanisms
of blockchain technologies, upon which SPRIG relies. 

\subsubsection{\label{subsec:distributed-systems-and-blockchain-basics}Distributed
Systems }

A blockchain is a certain type of distributed system. A distributed
system instance consists of execution instances of programs (often
called clients) running on a network of computers (often called nodes),
which communicate via a specified protocol. The protocol prescribes
the communications that the nodes should emit and receive. Early examples
of distributed systems across the internet include peer-to-peer file-sharing
networks, in which a node has a number of files available for other
nodes to request. As a whole, a distributed system can be viewed as
an execution instance of a program: a peer-to-peer file network can,
for instance, be viewed as a single database, emerging from the various
nodes.

Unlike regular program instances, distributed systems cannot be viewed
as running on any particular node, while emerging from the nodes;
this makes distributed systems more tolerant to localized failures
in the system. 

A distributed system is called (fully) decentralized when there is
no principal node coordinating the network. A fully decentralized
system instance is hence a form of consensus, emerging from the execution
of the clients on the nodes: the choice of the nodes to adhere to
the protocol defining it, by running client programs that respect
the protocol, is what gives life to the instance. 
\begin{rem}
In some regards, the mathematical activity (on planet Earth) can be
viewed as a decentralized system: mathematicians are agents who choose
to adhere to a protocol of logic rules, and whose work should be accepted
by other agents as long as it follows the protocol, and there is no
central authority entitled to decide what is correct mathematics or
who should be able to publish mathematics. Of course, an important
difference is that the protocol's rules are not completely specified
in practice and that most nodes are not computers, but humans. 
\end{rem}

Distributed systems have grown in importance over the last decades,
due in large part to the development of the internet. As discussed
in Section \ref{subsec:blockchain-and-cryptocurrency-basics} below,
a special of distributed systems have risen to prominence in the last
decade: blockchains. SPRIG is designed to run on such systems. 

\subsubsection{\label{subsec:blockchain-and-cryptocurrency-basics}Blockchain and
Cryptocurrency Basics}

A blockchain is a specific type of distributed system, where the underlying
database is made of a chain of immutable pieces of data called blocks,
which grows over time (new blocks are appended as the instance evolves).
This feature of blockchain allows for a consensus about time-stamped
data to develop.

For instance, Bitcoin is a blockchain instance consisting of blocks
describing transactions between accounts (represented by cryptographic
public keys), where each node possesses an entire copy of the blockchain;
about every 10 minutes a new block is added, which contains the transactions
that have been validated in that time interval. The Bitcoin software
is designed in such a way that the Bitcoin blockchain can be viewed
as a public ledger of amounts of currency units (called bitcoins)
owned by each account, and where validated transactions move bitcoins
between accounts. More generally, blockchains can be used to implement
crypto-currencies, and allow users to trade virtual assets, called
tokens. In such systems, accounts are also represented by a cryptographic
public key, and transactions from an account submitted to the network
are accepted if they are signed by the private key associated with
the account and the funds are still available.

A key feature of blockchains such as Bitcoin is that they are based
on a publicly available protocol (typically with an open-source reference
client implementation); as a result, the laws governing the system
are transparent (the `code is law' motto is sometimes used to describe
such systems), and it informs the adherence of various nodes and stakeholders
to the system. 

The implementation of most blockchain systems usually relies on the
internet's infrastructure, and on cryptographic primitives to ensure
both the integrity of the data and the authenticity of the transactions.
The guarantee of data integrity provided by blockchains is at the
heart of such protocols and of the interactions of the agents, who
can thus use blockchains as a medium for information exchange and
a trusted hub. As a result, a number of blockchains (such as Bitcoin,
Ethereum, etc.) have emerged as focal points (or Schelling points)
for a growing population of users \cite{breitman}. Blockchains thus
play the role of trusted platforms for agents interested in exchanging
information and digital assets. 

The fact that the functions of blockchains are executed by code running
on nodes has led to many extensions beyond the original application
of token ledgers. In particular, the advent of smart contracts, discussed
in Section \ref{subsec:smart-contracts}, has opened many new possibilities,
including the system proposed in this paper.

\subsubsection{\label{subsec:smart-contracts}Smart Contracts}

Smart contracts emerged naturally from the desire to leverage the
power of blockchains as decentralized computing platforms to implement
programs to perform automated transactions: for e.g. a cryptocurrency,
one would like to be able to run a program that automatically moves
some asset from an account to another, at a time when certain conditions
are met. Such programs run on the blockchain (i.e. are executed by
the nodes of the blockchain) and update its state (by contributing
new blocks); once running on the blockchain, such a program can be
viewed as a contract, guaranteeing the execution of certain transactions
if pre-specified conditions are met, hence the name smart contract. 

The behavior of a smart contract is specified by its code, together
with inputs from the blockchain: for instance, a smart contract may
be the recipient of a transaction from another agent on the blockchain,
or it may act according to a signed information source (for instance,
a trusted information feed from the physical world, known as an `oracle').
As a simple example, one can imagine a smart contract implementing
a chess competition with automated rewards distribution: players submit
their (cryptographically signed) moves to the blockchain, and the
outcome is either determined by one party resigning, both parties
agreeing to a draw, or by reaching a position computed as terminal
by the smart contract.

A number of blockchains have come up with developed infrastructures
for smart contracts, including Ethereum, Tezos, Algorand, Avalanche,
... Each of these platforms allows for the writing of smart contracts
in fairly rich (sometimes Turing-complete, as for Ethereum) high-level
languages, and their execution on the blockchain against a fee (depending
on the complexity of the operations, and payable in the blockchain's
cryptocurrency token). 

The smart contract infrastructure has enabled the construction of
numerous decentralized platforms, in particular in decentralized finance
(as discussed in Section \ref{subsec:decentralized-markets} below):
exchanges, betting markets (relying on information feeds), automated
market makers, stable coins, etc. can now be run as smart contracts.
Such platforms allow for applications that previously needed to rely
on the good behavior of expensive (and corruptible) trusted third-parties
to enforce the execution of the contracts. Their transparency allows
for a detailed risk analysis (see e.g. \cite{angeris-chitra,angeris-kao-chiang-noyes-chitra}).

Smart contracts platforms can be used to build trusted interactions
and consensus, to establish transparent, reliable, and efficient institutions.
SPRIG is designed to run on smart contracts (without relying on external
oracles), allowing it to aim for such goals, and to be a building
block for further decentralized applications (see Section \ref{subsec:derivatives-markets}
below). 

\subsubsection{\label{subsec:decentralized-markets}Decentralized Markets }

We now briefly discuss advances in the applications of smart contracts
to decentralized markets, which have experienced a surge of interest
in recent years; SPRIG can be viewed as a form of decentralized market
for mathematical derivations. 

One of the first interesting applications of smart contracts is that
of decentralized betting markets (e.g. \cite{augur}). In the simplest
form of decentralized betting smart contracts, two parties decide
to bet at given odds on the outcomes of some future event. To do so,
they create a smart contract to which they send their bets (i.e. the
contract acts as an escrow) and that can look up a pre-specified,
commonly trusted data feed, aka the `oracle'; when the event has happened,
the smart contract determines the outcome from the data feed and sends
the wagered funds to the winner.

For certain betting markets, no external oracle is even needed since
the relevant event occurs directly on the blockchain. The power of
the blockchain to move assets based on the results of computations
has attracted some attention in the mathematical community. Indeed,
in principle, checking a proof (in a machine-verifiable format) can
be performed by a smart contract (provided that the platform's language
is expressive enough, and enough computing resources are available):
in particular, the projects Qeditas \cite{white} and Mathcoin \cite{su}
are based on such ideas. Both aim at constructing a ledger of agreed-upon
mathematical statements where the prospect of financial rewards induces
agents to contribute their knowledge.
\begin{description}
\item [{Mathcoin}] In the Mathcoin project proposed by \cite{su}, the
ledger is constructed using a bottom-up approach: agents successively
append statements that are logical consequences of the previous ones,
starting from the Zermelo-Fraenkel axioms. These statements are appended
if the agents provide a valid proof at the machine level. Connected
to this growing ledger, a market allows the agents to bet on yet unproven
propositions. A user in possession of a result potentially relevant
to the proof of such a proposition can buy a derivative paying conditional
on the validity of the proposition, then post their result on the
ledger. They should subsequently benefit from an appreciation of the
price of the derivative. Hence agents are incentivized to contribute
their knowledge. However, the Mathcoin protocol and SPRIG differ in
several important respects. First, the former produces a ledger of
machine-level claims only, which is likely to be impractical for the
scientific community as a whole; the machine-language requirement
also presumably implies that growing the ledger will be a slow and
cumbersome task. By contrast, we use machine-language expansion only
as a boundary condition and expect SPRIG to produce concise, human-tailored
proofs. Second, its pricing function exposes Mathcoin to an attack
where agents are incentivized to post trivial claims. The associated
token is initially priced at $0.5$. The claimer can then immediately
post a proof and collect $1$, as their claim was proven. This attack
pollutes the blockchain and more importantly drains the public fund,
which is intended to reward agents who successfully bet on substantive
propositions, effectively rendering the system unusable. While \cite{su}
mentions this attack, no satisfactory fix is provided.
\item [{Qeditas}] The Qeditas system of \cite{white} also suggests the
construction of a ledger of propositions. There, agents are incentivized
to append a result (written in machine language) either to collect
bounties from a foundation or an individual interested in the result
or because they expect other agents to need it to prove something
else in the future and hence to buy its `rights'. As for Mathcoin,
the bottom-up approach combined with the requirement for complete
machine-level proofs implies usability and practicability issues.
\end{description}
To sum up, while projects such as Qeditas or Mathcoin are promising
endeavours, their functioning seems at odds with the way mathematicians
work, as discussed in Sections \ref{subsec:logic-view-of-mathematical-proofs}
and \ref{subsec:modern-mathematical-derivations-in-practice}. SPRIG
creates a decentralized market for mathematical derivations, which
allows one to avoid evaluating unneeded regions of the proof, while
still relying on the ability of smart contracts to arbitrate, in case
of disagreement, mathematical truths. 

\subsection{Outline\label{subsec:outline}}

As discussed in the previous subsections, mathematical proofs aim
at eliciting trust into the validity of statements and transmitting
insight into their justifications. While trust relies on the confidence
that a machine-verifiable proof could be produced if needed, machine-verifiable
proofs are difficult to produce and convey little insight on a per-line
basis. As a result, proofs are made at a high, informal level in practice,
and are much more concise than their machine-level counterparts, omitting
many details. The production and verification of such high-level proofs
thus represent a challenge: various agents with various levels of
information may disagree on what constitutes a complete and valid
proof. As a result, while a proof's validity is an objective question
in principle, in practice it becomes a complex multi-agent problem,
where incentive alignment problems may arise. 

In this paper, we introduce the SPRIG protocol, which aims at enabling
proof submission and verification in a decentralized manner, allowing
agents to participate with their various degrees of information, and
to be incentivized to behave honestly. It is designed to run on a
blockchain, allowing a smart contract to handle the distribution of
stakes and bounties, and to serve as an arbiter of debates without
relying on any trusted institution.

More precisely, the structure of the following sections is as follows.
\begin{itemize}
\item In Section (\ref{sec:the-protocol}), the SPRIG protocol is presented.
\begin{itemize}
\item In Section \ref{subsec:prologue}, the ideas leading to SPRIG are
introduced, starting with a simple game between a claimer and a skeptic
(Section \ref{subsec:claimer-and-skeptic-debate}), and introducing
variants one by one (Sections \ref{subsec:incentives-claimer-stakes-skeptic-bounties},
\ref{subsec:many-claimers-and-skeptics}, and \ref{subsec:question-as-root-of-the-process}).
\item In Section \ref{subsec:claim-of-proof-format}, the basic version
of SPRIG is described in detail: it is based on a hierchical proof
format, called Claim of Proof Format (Section \ref{subsec:claim-of-proof-format}),
and a recursive structure of nested claims and questions (Sections
\ref{subsec:top-down-informal-view} and \ref{subsec:formal-description}).
\item In Section \ref{subsec:illustrations-of-the-protocol}, the SPRIG
protocol is illustrated, through interactions mediated by it, in a
number of cases.
\item In Section \ref{subsec:variants-and-extensions}, a number of variants
of the basic version of SPRIG are presented, and their merits are
discussed. 
\item In Section \ref{subsec:blockchain-implementation}, various aspects
pertaining to the blockchain implementation of SPRIG are discussed. 
\end{itemize}
\item In Section \ref{sec:informal-game-theoretic-discussion}, a game-theoretic
perspective on SPRIG is introduced, presenting informally the effects
of the incentives structure on the agents' interactions, and the protocol's
resilience to attacks.
\begin{itemize}
\item In Section \ref{subsec:strategic-interactions-and-protocol-outcome},
the strategic interaction between claimers and skeptics, and the results
on the proofs constructed in this interaction are discussed. 
\item In Section \ref{subsec:robustness-properties}, the robustness properties
of SPRIG against various types of attackers are discussed. 
\end{itemize}
\item In Section \ref{sec:a-simplified-equilibrium-analysis}, an in-depth
quantitative analysis of a simplified model of SPRIG is presented. 
\begin{itemize}
\item In Section \ref{subsec:model-setup}, the simplified model is introduced,
which consists of a two-player game of depth $2$.
\item In Section \ref{subsec:model-solution}, the solution of the model
is presented, with a description of the minima.
\item In Section \ref{subsec:results}, key questions about the reliability
of SPRIG are answered in terms of the model's solution.
\item In Section \ref{subsec:discussion-of-the-models-analysis}, the dynamics
and robustness of SPRIG are discussed in light of the analysis of
the simplified model.
\end{itemize}
\item In Section \ref{sec:applications-and-outlooks}, a number of applications
of SPRIG and outlook for future research are discussed: 
\begin{itemize}
\item In Sections \ref{subsec:theorem-verification}, \ref{subsec:bounty-for-open-problem},
\ref{subsec:security-proof-certification}, a number of direct applications
of SPRIG to concrete verification situations are outlined: for theorem
proof verification, for the creation of mathematical challenges with
bounties, and for the elicitations of decentralized security audits.
\item In Sections \ref{subsec:automated-theorem-proving} and \ref{subsec:derivatives-markets},
possible uses of SPRIG as a platform for new applications are discussed,
in particular for the development of automated theorem proving and
derivatives markets, allowing agents to inject various types of information.
\item In Section \ref{subsec:beyond-mathematical-reasoning}, a number of
other applications, relying on external oracles, are proposed. 
\end{itemize}
\end{itemize}

\section{\label{sec:the-protocol}The SPRIG Protocol}

In this section, we describe the protocol at the heart of the present
paper, which allows one to construct and incentivize a debate between
claimers (provers) and skeptics to determine the validity of a high-level,
declarative mathematical derivation: SPRIG, short for Smart Proofs
via Recursive Information Gathering. 
\begin{itemize}
\item In Section \ref{subsec:prologue}, the key ideas of the protocol are
progressively introduced. 
\item In Section \ref{subsec:claim-of-proof-format}, the claim of proof
format upon which SPRIG is based is introduced.
\item In Sections \ref{subsec:top-down-informal-view}, \ref{subsec:formal-description}
and \ref{subsec:illustrations-of-the-protocol}, SPRIG is presented
in detail: first, via an informal top-down view, then via a formal
bottom-up definition, and finally via illustrative examples.
\item In Section \ref{subsec:variants-and-extensions}, a number of natural
variants and extensions of the basic SPRIG protocol are proposed.
\item In Section \ref{subsec:blockchain-implementation}, a number of aspects
of the blockchain implementation of SPRIG are discussed. 
\end{itemize}

\subsection{\label{subsec:prologue}Prologue}

In this prologue, we proceed step by step to introduce SPRIG: we start
by introducing a proof-checking protocol with two agents, a claimer
and a skeptic, debating the provability of a statement, and using
machine-level verification as the ultimate arbiter. 

\subsubsection{\label{subsec:claimer-and-skeptic-debate}Claimer and Skeptic Debate}

We first introduce the Claimer-Skeptic debate as a simple process
between two agents, called Claimer (pronoun: she) and Skeptic (pronoun:
he):
\begin{itemize}
\item Claimer claims to have proven a theorem in the sense of modern mathematical
proofs (\ref{subsec:modern-mathematical-derivations-in-practice}):
she has a high-level proof sketch of the theorem (a collection of
statements claimed to break down the difficulty of the theorem into
smaller pieces), and feels confident she could fill in the details,
down to machine level if needed (i.e. she can provide a sequence of
statements which follow from each other by the application of elementary
rules of logic); at the same time, she cannot or does not want to
provide all the details down to the machine level, because the proof
would be impractically long.
\item Skeptic sees what Claimer shows him. For any statement shown by Claimer,
Skeptic has beliefs about the probability that Claimer could indeed,
if pressed, provide the details.
\end{itemize}
For the game to allow for a verification of the proof, Claimer and
Skeptic use the following protocol:
\begin{itemize}
\item Skeptic may ask for more detail on any proof statement shown by  Claimer
that is of higher level than machine-level detail. 
\item Skeptic may invalidate the proof by revealing a mistake in the machine-level
proof details.
\item Claimer cannot indefinitely propose high-level proofs: if pressed
to give details down a certain number of levels (say 9), she must
reach machine-level details, or her proof is considered invalid.
\item Claimer has explicit bounds on the size of the proof sketches and
of the machine-level details.
\item After Claimer has published a proof sketch, Skeptic has a limited
(fixed in advance) amount of time to request details, and after Skeptic
has asked for details, Claimer has a limited (fixed in advance) amount
of time to provide details on a proof statement.
\end{itemize}
Claimer and Skeptic play a role that is somewhat akin to that of an
author and a reviewer; in practice, Skeptic can just be Claimer's
critical thinking, which probes for possible weaknesses of Claimer's
proof, to assess Claimer's confidence in her proof being indeed complete. 

The assessment made by this protocol is whether Claimer can provide
details in the proofs quickly enough: there is indeed a limit in terms
of how much into the details she can go (to avoid a case where Claimer
would just state tautologies, instead of giving proofs). 

A central weakness of the above protocol arises if there is no alignment
between Claimer and Skeptic: Skeptic may start bombing Claimer with
useless questions, or conversely not ask any question at all. In order
to prevent this, incentives can be put in place, as explained in Section
\ref{subsec:incentives-claimer-stakes-skeptic-bounties} below.

\subsubsection{\label{subsec:incentives-claimer-stakes-skeptic-bounties}Incentives:
Claimer's Stakes and Skeptic's Bounties}

In order to avoid the alignment problem in the debate between Claimer
and Skeptic introduced in Section \ref{subsec:claimer-and-skeptic-debate},
incentives can be added:
\begin{itemize}
\item Claimer must put a stake with the claim: this stake encourages Skeptic
to ask questions.
\item Skeptic must pay a bounty to ask a question: this prevents Skeptic
from bombing Claimer with useless questions.
\item In case Claimer can answer a question from Skeptic, she gets Skeptic's
bounty: this compensates her for the work required to answer.
\item In case Claimer cannot correctly answer a question from Skeptic, Skeptic
gets Claimer's stake. 
\end{itemize}
Setting the parameters correctly will incentivize Claimer and Skeptic
to do their work: Skeptic will only ask questions about the places
where he feels there is a a reasonably good chance Claimer cannot
fill in the details; conversely, he will not ask about obvious points
in the proof. This selective revealing of the proof may be useful
to an external observer as well: the details that are revealed are
only the interesting, non-trivial ones. At each stage, one side may
disagree with the other; the first one to stop debating loses, unless
we reach the machine level, where it is Claimer's burden to prove
her statement in machine-level language (otherwise she loses). 

In this sense, the machine level serves as the ultimate arbiter of
who is right; interestingly, if the incentives are set correctly,
a debate between a rational claimer and a rational skeptic would probably
end before reaching the machine level (as in a game between chess
masters where a checkmate position is almost never reached: the losing
side will resign beforehand).

\subsubsection{\label{subsec:many-claimers-and-skeptics}Many Claimers and Skeptics }

In Sections \ref{subsec:claimer-and-skeptic-debate} and \ref{subsec:incentives-claimer-stakes-skeptic-bounties},
we had a debate between only two agents (Claimer and Skeptic). With
good incentives, the roles of Claimer and Skeptic can in fact be completely
decentralized: we will have a market of agents, claimers and skeptics
(where an agent can play both roles at various levels). The skeptics
can ask details about any published proof sketch (by paying an upfront
bounty, being the first to ask, and doing so within time limits);
conversely, the claimers can propose proof sketches to any unanswered
question (by paying an upfront stake and doing so within time limits).
The claimer's stakes are actually split in two: an `up' stake and
`down' stake: if the claim of proof ends up being invalidated, the
`up' stake goes to the question that the claim of proof was trying
to answer, and the `down' stake to the question that first invalidated
the claim. 

This structure allows various agents to perform in various capacities:
for instance, agents with a good high-level vision can propose high-level
proof sketches, while agents who are more comfortable with low-level
details will provide proofs of sub-sub-claims, etc. Again, if the
bounties and stakes and times are set well, each agent will inject
their own information into the system by either proposing claim of
proofs and questions reflecting their beliefs, in a fully decentralized
way. 

\subsubsection{\label{subsec:question-as-root-of-the-process}Question as Root of
the Process}

A small variant of the protocol can be introduced in the case where
`we start with Skeptic': Skeptic may start by putting a bounty (for
instance, because he is interested in sponsoring research about a
question he cares about), to which claimers may propose proof sketches
(paying an `up' and `down' stake upfront). 

With this scheme, claimers should be able to submit several proofs
for a statement, while compensating skeptics who may debunk wrong
proofs. The rest of the process is completely symmetrical. 

This concludes the prelude part of the SPRIG protocol description.
In Section \ref{subsec:protocol-description}, a precise formalization
of the SPRIG is detailed. In Section \ref{subsec:variants-and-extensions},
a number of variants and extensions are presented. In Section \ref{subsec:blockchain-implementation},
questions associated with the blockchain implementation are discussed. 

\subsection{\label{subsec:protocol-description}SPRIG Protocol Description}

Building upon the insights of the previous sections, we now formalize
the Smart Proof by Recursive Information Gathering (SPRIG) protocol.
At the root of SPRIG is either a question or a claim of proof; the
protocol then builds a tree starting from the root, with questions
following claims of proof and vice versa. All the questions and claims
of proof consist of statements written in a formal mathematical language,
leaving no room for ambiguity (see Section \ref{subsec:claim-of-proof-format}). 

\subsubsection{\label{subsec:claim-of-proof-format}Claim of Proof Format }

SPRIG is based on the communication of unambiguous mathematical statements,
written in a formal proof language. The description we give here is
agnostic of the specific formal system; our format description can
be implemented using a declarative proof language such as Mizar, Isar,
or Lean. 

The format we describe is based on hierarchical proofs. A complete
machine-verifiable proof of a statement is a proof that can be structured
as a tree with nodes corresponding to statements, and with leaves
corresponding to machine-verifiable statements. SPRIG allows agents
to query and provide a subtree of the complete proof tree that is
large enough to reach a consensus about whether or not the tree can
be completed into a complete tree (with given size and time constraints),
as discussed in Sections \ref{subsec:top-down-informal-view} and
\ref{subsec:formal-description} below.

Recalling the definition of Section \ref{subsec:structured-proofs},
and setting aside the question of definitions for a moment (this will
be discussed in Definition below \ref{def:claim-of-proof-format}
below), the structured proof format of level $L\geq1$ is that of
a tree with the following structure: 
\begin{itemize}
\item The root (`top-level') is the statement $\mathbf{S}_{*}:\mathbf{A}_{*}\implies\mathbf{C}_{*}$
itself, where $\mathbf{A}_{*}$ includes axioms and accepted statements
used to derive the conclusion $\mathbf{C}_{*}$.
\item For each non-leaf (`high-level') statement $\mathbf{S}:\mathbf{A}\implies\mathbf{C}$,
its children $\left(\mathbf{S}_{j}\right)_{j=1,\ldots,k}$ are statements
$\mathbf{A}_{j}\implies\mathbf{C}_{j}$, where $\mathbf{A}_{j}$ is
of the form 
\[
\mathbf{A}_{j}=\mathbf{A}\cup\left\{ \mathbf{C}_{i}\text{ for }i\in\mathcal{I}_{j}\right\} \qquad\text{for some }\mathcal{I}_{j}\subset\left\{ 1,\ldots,j-1\right\} ,
\]
 where $\mathbf{C}_{k}=\mathbf{C}$. 
\item For each leaf (`low-level') statement $\mathbf{S}:\mathbf{A}\implies\mathbf{C}$,
a machine-verifiable proof is provided. 
\item The tree height (distance between the root and leaves) is at most
$L\geq1$. 
\end{itemize}
\begin{rem}
In our framework, an assumption $\mathbf{A}$ may include the introduction
of notation (e.g. `let $x$ be such that ...'); some mechanism for
the resolution of overloaded symbols is naturally needed (but not
discussed here, being an implementation detail). 
\end{rem}

\begin{rem}
Theorems are often explicitly of the form $\mathbf{T}:\alpha\implies\gamma$
(e.g. we could have $\alpha$ corresponding to `$f$ is a holomorphic
function on $\mathbb{C}$' and $\gamma$ corresponding to `$f$ has
a convergent power series expansion on $\mathbb{C}$'). In such a
case, we could write the statement with $\gamma=\mathbf{C}_{*}$,
and $\mathbf{A}_{*}$ would include $\alpha$, as well as the list
of axioms and assumed results used to derive $\gamma$. 
\end{rem}

\begin{rem}
Various formats of proof fit in this framework, including proofs by
contradictions, etc. See Section \ref{sec:appendix-claim-of-proof-examples}
in the Appendix for examples. 
\end{rem}

To make the writing of statements effective, definitions are needed,
which allow one to reserve notation to refer to properties of objects. 
\begin{defn}
\label{def:collection-of-definitions}A collection of definitions
$\mathbf{D}$ introduces symbols specifying properties of objects,
written in formal language, and specifies references from which other
definitions can be imported. 
\end{defn}

\begin{rem}
For instance, a definition could be `is-group(G, op)' which would
imply that $G$ is indeed a set, op is indeed a function $G\times G\to G$
and that the various properties of the op operation are satisfied. 
\end{rem}

\begin{rem}
In the format as we specify it, definitions need to be syntactically
correct, but not necessarily consistent; they are to be thought of
as mere shortcuts enabling more concise and clearer statement formulations. 
\end{rem}

\begin{rem}
In mathematics, definitions of objects such as Riemann's zeta function
$\zeta$ by a series such as $\sum_{n=1}^{\infty}n^{-s}$ on $\mathbb{H}_{1}:=\left\{ s\in\mathbb{C}:\Re\mathfrak{e}\left(s\right)>1\right\} $
involve a lemma (saying the series converges on $\mathbb{H}_{1}$);
in our framework, we would instead define a property `is-zeta-on-H1'
for a function $f:\mathbb{H}_{1}\to\mathbb{C}$ which would mean that
the series $\sum_{n=1}^{\infty}n^{-s}$ converges for any $s\in\mathbb{H}_{1}$
and its value is $f\left(s\right)$; a statement (needed to e.g. prove
the prime number theorem) would then assert that there exists a unique
function $\mathbb{H}_{1}\to\mathbb{C}$ that satisfies the `is-zeta-on-H1'
property; a lemma such as $\zeta\left(s\right)=\prod_{p}\frac{1}{1-p^{-s}}$
would then read \textquotedbl fix a function $\zeta:\mathbb{H}_{1}\to\mathbb{C}$;
assume that $\zeta$ satisfies the `is-zeta-on-H1' property; then
for any $s\in\mathbb{H}_{1}$, we have $\zeta\left(s\right)=\prod_{p}\frac{1}{1-p^{-s}}$\textquotedbl .
A specific language may include shorthands to make the alleviate the
notation, of course. 
\end{rem}

Adding definitions to the proof construction, we obtain the following
format for statements:
\begin{defn}
\label{def:claim-of-proof-format}The Claim of Proof Format (see Figure
\ref{fig:claim-of-proof}) consists of statements, high-level claims
of proof, and machine-level claims of proof, associated with a fixed
logic system $\Lambda$ (as in Section \ref{subsec:computer-proof-systems}): 
\begin{itemize}
\item A statement $\mathbf{S}$ consists of a context $\Gamma$ of definitions
and an implication $\mathbf{A}\implies\mathbf{C}$. 
\item A claim of proof $P_{L}$ of level $L\geq1$ of a statement $\mathbf{S}$
with context $\Gamma$ and implication $\mathbf{A}\implies\mathbf{C}$
consists of a chain of reasoning $\mathbf{D},\mathbf{S}_{1},\ldots,\mathbf{S}_{k}$,
where 
\begin{itemize}
\item $\mathbf{D}$ is a collection of definitions (as in Definition \ref{def:collection-of-definitions}).
\item $\left(\mathbf{S}_{j}\right)_{j=1,\ldots,k}$ are statements with
contexts $\Gamma\cup\mathbf{D}$ and where $\mathbf{S}_{j}$ is the
implication $\mathbf{A}_{j}\implies\mathbf{C}_{j}$, such that
\begin{itemize}
\item $\mathbf{A}_{j}$ is of the form $\mathbf{A}\cup\left\{ \mathbf{C}_{i}\text{ for }i\in\mathcal{I}_{j}\right\} $
for some $\mathcal{I}_{j}\subset\left\{ 1,\ldots,j-1\right\} $.
\item $\mathbf{C}_{k}=\mathbf{C}$. 
\end{itemize}
\item It is claimed that the statements $\mathbf{S}_{1},\ldots,\mathbf{S}_{k}$
have claims of proof of level $\leq L-1$. 
\end{itemize}
\item A claim of proof $P_{0}$ of level $0$ of a statement $\mathbf{S}$
is a sequence of statements which can be validated by a computer,
which follow the rules of the logic system $\Lambda$, and which allow
one to deduce $\mathbf{S}$. 
\end{itemize}
\end{defn}

\begin{rem}
A high-level claim of proof comes with the following implicit claim:
for each $j=1,\ldots,k$, deriving $\mathbf{S}_{j}$ is significantly
easier than deriving $\mathbf{S}_{*}$. 
\end{rem}

The length of a claim of proof is measured by its aggregated length
measure $\mu$ that is an increasing function of the length of its
chain of reasoning (measured in number of symbols in the language
in which it is expressed, possibly with a weight associated with different
symbols); for instance, a natural choice for $\mu$ is simply the
total length of the statements measured in number of symbols. 
\begin{rem}
\begin{figure}
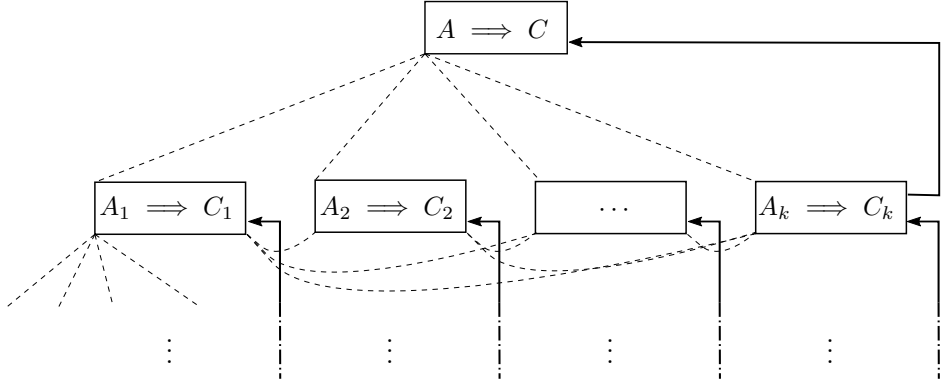

\caption{\label{fig:claim-of-proof}Claim of Proof Format. The boxes' left
sides correspond to assumptions, while the boxes' right sides correspond
to conclusions. The light dashed lines represent possible assumption
dependencies. }
\end{figure}
\end{rem}

\subsubsection{\label{subsec:top-down-informal-view}Top-Down Informal View of SPRIG}

Generalizing the examples, we now give an informal description of
SPRIG, in the order in which the interactions between agents take
place. 

Given a context $\Gamma$ and an aggregated length measure $\mu$
the validation mechanism goes as follows:
\begin{enumerate}
\item Given a parameter $L\geq1$, the root of the process consists of 
\begin{enumerate}
\item either a claim $C_{L}$ of level $L$ (a statement $\mathbf{S}$ with
context $\Gamma$, together with a claim of proof $\mathbf{P}_{L}$
of level $L$, as in Definition \ref{subsec:claim-of-proof-format})
with a pre-specified stake $\sigma_{L}=\sigma_{L}^{\downarrow}$;
\item or a question $Q_{L}$ (i.e. a statement $\mathbf{S}$ with context
$\Gamma$) with status `unanswered' and a pre-specified bounty $\beta_{L}$. 
\end{enumerate}
The root specifies maximum proof lengths $\lambda_{L-1},\ldots,\lambda_{0}$,
stakes $\sigma_{L-1}^{\uparrow},\sigma_{L-1}^{\downarrow},\ldots,\sigma^{\uparrow},\sigma_{1}^{\downarrow},\sigma_{0}^{\uparrow},\sigma_{0}$
and bounties $\beta_{L-1},\ldots,\beta_{0}$ to be used at each of
the lower levels. 
\item For $\ell\geq0$, a claimer might attempt to answer a question $Q_{\ell}=\left(\mathbf{S}\right)$
by producing a claim $C=\left(\mathbf{S},\mathbf{P}\right)$ of level
$\ell'\in\left\{ \ell,0\right\} $, i.e. by providing:
\begin{itemize}
\item A claim of proof $\mathbf{P}$ of level $\ell'$ of the statement
$\mathbf{S}$.
\item If $\ell'=\ell$: 
\begin{itemize}
\item The claim of proof $\mathbf{P}$ must be of total length at most $\mu\left(P\right)\leq\lambda_{\ell}$.
\item The claimer must lock a stake pair $\sigma_{\ell}^{\uparrow},\sigma_{\ell}^{\downarrow}$
(with $\sigma_{\ell}^{\uparrow}=0$ if $\ell=L$). 
\end{itemize}
\item If $\ell'=0$: 
\begin{itemize}
\item The claim of proof $\mathbf{P}$ must be of length at most $\lambda_{0}$, 
\item The claimer must lock a stake $\sigma_{\ell}^{\uparrow}$ and pay
a computation cost $c_{0}$. 
\item In this case, if the claim of proof is (automatically) validated,
it gets the status `validated', otherwise, it gets the status `invalidated'. 
\end{itemize}
\item In all cases:
\begin{itemize}
\item All claims of proof addressing $Q_{\ell}$ must be proposed within
a response time $\tau_{\ell}$ of $Q_{\ell}$'s publication.
\item If a claim of proof addressing $Q_{\ell}$ gets the status `validated',
then this claim is said to be answering $Q_{\ell}$ and $Q_{\ell}$
gets the status `answered'; if no such claim exists, then $Q_{\ell}$
gets the status `unanswered'. 
\end{itemize}
\end{itemize}
\item For $\ell\geq1$, a skeptic might dispute a level-$\ell$ claim $C_{\ell}=\left(\mathbf{S},\mathbf{P}_{\ell}\right)$
by asking a question $Q_{\ell-1}=\left(\mathbf{S}\right)$ where $\mathbf{S}$
is one of the statements appearing in the claim of proof $\mathbf{P}_{\ell}$.
\begin{enumerate}
\item The skeptic must lock a bounty $\beta_{\ell}$ associated with the
question. 
\item All questions about $C_{\ell}$ must be asked within the verification
time $\theta_{\ell}$ of $C_{\ell}$'s publication.
\item If a question originating from the claim gets the status `unanswered',
then this question is said to be a defeating question, and the claim
gets the status `invalidated'; if no such question exists, the claim
gets the status `validated'.
\end{enumerate}
\end{enumerate}
The incentivization mechanism for the proposal is based on bounties
$\left(\beta_{\ell}\right)_{\ell}$ and stakes $\left(\sigma_{\ell}^{\uparrow},\sigma_{\ell}^{\downarrow}\right)_{\ell}$
as follows:
\begin{enumerate}
\item If a claim $C_{\ell}$ addressing a question $Q_{\ell}$ is the first
one to get the status `validated', then $C_{\ell}$ receives the bounty
$\beta_{\ell}$ from $Q_{\ell}$.
\item If a claim $C_{\ell}$ addressing a question $Q_{\ell}$ gets the
status `invalidated', then $Q_{\ell}$ receives the stake $\sigma_{\ell}^{\uparrow}$
from $C_{\ell}$.
\item If a question $Q_{\ell}$ disputing a claim $C_{\ell+1}$ is the first
one to get the status `unanswered', then $Q_{\ell}$ receives the
stake $\sigma_{\ell+1}^{\downarrow}$ from $C_{\ell}$.
\end{enumerate}
In a nutshell, there are debates between claimers (agents providing
claims of proof for statements) and skeptics (agents asking questions
about proofs), where each side debates while having some `skin in
the game': claimers and skeptics must pay upfront to play, and they
will be paid back if their point is valid (i.e. the claim is validated,
or the question remains unanswered) and possibly further compensated
(for a question, if it is the first to defeat the claim it originates
from; for a claim, if it the first to answer the question it originates
from). Winning occurs when one of the sides concedes, and in case
no side wants to concede, after at most $D$ steps, one reaches the
point where only machine-level proofs are accepted; hence the ultimate
judge is an algorithm that runs the checking of the machine-level
proof. 

Claimers and skeptics have dual roles. Let us simply point out the
following differences:
\begin{itemize}
\item The skeptics only have a limited number of possible moves (limited
by the number of steps in the claims of proof that have been published),
while the provers have a virtually infinite number of possible moves
(they can provide any purported claim of proof).
\item While invalidated claims must share their stake to the question they
address (if it exists) and of the first defeating question, the answered
questions must only pay their bounties to the first validated claim
of proof closing them. 
\end{itemize}
\begin{rem}
The protocol interaction does not necessarily stop immediately after
the root status has been set. For instance, a claim may be invalidated
by a first unanswered question, but the status of questions that were
raised after that first question may still be undecided; the protocol
interaction must run until all questions and claims get a status. 
\end{rem}

\subsubsection{\label{subsec:formal-description}Formal Description}

We now give the formal description of the SPRIG protocol introduced
in Section \ref{subsec:top-down-informal-view}. A context $\Gamma$
is fixed by the root, as well as a high-level proof aggregated length
measure $\mu$. 

We work with claims $C_{\ell}$ and questions $Q_{\ell}$ of levels
$\ell=0,1,\ldots$. We denote by $\mathcal{C}_{\ell},\mathcal{Q}_{\ell}$
the corresponding types (and we write e.g. $C_{\ell}\in\mathcal{C}_{\ell}$
to indicate that $C_{\ell}$ is a claim of level $\mathcal{C}_{\ell}$).

For simplicity, the protocol assumes that questions and claims are
submitted in continuous time and cannot appear simultaneously. Similarly,
we assume that the claims and questions are published at the moment
of their creation. See Section \ref{subsec:blockchain-implementation}
for a discussion of this issue in the context of smart contracts. 

\paragraph*{Type $\mathcal{C}_{\ell}$ for $\ell\protect\geq1$}
\begin{itemize}
\item Data:
\begin{itemize}
\item An origin question $Q_{\ell}\in\mathcal{Q}_{\ell}\cup\left\{ \mathrm{none}\right\} $
(we say that the claim originates from $Q_{\ell}$).
\item A mathematical statement $\mathbf{S}$:
\begin{itemize}
\item The statement $\mathbf{S}$ of the origin $Q_{\ell}$ if $Q_{\ell}\neq\mathrm{none}$. 
\item An independent mathematical statement if $Q_{\ell}=\mathrm{none}$.
\end{itemize}
\item A claim of proof $\mathbf{P}_{\ell}=\mathbf{D},\mathbf{S}_{1},\ldots,\mathbf{S}_{k}$
of level $\ell$ of $\mathbf{S}$ of aggregated length $\mu\left(\mathbf{P}_{\ell}\right)\leq\lambda_{\ell}$, 
\item Parameters $\pi_{\mathcal{C}_{\ell}}$: max-length $\lambda_{\ell}$,
stake pair $\left(\sigma_{\ell}^{\uparrow},\sigma_{\ell}^{\downarrow}\right)$
(with $\sigma_{d}^{\uparrow}=0$ if $Q_{\ell}=\mathrm{none}$), verification
time $\theta_{\ell}$, $\mathcal{Q}_{\ell-1}$ parameters $\pi_{\mathcal{C}_{\ell-1}}$
if $\ell\geq1$ (bounty $\beta_{\ell-1}$, response time $\tau_{d-1}$,
$\pi_{\ell-1}$ parameters). 
\end{itemize}
\item Necessary creation of initial funds: $\sigma_{\ell}^{\uparrow}+\sigma_{\ell}^{\downarrow}$.
\item Status outcome:
\begin{itemize}
\item The claim gets the status `invalidated' if there exists a defeating
question, i.e. a question $Q_{\ell-1}\in\mathcal{Q}_{\ell-1}$ that 
\begin{enumerate}
\item originates from the claim and disputes from one of statements $\mathbf{S}_{1},\cdots,\mathbf{S}_{k}$
of its claim of proof;
\item respects the parameters $\pi_{\mathcal{Q}_{\ell-1}}$ (i.e. whose
parameters are $\pi_{\mathcal{Q}_{\ell-1}}$ );
\item has the appropriate creation funds;
\item appears less than $\theta_{\ell}$ units of time after the publication
of the claim;
\item gets the status `unanswered'.
\end{enumerate}
\item Otherwise (if no defeating question exists): the claim has status
`validated'.
\end{itemize}
\item Stakes/bounties outcomes:
\begin{itemize}
\item If the claim gets the status `validated': 
\begin{itemize}
\item the stakes $\sigma_{\ell}^{\uparrow},\sigma_{\ell}^{\downarrow}$
are reimbursed to the claim owner; 
\item if $Q_{\ell}\neq\mathrm{none}$, and it is the first descendent of
$Q_{\ell}$ to get the `validated' status, the bounty $\beta_{\ell}$
of $Q_{\ell}$ is paid to the claim. 
\end{itemize}
\item If the claim gets the status `invalidated': 
\begin{itemize}
\item the stake $\sigma_{\ell}^{\uparrow}$ is paid to the $\mathcal{Q}_{\ell}$
origin, if there is one;
\item the stake $\sigma_{d}^{\downarrow}$ is paid to the first defeating
question owner. 
\end{itemize}
\end{itemize}
\end{itemize}

\paragraph*{Type $\mathcal{Q}_{\ell}$ for $\ell\protect\geq0$. }
\begin{itemize}
\item Parameters $\pi_{\mathcal{Q}_{\ell}}$: bounty $\beta_{\ell}$, response
time $\tau_{\ell}$, $\mathcal{C}_{\ell}$ parameters $\pi_{\mathcal{C}_{\ell}}$
(max-length $\lambda_{\ell}$, stake pair $\left(\sigma_{\ell}^{\uparrow},\sigma_{\ell}^{\downarrow}\right)$,
verification time $\theta_{\ell}$, $\mathcal{Q}_{\ell-1}$ parameters
$\pi_{\mathcal{Q}_{\ell-1}}$ if $\ell\ge1$). 
\item Necessary creation of initial funds: $\beta_{\ell}$.
\item Data:
\begin{itemize}
\item An origin claim $C_{\ell+1}\in\mathcal{C}_{\ell+1}\cup\left\{ \mathrm{none}\right\} $,
we say that the question originates from $C_{\ell+1}$. 
\item A mathematical statement $\mathbf{S}$ with context $\Gamma$: 
\begin{itemize}
\item One of the statements $\mathbf{S}$ in the claim of proof $P_{\ell+1}=\mathbf{D},\mathbf{S}_{1},\ldots,\mathbf{S}_{k}$
of $C_{\ell+1}$ if $C_{\ell+1}\neq\mathrm{none}$; in such a case,
we say that the question disputes $\mathbf{S}$.
\item An independent mathematical statement if $C_{\ell+1}=\mathrm{none}$. 
\end{itemize}
\end{itemize}
\item Outcome: 
\begin{itemize}
\item The stake $\sigma_{\ell}^{\uparrow}$ is paid to the question owner
by any claim $C_{\ell}\in\mathcal{C}_{\ell}$ that
\begin{enumerate}
\item originates from the question;
\item respects the $\pi_{\mathcal{C}_{\ell}}$ parameters (i.e. whose parameters
are $\pi_{\mathcal{C}_{\ell}}$);
\item has the appropriate creation funds;
\item appears less than $\tau_{\ell}$ unit of time after the publication
of the question;
\item gets the status `invalidated'.
\end{enumerate}
\item If there is a claim $C\in\mathcal{C}_{\ell}\cup\mathcal{C}_{0}$ that 
\begin{enumerate}
\item originates from the question;
\item respects the $\pi_{\mathcal{C}_{\ell}}$ parameters;
\item has the appropriate creation funds;
\item appears less than $\tau_{\ell}$ unit of time after the publication
of the question
\item gets the status `validated'
\end{enumerate}
then the question gets the status `answered'.
\end{itemize}
Otherwise (i.e. no such claim has appeared) the question is marked
as `unanswered'.
\item Stakes/bounties outcomes: 
\begin{enumerate}
\item If the question gets the status `answered': 
\begin{itemize}
\item The owner of the first validated claim $C\in\mathcal{C}_{\ell}\cup\mathcal{C}_{0}$
originating from the question gets the bounty $\beta_{\ell}$; 
\item The next such claims get nothing (but lose nothing).
\end{itemize}
\item If the question gets the status `unanswered': 
\begin{itemize}
\item The bounty $\beta_{d}$ is reimbursed to the question owner. 
\item If the question is the first question originating from $C_{\ell+1}$
to get the `unanswered' status, the stake $\sigma_{\ell+1}^{\downarrow}$
is paid by the claim to the question owner; 
\item The next such questions get nothing (but lose nothing). 
\end{itemize}
\end{enumerate}
\end{itemize}

\paragraph*{Type $\mathcal{C}_{0}$}
\begin{itemize}
\item Data:
\begin{itemize}
\item An origin question $Q_{\ell}\in\mathcal{Q}_{\ell}$ (we say that the
claim originates from $Q_{\ell}$) for $\ell\geq0$.
\item The statement $\mathbf{S}$ of the origin $Q_{\ell}$.
\item A machine-verifiable claim of proof $\mathbf{P}_{0}$ of length $\leq\lambda_{0}$. 
\item Parameters: max-length $\lambda_{0}$, stake $\sigma_{0}^{\uparrow}$,
computation cost $c_{0}$. 
\end{itemize}
\item Necessary creation initial funds: $\sigma_{0}^{\uparrow}+c_{0}$.
\item Status outcome:
\begin{itemize}
\item If $\mathbf{P}_{0}$ is validated by the computer verification system,
the claim gets the status `validated'.
\item Otherwise, the claim gets the status `invalidated'. 
\end{itemize}
\item Stakes/bounties outcomes:
\begin{itemize}
\item If the claim gets the status `validated': the stake $\sigma_{0}^{\uparrow}$
is reimbursed to the claim owner. 
\item If the claim gets the status `invalidated': the stake $\sigma_{0}^{\uparrow}$
is paid to the origin.
\item The computation cost $c_{0}$ is burnt.
\end{itemize}
\end{itemize}
\begin{rem}
As each claim and question contains the parameters that the claim
and questions originating from it must respect, the parameters of
the entire interaction defined by the SPRIG protocol are specified
by the parameters of the root. 
\end{rem}

\subsection{\label{subsec:illustrations-of-the-protocol}Illustrations Of SPRIG}

In this subsection, we present illustrations of SPRIG-based interactions. 
\begin{itemize}
\item Claims of proof are depicted as horizontal segments, with dashed segments
representing machine-level proofs. 
\item Questions on parts of the proofs are represented as vertical segments. 
\item Small teal segments on the claims of proof/questions segments represent
the end of the allotted response time. 
\item The validation status of a claim of proof at the end of the interaction
is represented at the right end of the line (validated:$\checkmark$,
invalidated: $\times$)). Claims of proof written in machine-level
language are marked with a green diamond. 
\item Questions are represented as vertical lines (emanating from a statement
in a claim of proof), and their eventual status is represented at
the bottom end of the line (answered: $\checkmark$, unanswered: $\times$). 
\end{itemize}
Four basic examples are provided (Figures \ref{fig:basic-validated-claim}--\ref{fig:basic-unanswered-question}),
which are subparts of two examples of SPRIG  runs, one with a claim
of proof as root (Figure \ref{fig:claim-root-complete}) and another
one with a question as root (Figure \ref{fig:question-root-complete}).

\begin{figure}
\begingroup%
  \makeatletter%
  \providecommand\color[2][]{%
    \errmessage{(Inkscape) Color is used for the text in Inkscape, but the package 'color.sty' is not loaded}%
    \renewcommand\color[2][]{}%
  }%
  \providecommand\transparent[1]{%
    \errmessage{(Inkscape) Transparency is used (non-zero) for the text in Inkscape, but the package 'transparent.sty' is not loaded}%
    \renewcommand\transparent[1]{}%
  }%
  \providecommand\rotatebox[2]{#2}%
  \newcommand*\fsize{\dimexpr\f@size pt\relax}%
  \newcommand*\lineheight[1]{\fontsize{\fsize}{#1\fsize}\selectfont}%
  \ifx\svgwidth\undefined%
    \setlength{\unitlength}{174.33494164bp}%
    \ifx\svgscale\undefined%
      \relax%
    \else%
      \setlength{\unitlength}{\unitlength * \real{\svgscale}}%
    \fi%
  \else%
    \setlength{\unitlength}{\svgwidth}%
  \fi%
  \global\let\svgwidth\undefined%
  \global\let\svgscale\undefined%
  \makeatother%
  \begin{picture}(1,0.62789462)%
    \lineheight{1}%
    \setlength\tabcolsep{0pt}%
    \put(0,0){\includegraphics[width=\unitlength,page=1]{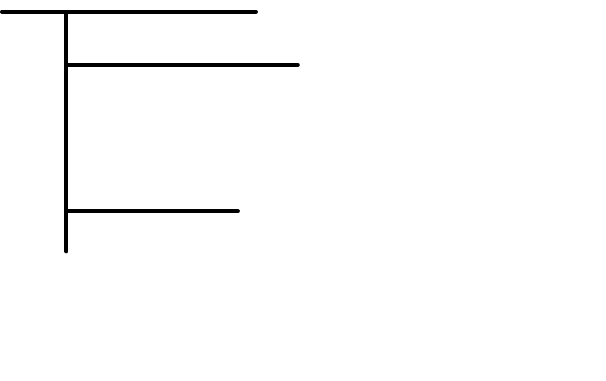}}%
    \put(0.50287829,0.57289569){\makebox(0,0)[lt]{\lineheight{1.25}\smash{\begin{tabular}[t]{l}\\\textcolor{red}{\ding{55}}\end{tabular}}}}%
    \put(0.09559243,0.00592374){\makebox(0,0)[lt]{\lineheight{1.25}\smash{\begin{tabular}[t]{l}\textcolor{blue}{\ding{51}}\end{tabular}}}}%
    \put(0.50181293,0.26545133){\makebox(0,0)[lt]{\lineheight{1.25}\smash{\begin{tabular}[t]{l}\textcolor{red}{\ding{51}}\end{tabular}}}}%
    \put(0.59490274,0.59532395){\makebox(0,0)[lt]{\lineheight{1.25}\smash{\begin{tabular}[t]{l}\textcolor{red}{\ding{51}}\end{tabular}}}}%
    \put(0,0){\includegraphics[width=\unitlength,page=2]{figure-basic-claim-validated.pdf}}%
    \put(0.34446799,0.41417321){\makebox(0,0)[lt]{\lineheight{1.25}\smash{\begin{tabular}[t]{l}\\\textcolor{blue}{\ding{55}}\end{tabular}}}}%
    \put(0.43050929,0.41417321){\makebox(0,0)[lt]{\lineheight{1.25}\smash{\begin{tabular}[t]{l}\\\textcolor{blue}{\ding{55}}\end{tabular}}}}%
    \put(0,0){\includegraphics[width=\unitlength,page=3]{figure-basic-claim-validated.pdf}}%
  \end{picture}%
\endgroup%

\caption{\label{fig:basic-validated-claim}A basic validated claim of proof
(top horizontal segment): one question was raised (vertical segment
on the left). In answer to this question, a first claim of proof was
proposed (middle horizontal segment) and invalidated by two unanswered
questions (two short vertical segments), but then a second claim of
proof was proposed, which was validated as no question was raised
about it. }
\end{figure}

\begin{figure}
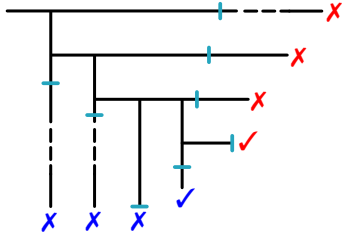

\caption{\label{fig:basic-invalidated-proof}A basic invalidated claim of proof
(top horizontal segment): one question was raised (vertical segment
on the left); as an answer to this question, a claim of proof was
proposed (second horizontal segment from the top); the claim of proof
was itself questioned, and an additional claim of proof was proposed
as an answer (second horizontal segment from the bottom); two questions
were raised about that additional claim of proof (the two bottom-most
vertical segments), the first of which was unanswered, and the second
of which was answered by a validated claim of proof (bottom-most horizontal
segment). As a result, the claim was invalidated.}
\end{figure}

\begin{figure}
\begingroup%
  \makeatletter%
  \providecommand\color[2][]{%
    \errmessage{(Inkscape) Color is used for the text in Inkscape, but the package 'color.sty' is not loaded}%
    \renewcommand\color[2][]{}%
  }%
  \providecommand\transparent[1]{%
    \errmessage{(Inkscape) Transparency is used (non-zero) for the text in Inkscape, but the package 'transparent.sty' is not loaded}%
    \renewcommand\transparent[1]{}%
  }%
  \providecommand\rotatebox[2]{#2}%
  \newcommand*\fsize{\dimexpr\f@size pt\relax}%
  \newcommand*\lineheight[1]{\fontsize{\fsize}{#1\fsize}\selectfont}%
  \ifx\svgwidth\undefined%
    \setlength{\unitlength}{149.78031751bp}%
    \ifx\svgscale\undefined%
      \relax%
    \else%
      \setlength{\unitlength}{\unitlength * \real{\svgscale}}%
    \fi%
  \else%
    \setlength{\unitlength}{\svgwidth}%
  \fi%
  \global\let\svgwidth\undefined%
  \global\let\svgscale\undefined%
  \makeatother%
  \begin{picture}(1,0.89150116)%
    \lineheight{1}%
    \setlength\tabcolsep{0pt}%
    \put(0,0){\includegraphics[width=\unitlength,page=1]{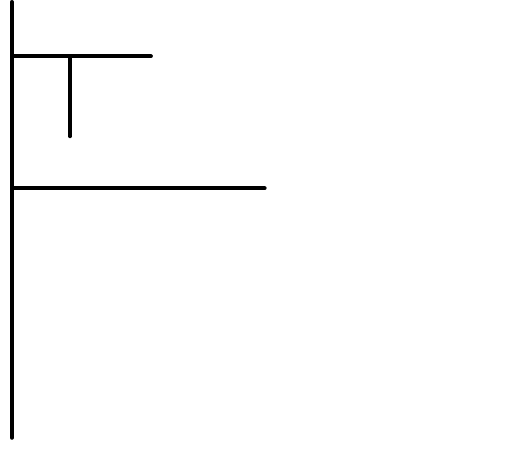}}%
    \put(0.0102622,0.00689487){\makebox(0,0)[lt]{\lineheight{1.25}\smash{\begin{tabular}[t]{l}\textcolor{blue}{\ding{51}}\end{tabular}}}}%
    \put(0.47029659,0.43161511){\makebox(0,0)[lt]{\lineheight{1.25}\smash{\begin{tabular}[t]{l}\textcolor{red}{\ding{51}}\end{tabular}}}}%
    \put(0,0){\includegraphics[width=\unitlength,page=2]{figure-basic-question-answered.pdf}}%
    \put(0.52849208,0.51173238){\makebox(0,0)[lt]{\lineheight{1.25}\smash{\begin{tabular}[t]{l}\textcolor{red}{\ding{51}}\end{tabular}}}}%
    \put(0.29571022,0.76209915){\makebox(0,0)[lt]{\lineheight{1.25}\smash{\begin{tabular}[t]{l}\textcolor{red}{\ding{55}}\end{tabular}}}}%
    \put(0.11355715,0.57426956){\makebox(0,0)[lt]{\lineheight{1.25}\smash{\begin{tabular}[t]{l}\textcolor{blue}{\ding{55}}\end{tabular}}}}%
    \put(0.21181408,0.19559223){\makebox(0,0)[lt]{\lineheight{1.25}\smash{\begin{tabular}[t]{l}\textcolor{blue}{\ding{51}}\end{tabular}}}}%
    \put(0,0){\includegraphics[width=\unitlength,page=3]{figure-basic-question-answered.pdf}}%
  \end{picture}%
\endgroup%
\caption{\label{fig:basic-answered-question}A basic answered question (left-most
vertical segment): a first claim of proof was proposed (top-most horizontal
segment), which was then invalidated by an unanswered question, but
then a second claim of proof was proposed which was validated after
a question was raised (bottom-most vertical segment), and that question
was answered by a claim (bottom-most horizontal segment) which itself
was not questioned.}
\end{figure}

\begin{figure}
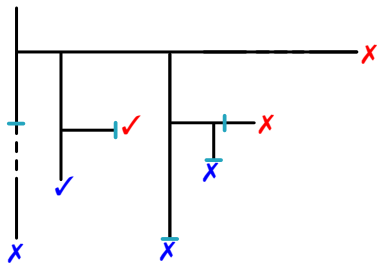

\caption{\label{fig:basic-unanswered-question}A basic unanswered question
(left-most vertical segment). A claim of proof was proposed (top-most
horizontal segment), which resisted a first question (second left-most
vertical segment), as that question was answered by an unquestioned
claim of proof, but which did not resist the second question, as the
only claim of proof answering it (right-most horizontal segment) was
invalidated by an unanswered question.}
\end{figure}

\begin{figure}
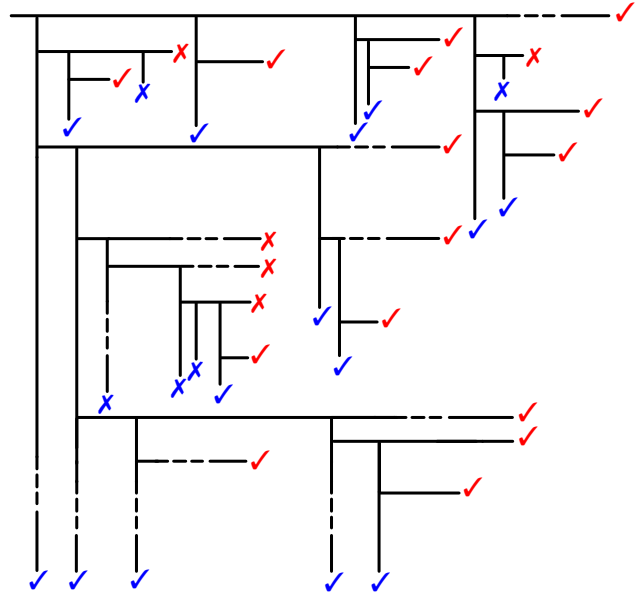

\caption{\label{fig:claim-root-complete}In this case, the claim was eventually
validated. Time goes horizontally in the lifetime of claims and vertically
in the lifetime of questions, but it is not represented at scale.\protect \\
In this figure, four questions were asked and successfully answered;
for the first question, a first claim of proof was proposed, which
was then invalidated, followed by one which was later validated; the
same is true for the fourth question that was asked. The second and
third questions were answered by claims of proof which were later
validated. The second claim of proof proposed for the first question
was itself only validated after the first question about it saw two
claims of proof (a first one, which was invalidated after going down
3 more levels), and a second one, which was validated after questions
were asked and answered at the machine level.}
\end{figure}

\begin{figure}
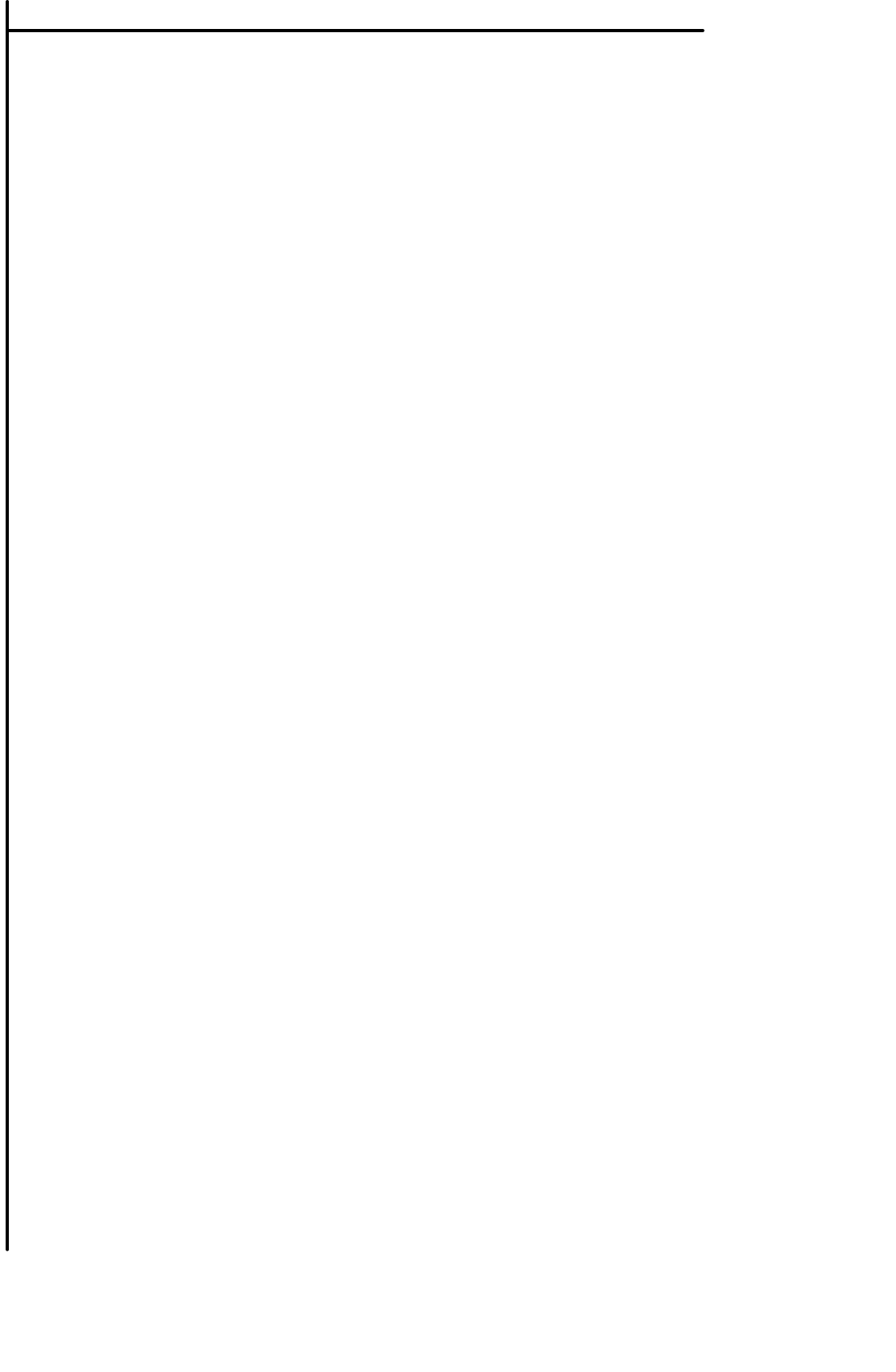

\caption{\label{fig:question-root-complete}Illustration of the SPRIG protocol
when the root is a question. The same notation convention is used
as in Figure \ref{fig:claim-root-complete}. In this case, the question
was answered by the third claim of proof (and the interaction was
ended as a result); a fourth claim of proof was submitted, but its
status was not yet decided at the time the interaction has ended;
all question marks indicate statuses that are not yet assigned.}
\end{figure}

\subsection{\label{subsec:variants-and-extensions}Variants and Extensions}

In this section, we present a number of variants and extensions of
the SPRIG protocol, which can be enabled to optimize for various goals
under certain environments. Many more variants and extensions can
in principle be considered, but we focus here on the ones that appear
to be the most naturally motivated and that live directly on the protocol
itself. A number of further interesting extensions can be then built
upon protocol instances, in particular, decentralized markets for
derivatives can rely on using protocol instances as oracles, as discussed
in Section \ref{subsec:derivatives-markets}. While these variants
and extensions appear to be promising, their detailed analysis is
significantly more complex and goes beyond the scope of this article.

\subsubsection{\label{subsec:time-varying-stakes-and-bounties}Time-Varying Stakes
and Bounties}

Intuitively (and as discussed in Sections \ref{sec:informal-game-theoretic-discussion}
and \ref{sec:a-simplified-equilibrium-analysis} below), the trust
in the fact that a claim of proof is correct depends on how favorable
the incentives are to those asking questions: if a skeptic has little
to gain, and too much to risk in asking questions (in terms of explicit
incentives), he may not ask a question about a claim, unless he is
very confident that the question cannot be answered (and he may not
want invest time and energy to find questions about claims). 

An agent publishing a claim of proof as a means to validate it may
wish to establish a high level of trust in it (by offering a high
stake, to be paid to a skeptic successfully challenging her claim),
but may herself not be very confident in its ultimate validity: there
may be a fairly obvious mistake (for instance of notation), and she
would not want to pay a high price for an obvious mistake. Similarly,
an organization may want to incentivize the solution to a given open
problem, but may not want to pay too much for it, if the question
turns out to be obvious. 

A solution to this is to rely on time-dependent stakes and bounties,
in a way similar to Dutch auctions: start with conditions that are
very favorable to the defending side (the side at the root of the
interaction), and make them more and more favorable for the challenging
side. If there is an obvious challenge (i.e. a question of whether
the root is a claim of proof, if the root is a question), challenging
agents will still be incentivized to pose it as soon as possible (rather
than to wait to increase their reward), as they are in competition
with other challenging agents.

At the same time, the owner of a claim may want to get back some of
the liquidity that she locked into the smart contract after a while,
while still incentivizing the search for mistakes in the proof; in
this case, having the stakes decrease over time could prove useful. 

In any case, the bounty and stake parameters at all levels of SPRIG
can be replaced by time-evolving functions, which must be specified
when the protocol instance is created. 

\subsubsection{\label{subsec:generalized-bounties-and-stakes} Generalized Bounties
and Stakes}

In the basic version of SPRIG, a bounty $\beta$ must be locked to
ask a question, which will be paid to the first validated claim of
proof answering it; dually, a stake pair $\left(\sigma^{\uparrow},\sigma^{\downarrow}\right)$
must be locked to propose a claim of proof, where, in case of invalidation,
the part $\sigma^{\uparrow}$ is paid to the question the claim of
proof was trying to answer, and the part $\sigma^{\downarrow}$ is
paid to the first question that invalidated the claim of proof. A
number of variants can be introduced in terms of distributions for
stakes and bounties. 

It may seem natural to propose an `upwards' bounty $\beta^{\uparrow}$
to be paid to the claim of proof a question derives from; however,
as discussed in Section \ref{sec:informal-game-theoretic-discussion},
such upwards bounties may however open the possibilities for certain
attacks (called the `Plagiarist's attack' below). 

Another possible extension to incentivize agents to disclose more
information into challenging a claim consists in splitting the stake
$\sigma^{\downarrow}$ into $p$ shares $\sigma_{1}^{\downarrow},\ldots,\sigma_{p}^{\downarrow}$
for the first $p$ unanswered questions about the claim (this may
require forbidding to ask the same question more than once). While
a single unanswered question suffices to invalidate a claim, a stake
may incentivize agents to look more closely at a claim of proof even
after a first question has been raised. On the other hand, rewarding
multiple claims of proof that successfully answer a question appears
to be much more delicate: this opens the possibility to Plagiarist's
attacks that are difficult to counter. 

\subsubsection{\label{subsec:claim-of-proof-dependent-parameters}Claim-of-Proof-Dependent
Parameters}

In the basic version of SPRIG, a fixed stake is associated with a
claim of proof, together with a fixed amount of time to raise questions,
and a fixed limit on the claim of proof length. While this setup incentivizes
the publishing of correct claims of proof and incentivizes agents
who spot a problem in a claim of proof to question it, it does account
for the fact that claims of proof may be more or less hard to read.
A way to account for this in the protocol is to allow for the stakes/bounties
and verification times to depend on the complexity of the claims of
proof submitted: in particular, longer claims of proof (as measured
in the lengths of their claims, or the number of them) may warrant
longer verification times (or the verification could be encouraged
by requiring higher stakes); for claims of proof consisting of many
statements, one may want to reduce the bounty to ask a question. 

For such variants, the parameters set by the root owners should be
replaced by functions (also set by the root owner) of the proof complexity
and number of claims. 

\subsubsection{\label{subsec:synchronous-protocol}Synchronous SPRIG }

The original version of SPRIG is intrinsically asynchronous in nature:
each question and each claim of proof runs on an independent clock,
and new questions or claims of proof may come at an arbitrary time,
starting their own clocks. While this process incentivizes agents
to disclose information as soon as they have it, and is more efficient
at closing obvious cases sooner rather than later, a number of situations
may suggest using a synchronous variant of the protocol: after a question
is posed (say), a fixed response time is given to post claims of proof.
Claims of proof are hidden until the response time is elapsed, at
which point they are revealed. Then, a round of questioning starts,
giving a fixed amount of time to ask questions about the various claims
of proof; the questions are also hidden and revealed when the questioning
round ends. Then a third round starts, in which claims of proof are
proposed in response to the questions can be posted and revealed at
the end of the third round. 

This variant of the protocol may be useful for revealing information
at specific times (e.g. yearly contest) or to make the best uses of
a community's resources (perhaps the best critics of a theorem' claim
of proof are other agents trying to submit at the same time their
own claim of proof, and they can focus on criticizing others' proofs
some of the time, while focussing on answering questions the rest
of the time).

\subsubsection{\label{subsec:exclusive-disclosure}Exclusive Disclosure}

In the basic version of SPRIG, agents ask questions when they doubt
the validity of a step in a claim of proof. However, an intrinsic
motivation may come from a question-raiser: that they are curious,
for independent reasons, about the answer to the question. In such
a setup, it may be useful to guarantee to the first question-raiser
an exclusive access to the answers to a question for a brief amount
of time, before the answer is made public (the poster of answers would
be incentivized not to disclose their answers to other parties in
the exclusive time period, as it limits their attack surface). 

\subsubsection{\label{subsec:expedited-validation-or-invalidation}Expedited Validation
or Invalidation}

In a number of cases, it may be desirable to expedite a validation
process: a variant can be introduced that allows a claimer to reduce
the validation time in exchange for higher stakes and/or lower response
times for the lower levels. This feature may prove desirable in certain
situations, but it must be dealt with carefully in order to not introduce
flaws (leading to the validation of a proof that should not have been
validated). 

\subsubsection{\label{subsec:open-questions-and-multi-question-bounties}Open Questions
and Multi-Question Bounties}

One may want to put at the root of a SPRIG instance both a question
and its negation: for instance, the Clay Institute offers a prize
for the first proof of $P=NP$ or of $P\neq NP$. In this case, a
single bounty (1M USD) is put at the root of the two questions, and
it should go to the first validated claim of proof for either question
that gets validated (as a result there is no bounty left for the other
question; this should not be a problem as a statement and its negation
should not have both validated proofs!).

More generally, an institute may want to put a single bounty for the
first agent answering one of a list of questions; as soon as one of
the questions is answered, there will be no bounty left for answering
other questions will be cancelled. Even more generally, a limited
number $K$ of bounties could be made available for the first $K$
questions answered, after which there will be no bounties left for
answering the other questions. 

\subsubsection{\label{subsec:stake-sharing-and-bounty-sharing}Stake-Sharing and
Bounty-Sharing}

A possible downside of the system is the barrier of entry for a participating
agent, who may not have enough funds, while at the same time possessing
useful information. 

For questions, such an agent could put a partial bounty, and wait
until this bounty is completed by other agents: at that moment, the
question is formally asked, and should the question be the first to
invalidate the claim of proof, the stake will be shared by the agents
who put the bounty, at the pro-rata of their bounty share. 

For \textcolor{black}{claims of proof}, such an agent could put an
encrypted claim of proof with a partial stake, try to find other agents
who also believe in it to complete the stake (for instance, by proving
her identity and using her reputation), and decrypt it when the stake
is full (thus avoiding a plagiarist to copy her claim of proof); again,
in this case, if there is a bounty to be won, it will be shared among
the stake holders at the pro-rata of their stake (or according to
some other pre-determined rule). 

\subsection{\label{subsec:blockchain-implementation}Blockchain Implementation}

The SPRIG protocol presented in Section \ref{subsec:protocol-description}
and its variants and extensions presented in Section \ref{subsec:variants-and-extensions},
are designed so it can be implemented on a blockchain, in a fully
decentralized manner, without reliance on an external oracle. In this
subsection, we discuss a number of design questions related to the
implementation of the protocol on a blockchain infrastructure. 

\subsubsection{\label{subsec:automated-proof-settlement}Automated Proof Settlement}

As emphasized in the top-down view of SPRIG (Section \ref{subsec:top-down-informal-view}),
what settles the boundary conditions of the protocol (and hence ensures
its good functioning) is the presence of an ultimate arbiter, in the
form of a computer-based system to verify machine-level claims. Implementing
SPRIG on a blockchain thus requires the ability to perform the necessary
computations on the blockchain to ensure transparency of the result
of the computation (or to offload the computation to another blockchain,
or to find a verifiable way to ensure the relevant computations were
done off-chain). 

While a large number of powerful proof verification programs are available
(see Section \ref{subsec:computer-proof-systems}), their emphasis
is usually on helping users to write proofs. The most desirable features
for a blockchain-based proof verification system are somewhat different. 
\begin{itemize}
\item Low memory usage: ultimately, a smart contract needs to be able to
verify any step of the computation. The sequence of computations may
not need to be performed entirely on-chain, as long as it is auditable
(see e.g. \cite{eberhardt-heiss}).
\item Syntax making the writing of definitions and statements (as specified
in Section \ref{subsec:claim-of-proof-format}) should be transparent
to the agents. This may be helped by the development of open off-chain
statement translators, assisting the users in the formalization of
definitions and statements. 
\end{itemize}
On the other end, the system living on the blockchain can be very
primitive in its ability to assist users to write down proofs; should
debates ever go down to the machine level, proof assistants can in
principle be used off-chain to propose machine-level claims of proof.
Still, if incentives are set well and agents are rational, the presence
of a well-functioning proof system will only serve as a deterrent:
close enough to the machine level, rational skeptics and agents should
already agree on the existence of a machine-level proof and the side
that is wrong is incentivized to concede early.

As a result of the above design goals and considerations, the development
of a proof verification system tailored for them seems desirable.

\subsubsection{\label{subsec:timing-and-concurrency-issues}Timing and Concurrency
Issues}

The block-based structure of blockchains serves crucially as a time-stamp
mechanism to validate the transactions: the consensus on the order
of the blocks serves as the measure of the passage of time (and crucially
at determining anteriority of modifications submitted to the blockchain:
this is in particular what prevents double-spending with Bitcoin).
As a result, the natural time unit of a blockchain is the number of
blocks emitted so far. 

In the description of SPRIG, time is treated as a continuous resource,
and time is asynchronous (except in the synchronous variant discussed
in Section \ref{subsec:synchronous-protocol}). For blockchains with
a sufficiently short validation time, and for non-trivial enough problems
discussed with the protocol, it is unlikely that two questions are
asked simultaneously (i.e. in the same block); however, in such a
case, a rule should be specified. Still, it is important to keep this
granularity in mind to avoid attacks by e.g. a quick plagiarist who
could copy a claim of proof and try to push it into the same block;
the solution in such a case is simply to make claimers first commit
a signed and encrypted version of their proof at least one block before
disclosing its content.

\subsubsection{\label{subsec:stakes-bounties-lock}Stakes and Bounties Lock}

For the protocol implementation, the agents need to lock their bounties
and stakes in the smart contract for a long time. In case they need
liquidity, it is possible for them to resell (i.e. transfer ownership
of) their stake in the contract to a third party. 

For the variant with time-varying stakes and bounties \ref{subsec:time-varying-stakes-and-bounties},
a number of challenges also arise: either the staker should put the
maximum amount of capital upfront or they could be mandated to inject
additional capital as time passes (at the risk of losing their stake
if they don't do so).

\section{\label{sec:informal-game-theoretic-discussion} Informal Game Theoretic
Discussion}

\subsection{\label{subsec:strategic-interactions-and-protocol-outcome}Strategic
Interactions and Protocol Outcome}

As mentioned in Section \ref{subsec:markets-information-and-games},
understanding how agents interact through the SPRIG protocol, as well
as interpreting the validation process outcome requires taking an
economic perspective. Indeed, while SPRIG is a set of rules that,
given the decisions of various users, deterministically defines a
tree, allocates rewards, and eventually settles the status of the
claims and questions, these very decisions are in essence strategic.

A complete characterization of the strategic interaction between users
is out of reach as they depend on a variety of elusive elements. First,
we do not have access to the real world's information structure (the
information set of each user and their beliefs about others\textquoteright{}
information sets), which at any rate would be highly intricate. Second,
this information structure is endogenous since the incentive scheme
can lead agents to work and gather additional information in a way
that is hard to capture (it depends e.g. on the mathematical background
of the agent, the difficulty of the problem). Third, the incentive
scheme itself is not fully characterized by the protocol bounties
and stakes as e.g. (i) a claimer presumably enjoys an (unobservable)
intrinsic reward when their proof is accepted and (ii) there might
be external incentives, too, as rewards related to SPRIG's outcome
might conceivably also be collectable in a secondary/derivatives market.
However, we can identify for each category of agents a number of high-level
features independent of the details discussed above:
\begin{description}
\item [{Provers}] The decision to enter (i.e. start interacting) in a SPRIG
instance depends on one\textquoteright s confidence about the validity
of one\textquoteright s claim of proof, the explicit incentives (stakes
and bounties), the private incentives (intrinsic reward of having
one\textquoteright s claim of proof accepted), beliefs about the skeptics'
ability to identify a flaw in the claim and beliefs about their incentives
for attempting to do so. Given that the skeptics' incentives are also
partly shaped by expectations about the incentives of subsequent claimers,
the validation game is dynamic; the final, machine-level step provides
the boundary condition. One important aspect of this dynamic process
is that the initial claimer need not be the one to address all (or
indeed any) subsequent questions from the skeptics. If the blockchain\textquoteright s
users get a sufficiently good grasp of the claimer\textquoteright s
argument, competition fostered by the incentive scheme makes it likely
that ungrounded skeptics\textquoteright{} challenges are answered
by third parties. This should deter `spamming' by the skeptics. Hence,
the initial claim must be sufficiently clear for the baton to be passed;
but being too explicit and detailed does not seem optimal for the
claimers either. Indeed, in that case, they perform upfront a task
that would have only been needed in case of a question, so that, if
the proof is correct, there is no need to do it immediately, and if
it is incorrect, the excessive level of detail makes it easier to
detect.
\item [{Skeptics}] The decision-making process of skeptics is similar in
the sense that it responds to the same incentive scheme and relies
on the formation of beliefs over the same objects. First, and obviously,
Skeptic\textquoteright s incentives to challenge increase with their
subjective probability of the claim of proof being wrong. Second,
they decrease with the probability that any part of the initial claim
can be converted into machine language in due time if necessary. A
skeptic can deem this unlikely if they observe that the claim of proof
or parts of it are somewhat obscure, or even if they have a sense
that the claim should be correct but the proof is too convoluted to
be transformed into machine language before the deadline. But there
are other incentives for a skeptic to challenge: they might want to
obtain information that is also relevant to another ongoing validation
process; or purely out of scientific interest. 
\end{description}
(Claim of) proof shapes are endogenous in SPRIG, emerging from the
interaction between claimers and skeptics. Indeed, it appears from
the discussion above that with properly designed incentives, claimers
would benefit from writing (claims of) proofs that are concise and
elegant (and easier to convert into machine language if necessary),
without being excessively terse (e.g. because no third party would
have enough information to defend the claim if needed). Hence, beyond
providing a decentralized way to produce a consensus about mathematical
claims, SPRIG also naturally delivers balanced, `agent-tailored' proofs:
sufficiently detailed to be convincing but sufficiently concise to
give intuition and be remembered. In particular, we expect that one
would rarely, if ever, need to reach the final, machine-language step.
As soon as the convertibility to machine level is credible, no skeptic
would have an incentive to push the process to that step. (The classic
analogy is with a government guaranteeing to intervene in case of
a banking panic; if such a guarantee is credible, then the panic would
not occur, and the intervention would never be needed).

The economic approach is also key for dealing with a crucial point:
how to interpret the fact that a claim of proof has been accepted
by the protocol? Understanding incentives is necessary for answering
this question. The simplest example is one where a claim has been
accepted without any challenge: is it because all users were fully
convinced or because the incentive scheme makes it prohibitively costly
(in expectation) to ask questions? In principle, given the correct
economic model (data of all incentives and information structure),
any Bayesian observer can use the protocol's outcome to compute the
probability that the proof is known by the market participants.

In Section \ref{sec:a-simplified-equilibrium-analysis}, we investigate
some of the game-theoretic aspects presented above in a stylized setting.
In particular, we explain how to do a Bayesian estimate of the probability
that a claim of proof is correct given that it has been accepted in
a (highly) simplified version of the protocol.

\subsection{\label{subsec:robustness-properties}Robustness Properties}

A\textcolor{black}{s in many }blockchain systems, strategies in SPRIG
can, broadly speaking, be divided into `honest strategies' and `attacks'.
The former refers to actions whose motivations are aligned with the
purpose of the protocol. The latter refers to attempts to game the
system, i.e. take advantage of the incentive scheme without contributing
to the end goal of the blockchain. We expect SPRIG to be\textcolor{black}{{}
robust. First, its} trust model shares similarities with that of optimistic
rollups \cite{OptRollUp} \cite{Arbitrum18}, and of the TrueBit protocol
\cite{TrueBit19}. Second, it features a large array of\textcolor{black}{{}
}parameters which we expect to be sufficiently rich to shape incentives
that deter attacks. To guide the specific choice of the parameters,
we now list several potential attacks together with which parameters
are to be tuned to thwart them.

\subsubsection{The Carpet-Bomber}

A skeptic may decide to question all parts of a claim in the hope
of stalling the process. The idea would be to induce the claimer to
concede by lack of resources and because there are not enough\textcolor{black}{{}
third-party }claimers available to help them defend. This is similar
to a DDoS attack on the protocol. The skeptic's goal is to collect
the stake of the claim.

This attack can be thwarted by appropriately choosing the question
bounties and the time allotted for the subsequent claimers' replies.
The former should not be too small relative to the stake and the latter
should be sufficiently long.

\subsubsection{The Nitpicker}

A skeptic may decide to ask for more and more details about a claim
of proof and refuse to concede until the machine level is reached.
Such an attack is not only based on the hope that a flaw will be identified
at some point but more importantly on the skeptic's desire to delay
the acceptance of the claim as much as possible. One reason could
be that the skeptics are themselves a claimer of an identical or similar
result, which they want to be accepted first.

This attack induces the claimer to present their proof with lemmas
of similar complexity. This mitigates incentives to nitpick as it
reduces the depth needed in order to expand the proof up to machine
level. Well-balanced bounties (i.e. not too low) and deadlines that
take into account the possibility of nitpicking (i.e. the maximal
time allotted for expansion up to machine level might indeed be reached)
contribute further to thwarting nitpicking attacks.

\subsubsection{The Evasive Prover }

A claimer may decide to be evasive, i.e. to stuff his claim with a
combination of irrelevant lemmas (purposely looking intricate but
for which they actually hold a machine level-proof) and one lemma
of complexity similar to the initial theorem, such that it is not
clear to outsiders which lemma to question. The goal is to deflect
questions towards the irrelevant lemmas and hence get the Claim accepted
and collect the questions' bounties. 

This attack can be thwarted by choosing the following parameters appropriately:
the stakes, the time allotted for the subsequent skeptics' questions,
the maximal level of a proof, and the maximal length of a claim. The
first two should be sufficiently large to incentivize skeptics to
work and identify the weak link. The last two should be sufficiently
small in order to cap the number of deflection targets and force the
claimer to `show their hand' quickly enough.

\subsubsection{The Sandbagger}

This attack mirrors the Carpet-Bombing one: a claimer may decide to
answer a Question with a multitude of claims in the hope of stalling
the process. The idea would be to induce the skeptic to concede by
lack of resources and because there are not enough third-party skeptics
available to help them continue challenging. The claimer's goal is
to collect the bounty of the Question.

This attack can be thwarted by appropriately choosing the claim stakes
and the time allotted for the subsequent skeptics' questions. The
former should not be too small relative to the bounty and the latter
should be sufficiently long.

\subsubsection{The Misleader}

A claimer may decide to stuff his claim with dubious lemmas and pursue
one of the following two strategies: 
\begin{itemize}
\item they attack the dubious lemmas and provide answers themselves in order
to ``intimidate'' the skeptics (improving their general credence
in the initial claim);
\item they attack the dubious lemmas and postpone answers to the very last
moment in order to mislead the skeptics into believing that other
users are already challenging (so there is no point in joining the
fray, as the stakes no longer seem earnable).
\end{itemize}
This attack can be handled similarly to the Sandbagger attack.

\subsubsection{The Plagiarist}

A mathematically illiterate agent can have a firm belief that a claimer
is able to answer a given Question correctly. This may occur, for
instance, on occasion, when an eavesdropper obtains information that
a researcher has a proof for a Question submitted by an institution,
or, more frequently, when the Question concerns the claim of another
claimer. The agent may then attempt to appropriate the proof of the
claimer in order to collect the Question's bounty as illustrated below. 

Consider a mathematician Alice who found a correct proof of a theorem.
She posts a corresponding claim on the blockchain. Then, Bob asks
a question about the claim, targetting statement $\mathbf{S}$. At
this stage, there could be an incentive for Charlie, the mathematically
illiterate agent, to immediately reply to Bob's question with a tautological
answer: `the proof of $\mathbf{S}$ is $\mathbf{S}$'; and to stick
to this strategy when questions are asked/repeated until Alice decides
to provide an answer herself to ensure that her Claim is not rejected.
From that point onwards, Charlie replicates any of Alice's replies
and challenges her using the questions skeptics ask him. 

In doing so, people may prefer to ask the questions straight to Alice
to know if she can provide a satisfying answer; they only challenge
Charlie with the same question if she does not. This may lead to a
faster approval of Charlie's claim. 

Fortunately, Alice can defend herself: as soon as she is challenged,
she answers the question, then asks Charlie the same question if it
was not already done by another skeptic and, instantaneously, provides
the same answer. This thwarts the Plagiarist's attack since it provides
a zero-cost defense mechanism to ensure that Charlie's Claim cannot
be validated before her own Claim.

\section{\label{sec:a-simplified-equilibrium-analysis}A Simplified Equilibrium
Analysis}

In this section, we analyze a tractable sequential game that captures
several key features of the strategic interaction between claimers
and skeptics through SPRIG. The adequate equilibrium concept for such
dynamic games with incomplete information is that of Perfect Bayesian
Equilibrium (PBE) \cite{fudenberg-tirole}. Our setup consists of
a game involving two players: Claimer (pronoun: she) and \textcolor{black}{Skeptic}
(pronoun: he). In our setup, Skeptic does not observe the initial
confidence of Claimer (a correctly estimated probability that her
proof is validatable, i.e. it is possible to unroll it down to machine
level) and must therefore form beliefs about it to proceed. A PBE
is a collection of actions and beliefs such that:
\begin{enumerate}
\item Given beliefs, the action taken by any agent at any node of the game
tree maximizes her or his expected utility.
\item Beliefs at each node of the game tree are consistent with the history
of actions, i.e. computed from Bayes' rule. 
\end{enumerate}
Hence, by constructing PBEs, one recognizes that the mere fact of
initiating a process in the protocol (or, in general, of pushing it
further) has informational content: intuitively, if a claimer posts
a proof, this should reflect the fact that she is relatively confident
about her proof and it should lead to an upwards update of the outsiders'
beliefs. For simplicity, we focus on the signaling content of the
entry decision, not of the parameters (deadlines, stakes, and bounties)
chosen at initiation. One could assume there is a set of `default'
parameters suggested by the protocol and then using them conveys limited
(although non-empty) signaling content. There is no fundamental obstacle
in extending the solution of our game to the case where the choice
of parameters is endogenous; but this would lead to a dramatic increase
in complexity without altering the key messages that we want to convey
in this section.

Section \ref{eco:sec:discussion} discusses the strengths and limitations
of our simplified protocol model and highlights directions in which
it can be enriched.

\subsection{\label{subsec:model-setup}Model Setup}

We consider a highly stylized version of SPRIG. The maximal level
is two and there are only two agents: Claimer and Skeptic. Both are
risk-neutral and do not discount the future. Claimer is endowed with
a claim of proof $C$ (of some statement). 

We say that the claim is validatable if it is possible to unroll it
down to machine level and that it is accepted if either Skeptic renounces
challenging or the claim is indeed unrolled down to machine level. 

Claimer initially receives a (random) signal $P\in[0,1]$, uniform
on $[0,1]$, such that 
\begin{equation}
\mathbb{E}[X\lvert P]=P
\end{equation}
where $X=\mathbf{1}_{\left\{ C\text{ is validatable}\right\} }$ (where
$\mathbf{1}_{A}\left(x\right)=1$ if $x\in A$ and $\mathbf{1}_{A}\left(x\right)=0$
if $x\notin A$). We use the economic term signal to refer to a random
variable whose realization can be informative about the variable $X$
of interest. Here, we could define $U$ to be an independent copy
of $P$ and assume that $C$ is validatable exactly on the event $\left\{ U\leq P\right\} $.
In words, Claimer has more information than Skeptic, as she knows
an updated probability, the realization of $P$, that her claim of
proof is validatable. By contrast, Skeptic initially only has the
knowledge that $P$ is uniformly distributed over $[0,1]$. 

On top of the potential collection of bounties, Claimer derives private
benefits from having her claim of proof accepted. We denote by $B_{2},B_{1},B_{0}\ge0$
the benefits of being accepted at level $2,1,0$ respectively. 

If Claimer decides not to post $C$, the game ends immediately, and
both players receive a payoff of $0$. If she posts $C$, the protocol
specifies a stake $\sigma_{2}^{\downarrow}$ to be collected by Skeptic
in case of a successful challenge. If $C$ remains unchallenged (no
questions are asked about it within time $\theta_{2}$ after publication),
then Claimer gets $B_{2}$, and Skeptic gets $0$. If Skeptic challenges
the claim of proof within time $\theta_{2}$, staking a bounty $\beta_{1}$,
we make the assumption that Claimer gets to know the realization of
$X$ and will be able to provide a machine-level proof of her claim
if valid at level $0$ of the protocol. Given this realization, she
decides whether or not to post a claim at level $1$ within time $\tau_{1}$
after the publication of Skeptic's question. If she does not, her
final payoff if $-\sigma_{2}^{\downarrow}$ and Skeptic's is $\sigma_{2}^{\downarrow}$.
If she does, the protocol specifies two stakes $\sigma_{1}^{\uparrow},\sigma_{1}^{\downarrow}$.
Skeptic has a last chance to challenge the claim within time $\theta_{1}$
after its publication: if he does not, Claimer's payoff is $B_{1}+\beta_{1}$
and Skeptic's is $-\beta_{1}$. If he does, he stakes a bounty $\beta_{0}$
and then  Claimer posts the machine language proof if available within
time $\tau_{0}$ after publication of Skeptic's question. Claimer's
final payoff is $-\sigma_{2}^{\downarrow}-\sigma_{1}^{\uparrow}-\sigma_{1}^{\downarrow}$
if she can not provide a machine language proof, and $B_{0}+\beta_{0}+\beta_{1}$
if she can. The corresponding Skeptic's payoffs are $\sigma_{2}^{\downarrow}+\sigma_{1}^{\uparrow}+\sigma_{1}^{\downarrow}$
and $-\beta_{0}-\beta_{1}$. Since we consider a single skeptic, the
recipient of the up and down stakes $\sigma_{1}^{\uparrow},\sigma_{1}^{\downarrow}$
is the same and hence we can aggregate those in $\sigma_{1}=\sigma_{1}^{\uparrow}+\sigma_{1}^{\downarrow}$.
From now on, we also denote $\sigma_{2}=\sigma_{2}^{\downarrow}$.
The game is represented in Figure \ref{eco:fig:game_tree}. 

The time lengths $\theta_{1},\theta_{2},\tau_{0},\tau_{1}$ are key
to make sure that the status (`challenged or not') of a claim/question
is eventually settled but their value does not play a role in our
stylized model. Hence, our game is fully characterized by the parameter
set 

\[
\Theta=\left\{ B_{0},B_{1},B_{2},\sigma_{1},\sigma_{2},\beta_{0},\beta_{1}\right\} .
\]

\begin{figure}
\hspace*{-35pt}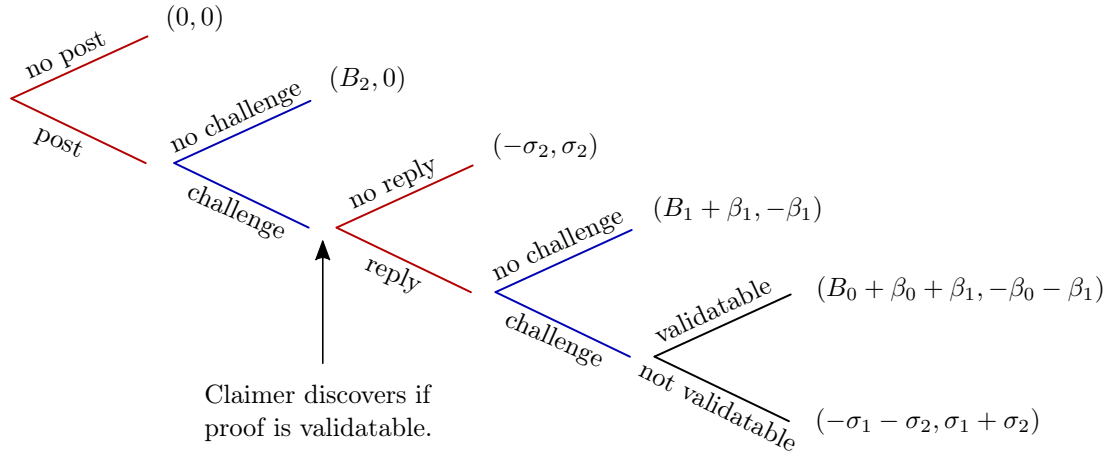
\caption{Simplified protocol game tree}
\label{eco:fig:game_tree}
\end{figure}

\subsection{\label{subsec:model-solution}Model Solution}
\begin{prop}
\label{prop:game} The simplified protocol game possesses a unique
Perfect Bayesian Equilibrium. Depending on the parameter set $\Theta$,
this PBE takes one of the three types detailed below. In all of them,
there is a threshold $\pi^{*}:=\pi^{*}(\Theta)\in[0,1)$ such that
Claimer posts if and only if $P\geq\pi^{*}$ and: 
\begin{itemize}
\item Type 1: Skeptic always challenges, Claimer replies if $X=1$ and replies
with probability $p:=p(\Theta)$ if $X=0$. Conditional on reply,
Skeptic challenges w.p. $q_{1}:=q_{1}(\Theta)$ (and Claimer successfully
terminates the process if and only if $X=1$). 
\item Type 2: Skeptic challenges the initial claim with probability $q_{2}:=q_{2}(\Theta)\in(0,1)$.
Then actions unfold as in Type 1. 
\item Type 3: $\pi^{*}=0$ and Skeptic never challenges. 
\end{itemize}
\end{prop}

The parameters $q_{1},q_{2},p$ and $\pi^{*}$ are known explicitly
and their values are provided in Appendix \ref{eco:proof:prop1}.
The proof of the proposition can be found in the same Appendix. For
illustrations on how the nature of the equilibrium (type 1, 2, or
3) depends on the parameters, see Section \ref{subsec:results}. 

\subsection{Extracting Relevant Information from the Protocol's Outcome}

One key appeal of our model is that it allows us to compute various
measures of protocol reliability, in particular the likelihood that
type I or type II errors (in the statistics sense) occur. This is
crucial because the outcome of the protocol's validation process for
a claim (accepted/validated or rejected/invalidated) does not say,
in isolation, what credence the agents' community should ha\textcolor{black}{ve
in the cl}aim. Section \ref{subsec:results} discusses further these
issues.

\subsubsection{Notation}

We shall need the following notation: 
\begin{itemize}
\item $\mathbf{A}$ (resp. $\mathbf{A}^{c}$) is the event `The claim is
accepted' (resp. rejected, i.e. not accepted). 
\item $\mathbf{A}_{2}$ (resp. $\mathbf{A}_{1}$) is the event `The claim
is accepted at level $2$', i.e. no question was asked (resp. level
$1$, i.e one question was asked). 
\item $\mathbf{Q}_{0}$ is the event `Skeptic challenges Claimer after the
reply of Claimer', i.e. Skeptic posts a question at level $0$. 
\item $\mathbf{R}$ is the event `Claimer replies to first challenge' (i.e.
posts a claim at level $1$). 
\end{itemize}

\subsubsection{Results}
\begin{prop}
\label{prop:prob-extraction} In a Type 1 equilibrium: 
\begin{itemize}
\item The probabilities that a claim of proof is accepted (resp. accepted
and valid) are 
\begin{eqnarray}
\mathbb{P}(\mathbf{A}) & = & \frac{1}{2}\left(1-\pi^{*}\right)\left(1+\pi^{*}+(1-\pi^{*})p(1-q_{1})\right)\\
\mathbb{P}(\mathbf{A},X=1) & = & \frac{1}{2}\left(1+\pi^{*}\right)\left(1-\pi^{*}\right).
\end{eqnarray}
\item The probabilities that a claim of proof is accepted given that it
is valid (resp. false) are 
\begin{eqnarray}
\mathbb{P}(\mathbf{A}\lvert X=1) & = & \left(1+\pi^{*}\right)\left(1-\pi^{*}\right)\\
\mathbb{P}(\mathbf{A}\lvert X=0) & = & \left(1-\pi^{*}\right)^{2}p(1-q_{1}).
\end{eqnarray}
\item The probabilities that a claim of proof is valid given that it is
accepted (resp. rejected) are 
\begin{eqnarray}
\mathbb{P}(X=1\lvert\mathbf{A}) & = & \frac{1+\pi^{*}}{1+\pi^{*}+\left(1-\pi^{*}\right)p(1-q_{1})}\\
\mathbb{P}(X=1\lvert\mathbf{A}^{c}) & = & \frac{\pi^{*2}}{\pi^{*2}+1-\left(1-\pi^{*}\right)^{2}p(1-q_{1})}.
\end{eqnarray}
\item The probabilities that a claim of proof is accepted at level 2 (resp.
1) given that it is accepted and valid are: 
\begin{eqnarray}
\mathbb{P}(\mathbf{A}_{2}\lvert\mathbf{A},X=1) & = & 0\\
\mathbb{P}(\mathbf{A}_{1}\lvert\mathbf{A},X=1) & = & 1-q_{1}.
\end{eqnarray}
Such expressions can also be derived in the case of Type 2 and Type
3 equilibria: see the proof. 
\end{itemize}
\end{prop}

\subsection{Results\label{subsec:results}}

We now use our model solution to explore various trade-offs faced
by protocol designers. Obviously, the fact that we consider both a
stylized model and a simplified information structure does not allow
us to produce general and quantitative positive or normative statements
about SPRIG parameters. However, our sequential game with imperfect
information is rich enough to illustrate several important forces
that must be taken into account by designers, and that can be illustrated
through examples. As baseline parameters, we consider $B_{2}=10,B_{1}=B_{0}=40$
and $\beta_{1}=\sigma_{1}=\beta_{0}=5$, and let the stake $\sigma_{2}$
vary. Of course, one could evidence similar trade-offs by varying
another parameter, as well as the transitions between the different
equilibrium types. We focus on varying $\sigma_{2}$ merely for the
sake of brevity.

\subsubsection{Stakes and bounties, entry and reliability ratio}

The two following properties are desirable for the protocol: 
\begin{itemize}
\item (i) Have as many correct claims of proof as possible passing through
the protocol.
\item (ii) Have a (very) high probability that an accepted claim of proof
indeed corresponds to a proof (i.e. is correct). The left panel of
Figure \ref{eco:fig:fig1} evidences that the two objectives are,
in general, in conflict with each other. In this plot, the solid line
depicts the probability that a claim is produced and accepted by the
protocol, while the dashed line depicts the probability that a correct
claim is produced and accepted by the protocol. Define the reliability
ratio $RR$ as the ratio of the latter by the former (`dashed/solid').
\end{itemize}
$RR$ is close to the desirable 100\% as long as the equilibrium is
of Type 1, and quickly deteriorates as we enter the Type 2 equilibrium
region. While there is no ideal conciliation of Objectives (i) and
(ii) above, the left panel of Figure \ref{eco:fig:fig1} suggests
that a good way to resolve the trade-off is to select a stake $\sigma_{2}$
just barely sufficient to incentivize Skeptic to systematically challenge.
This maximizes the probability of having a claim going successfully
through the protocol among Type 1 equilibria. Of course, by reducing
$\sigma_{2}$ further, one could increase this probability further,
but the `price' to pay (the quick drop of $RR$) is likely to be prohibitive.

The right panel of Figure \ref{eco:fig:fig1} indicates that $\pi^{*}$,
the equilibrium entry threshold, increases with the bounty $\sigma_{2}$.
This is consistent with intuition: if she must pay a large amount
in case of a successful challenge, Claimer will only enter when she
is very confident about her claim of proof. Hence, increasing $\sigma_{2}$
reduces entry; but it also increases the likelihood that a claim of
proof is true conditional on entry. Thus, the impact of $\sigma_{2}$
on $\mathbb{P}(A)$ was a priori non-trivial. The left panel of Figure
\ref{eco:fig:fig1} indicates that there is a monotone decreasing
relationship between the two variables. In fact, this is always true,
as one can easily deduce from the formulas of Proposition \ref{prop:prob-extraction}.

\begin{figure}[H]
\begin{centering}
\includegraphics[width=220pt]{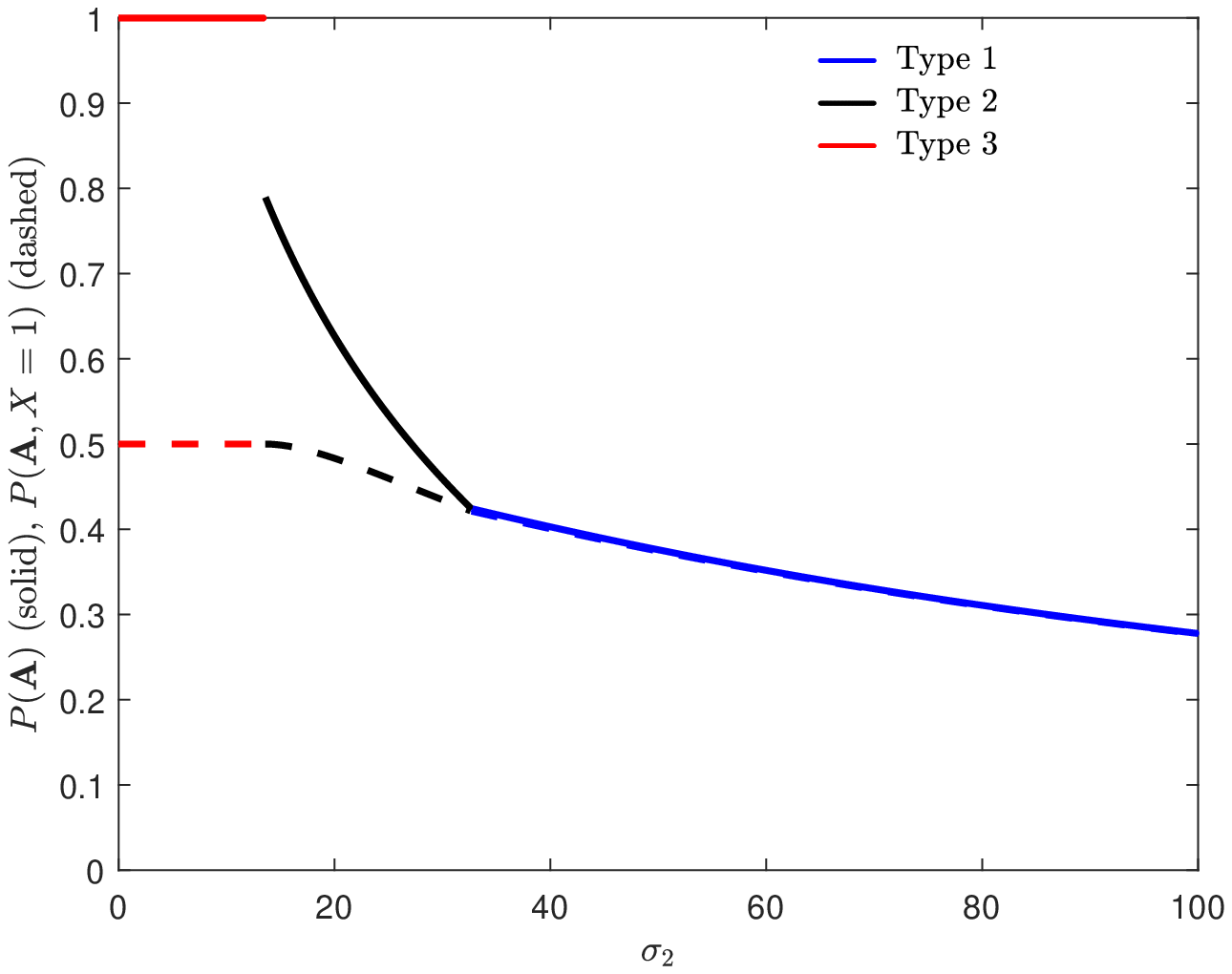} \hspace{20pt}\includegraphics[width=220pt]{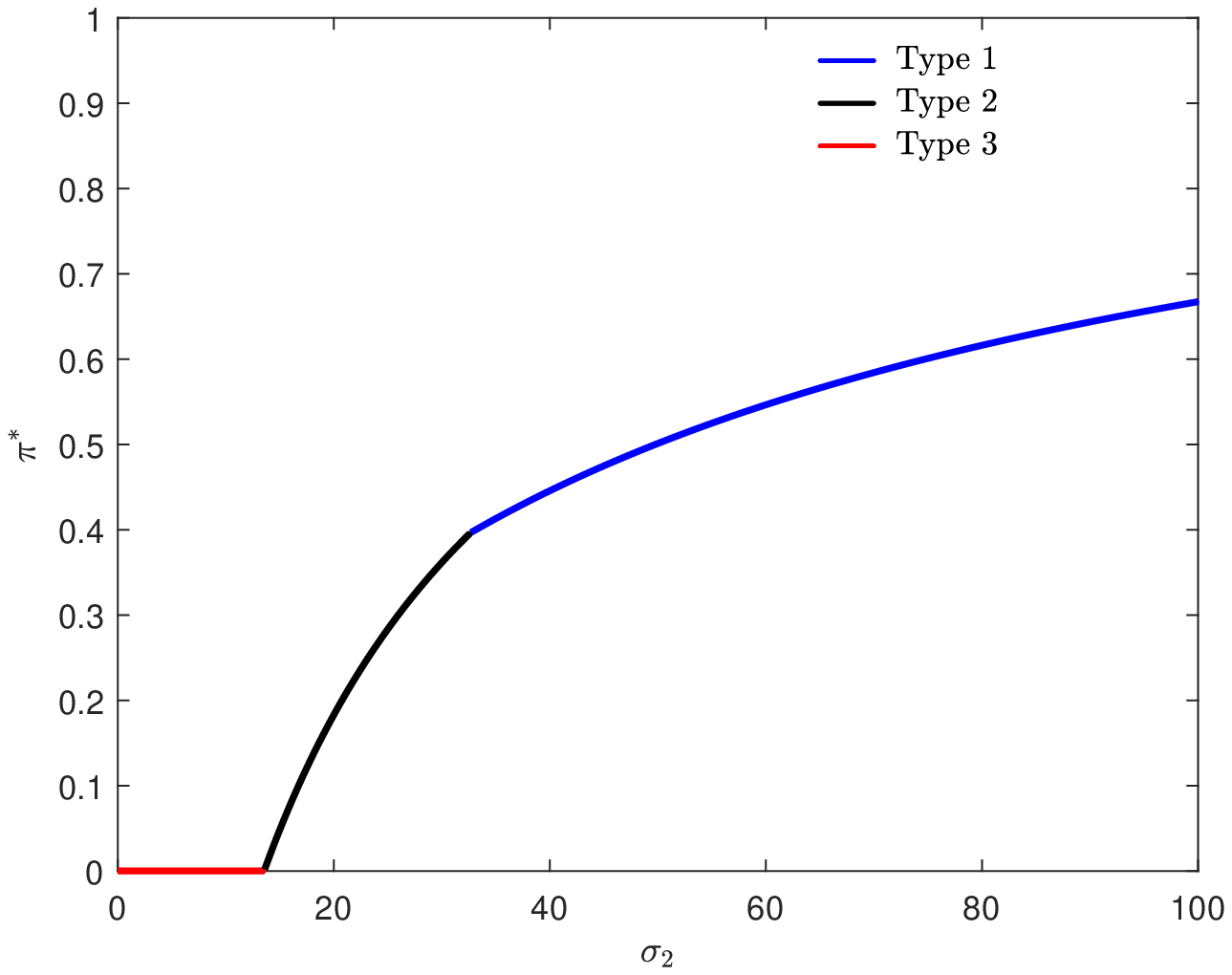}
\par\end{centering}
\centering{}\caption{Bounties, reliability and entry.}
\label{eco:fig:fig1}
\end{figure}

\subsubsection{Statistical Type I and Type II Errors\label{subsec:statistical-type-1-and-type-2-errors} }

In this section, we focus on the four quantities $\mathbb{P}(\mathbf{A}^{c}\lvert X=1)$,
$\mathbb{P}(\mathbf{A}\lvert X=0)$, $\mathbb{P}(X=0\lvert\mathbf{A})$
and $\mathbb{P}(X=1\lvert\mathbf{A}^{c})$. All are measures of the
likelihood that the protocol produces an undesirable outcome (at least
from a scientific standpoint, as a claimer would presumably have no
problem with having an incorrect claim accepted). The last two correspond
to the standard statistical notions of Type I and Type II errors,
respectively. The reliability ratio $RR$ introduced above is simply
the complement to the probability of a Type I error.

Being able to compute such measures is of paramount importance for
protocol users and for the mathematical community at large. Without
them, there is no clear link between the outcome of the validation
process and the credence that humans shou\textcolor{black}{ld give
to} a claim (or its negation). In particular, humans may want to lower
their confidence in claims accepted in some particular equilibrium
type and state. As an illustration, consider the right panel of Figure
\ref{fig:prob-false-positives}. If the equilibrium is of Type 3,
all proofs are accepted, but this is irrelevant from a scientific
standpoint, as the probability of Type I error is $\frac{1}{2}$.
If stakes and bounties are designed in such a way that challenging
is prohibitively expensive, one should not give too much credit to
a claim simply because it has passed through the protocol. Such a
design would be severely flawed. More generally, given a model that
allows one to predict the equilibrium type, one can and should observe
the blockchain in order to refine the statement `the claim has been
accepted' into `the claim has been accepted at level $d$ after history
$h$' and update the correctness probability accordingly.

While the right panel of Figure \ref{fig:prob-false-positives} illustrates
that there are some entirely flawed protocol designs (if it generates
a Type 3 equilibrium) it also highlights that there is no perfect
design: one cannot simultaneously decrease the likelihood of Type
I and Type II errors. Again, the juncture point between Type 1 and
Type 2 equilibria seems to be a good candidate: for instance, any
choice of a larger $\sigma_{2}$ would only very marginally decrease
the probability that an accepted claim is incorrect, but significantly
increase the probability that a rejected claim is correct.

Arguably, the aggregate costs of accepting invalid claims are much
larger than the costs of rejecting correct ones. Indeed, once accepted,
an invalid claim could be used repeatedly in subsequent research or
applications, so that the mistake propagates and its consequences
grow. Moreover, there might not be enough incentives or reasons to
challenge the claim again in the future. By contrast, a wrongly rejected
claimer could always rewrite her claim of proof, improve communication
and post again at a later stage, getting another chance to be accepted.

The left panel of Figure \ref{fig:prob-false-positives} tells a similar
story, with a perspective closer to the point of view of the claimer.
Reducing the risk of rejecting a correct claim of proof increases
the risk of accepting a wrong claim of proof. Once again, thinking
about the aggregate costs of both types of error should allow designers
to select their preferred parameters.

As can be seen, producing graphs such as those of Figure \ref{fig:prob-false-positives}
gives a lot of information about which parameter values are the most
effective. This will be precious for future (more realistic and quantitatively
accurate) model descriptions of the protocol.

\begin{figure}[H]
\begin{centering}
\includegraphics[width=220pt]{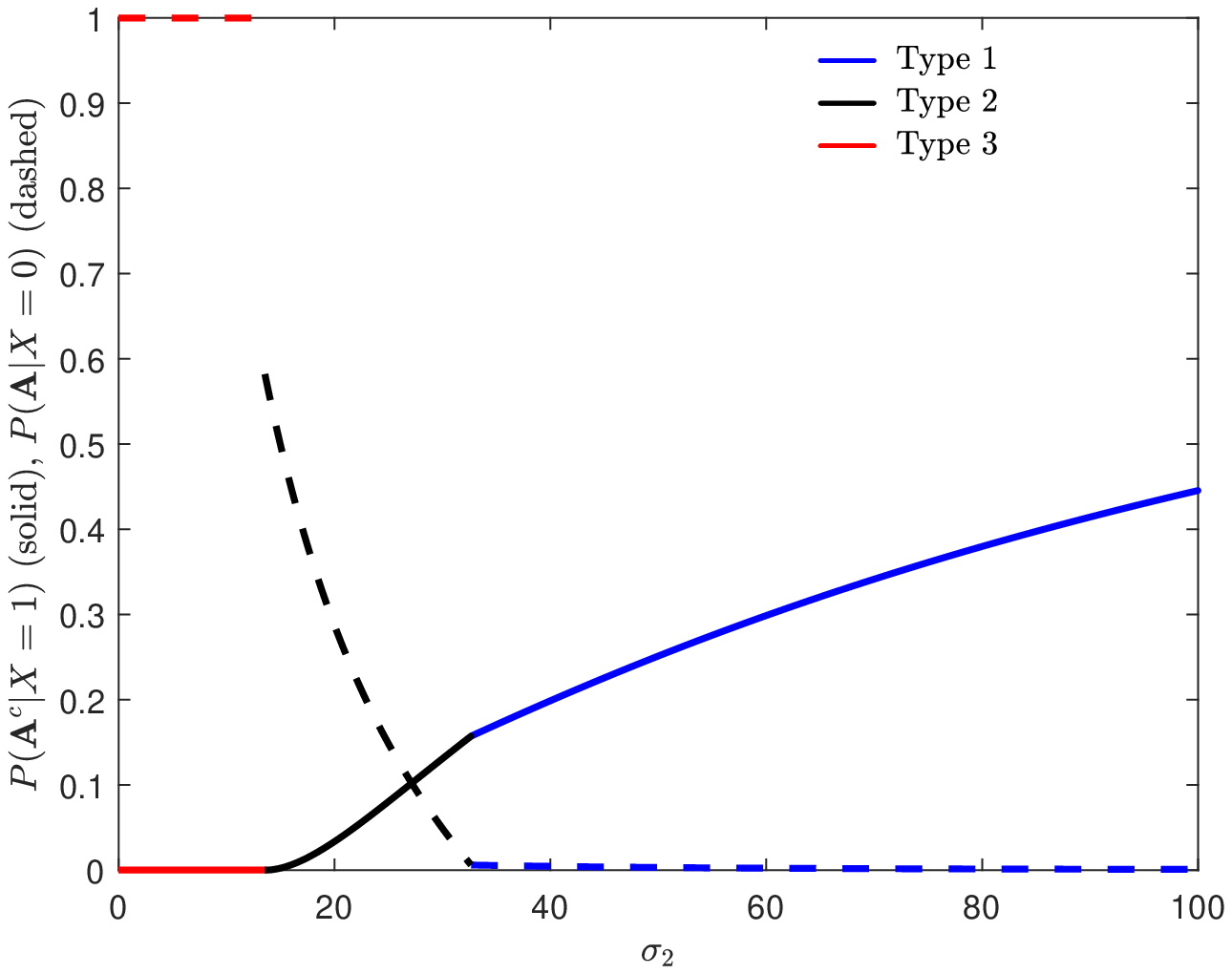} \hspace{20pt}\includegraphics[width=220pt]{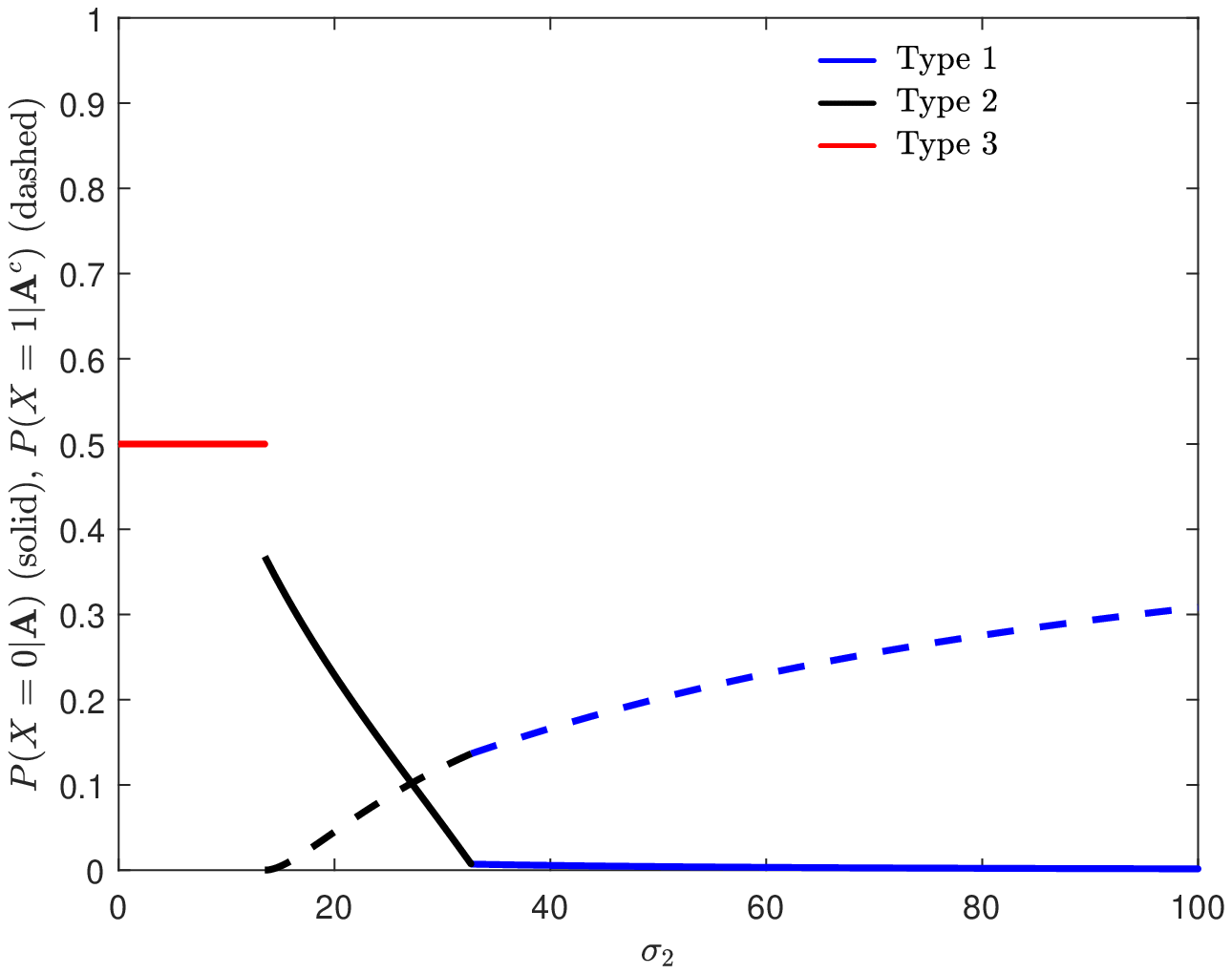}
\par\end{centering}
\centering{}\caption{Probability of false positives and related reliability measures}
\label{fig:prob-false-positives}
\end{figure}

\subsubsection{Terminations at Intermediate Level\label{subsec:terminations-at-intermediate-depth}}

Fixing the likelihood of statistical errors discussed above, a short
acceptance process (termination after a few steps) has advantages
and drawbacks. On the one hand, accepting correct\textcolor{black}{{}
claims of proof quickly saves significant time and intellectual energy
that can be invested in tackling other problems. On the} other hand,
longer acceptance processes can have positive externalities, as they
involve clarifying steps and new lemmas that could be useful in other
contexts. In the current discussion, we wish to focus on the former
point and consider that accepting a correct proof quickly is desirable---to
the extent of course that it does not harm its credibility too much,
see Section \ref{subsec:statistical-type-1-and-type-2-errors}. But
the latter requirement is key: in our model, claims accepted immediately
after posting have little scientific relevance. Indeed, we saw that
the reliability ratio quickly decreases away from 100\% as the likelihood
to terminate immediately increases away from 0. Hence, a good proxy
for `having proofs that terminate before final level without harming
reliability' is the probability of termination after exactly 1 step,
$\mathbb{P}(\mathbf{A}_{1}\lvert\mathbf{A},X=1)$. This quantity is
depicted in Figure \ref{fig:prob-correct-termination-at-intermediay-depth}.

Again, the behaviour of this quantity as a function of $\sigma_{2}$
is non-trivial. Indeed a larger stake (i) increases the incentives
to challenge (direct `greed' effect) but also (ii) decreases them
as the average quality of a proof is higher (this is evidenced by
the fact that $\pi^{*}$ is larger). As before, a good choice seems
to take the lowest $\sigma_{2}$ that implements a Type 1 equilibrium.

\begin{figure}[H]
\centering{}\includegraphics[width=220pt]{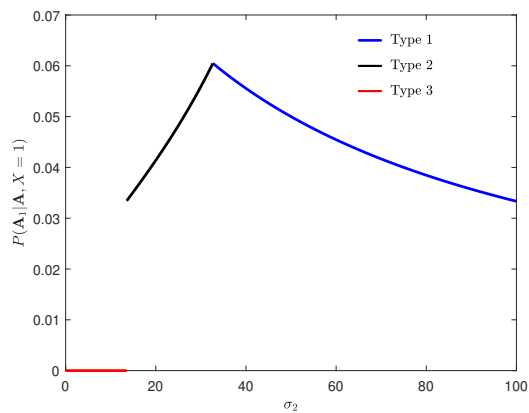} \caption{Probability of correct termination at intermediary level.}
\label{fig:prob-correct-termination-at-intermediay-depth}
\end{figure}

\subsection{\label{subsec:discussion-of-the-models-analysis}Discussion of the
Model's Analysis}

\label{eco:sec:discussion} Our stylized model captures several important
features. First, we take into account the likely information asymmetry
that exists between claimers and skeptics (at least at the time a
claim is posted). Then, the theory of signaling games allows us to
predict how the mere fact of posting a claim impacts the `market's
beliefs'. Second, our simplified game is dynamic. In particular, we
can understand the impact of stakes and bounties at further levels
on current decisions, and get estimates of the probability that the
protocol stops before the final, machine-language step. Third, our
model is rich enough to inform a Bayesian agent about the correctness
of a claim, using its status in the protocol (accepted vs rejected).

The model can be enriched in several directions. First, Skeptic could
be able to perform some work on his own: by paying some cost, he could
be able to access a private signal about the correctness of Claimer's
proof. We have studied this possibility in a one-period version of
the model. New insights appear, but the enlargement of Skeptic's strategy
set also implies the emergence of multiple equilibria. Those are potentially
interesting but render predictions difficult. Second, the information
structure could be enriched: there could be several skeptics (and
several subsequent claimers replying  to these skeptics), each potentially
endowed with their own information. A challenging and very interesting
question is to understand how information dynamically incorporates
in such a context. Third, our stylized model is not able to answer
questions such as whether we should expect SPRIG instances to generally
terminate at the top level or the machine level, or instead in between,
or regarding the structure of the tree that an initial claim or question
generates.

\section{\label{sec:applications-and-outlooks}Applications and Outlook}

In this section, we discuss how the SPRIG protocol provides a solution
to the challenges of mathematical derivation raised in Section \ref{sec:introduction},
which are centered around the communication of trustable, succinct,
and informative proofs in a system with agents with various levels
of information. 

\subsection{\label{subsec:theorem-verification}Theorem Verification}

The validation of a theorem's proof by authors can be done through
a protocol: they can place a stake (which may increase over time,
as discussed in \ref{subsec:time-varying-stakes-and-bounties}) and
set up a SPRIG instance for a given amount of time, incentivizing
anyone to find a gap in their proof. Compared to the classical publishing
model, many more agents are incentivized to be skeptical of the proof
(and no one is pressured to participate either), and their questions
can be assumed to be made in good faith (since there is nothing to
gain by asking trivial questions); also, the anonymity of the reviewers
is guaranteed (unlike the reviewing process, which consists in the
redaction of a report, and which relies on an editorial board, both
of which may leak information). The results of the validations can
thus be made transparent and convey information about the validation
of theorems. At the same time, as they provide their claim of proof,
the authors can also publish a paper written in an informal way, which
may help the community participate in the process more rationally. 

\subsection{\label{subsec:bounty-for-open-problem}Bounty for Open Problem}

The research on open questions in mathematics can be incentivized
by bounties, such as the celebrated Millennium Problems posted by
the Clay Institute. In cases such as that of the Millennium Problems,
an open two-sided question is at the root of the problem, as discussed
in Section \ref{subsec:open-questions-and-multi-question-bounties}. 

SPRIG allows one to outsource the validation of claims of proof (which
in principle relies on a committee), to disincentivize bogus claims
of proof (a stake must be put to propose a claim of proof), and to
limit conflicts of interest. 

Incentives for shorter answers can also be added, by creating extra
challenges, with tighter limits on proof length, or by using claim-of-proof-dependent
parameters (as in Section \ref{subsec:claim-of-proof-dependent-parameters}).

\subsection{\label{subsec:security-proof-certification}Security Proof Certification}

An organization may want to elicit trust in its system. For instance,
it may want to publish their source and incentivize the public to
find security flaws in any of $N$ subsystems. It may have a limited
number $K$ of bounties available for finding problems in any of the
$N$ subsystems. This may be done using the multi-question bounty
variant discussed in Section \ref{subsec:open-questions-and-multi-question-bounties},
and elicit a trust in the system (e.g. that there is no flaw in any
of the subsystems) that is as strong as if there were $N$ bounties,
while at the same time locking up and risking only $K$ bounties worth
of capital. 

\subsection{\label{subsec:automated-theorem-proving}Automated Theorem Proving}

A great deal of effort has been put in recent years into constructing
intelligent automated provers, relying on e.g. reinforcement learning
techniques or text prediction mechanisms, with encouraging successes
\cite{urban-jakubuv}. SPRIG can serve as a playground for the development
of such agents, allowing them to participate using some level of information
(in particular by first developing an ability to write low-level proofs
or to validate them), and learning by playing.

\subsection{\label{subsec:derivatives-markets}Derivatives Markets}

A promising feature of SPRIG is that its outcomes can then be used
as oracles for other smart contracts. In particular, other prediction
markets can run on such outcomes. 

For instance, agents can inject information by betting that a certain
question will or will not be answered before a certain time. Or they
could bet that conditionally on there being an unanswered question,
this question will challenge a specific step $S$ of the proof, thereby
indicating that $S$ might be the weak link.

Securities markets relying on SPRIG may also prove to be useful for
incentivizing different types of contributions for agents. For instance,
an agent able to provide good formalizable heuristics but not knowing
how to formalize them may participate in a market betting that a certain
question will be answered before a certain time, buy (for relatively
cheap) a security betting that it will be answered, and then publish
her heuristics; if it looks like the heuristics can be formalized
by some agent in time, the odds for betting that the question will
be answered will change, and she can net a profit by re-selling her
security or letting it mature. Thus, she can inject interesting information
into the market, i.e. information that changes the feasibility landscape
of proof construction by the community. 

\subsection{\label{subsec:beyond-mathematical-reasoning}Beyond Mathematical
Reasoning}

Beyond mathematics, many fields rely on rigorous formal reasoning
intertwined with external elements of reasoning. Adding support for
external sources for SPRIG appears promising for a number of applications:
\begin{itemize}
\item Support for importing empirical knowledge in the protocol: this would
allow it to submit and verify arguments pertaining to experimental
sciences.
\item Support for numerically-justified heuristic arguments or recognized
heuristics: this would allow for useful derivations in e.g. theoretical
physics.
\item Support for validated time-stamped predictions: this could help rational
discourse in disciplines based on forecasting. An economic model could
be presented and challenged similarly t\textcolor{black}{o claims
of proof} in SPRIG. In lieu of the machine-level terminal condition,
the final validation step would be given by the publication of official
numbers. A `claim' (i.e. a model) would then be validated if it has
correctly predicted a n-tuple of economic variables (e.g. interest
rate set by the Fed, GDP) up to a prespecified error margin.
\item Support for oracles with zero-knowledge proofs: this would allow for
auditable arguments in public debates in which certain sources must
be protected. 
\end{itemize}

\section{Conclusion}

In this paper, we introduced the Smart Proofs by Recursive Information
Gathering (SPRIG) protocol, which allows agents to propose and verify
succinct and informative proofs in a decentralized fashion. Claimers
and skeptics `debate' about statements and their proofs: consensus
arises from the skeptics being able to request details on steps that
they feel could be problematic and from the claimers being able to
provide details answering the skeptics' requests. Importantly, to
participate in the process, claimers and skeptics must attach a bounty/stake
to their moves: this gives the proper incentive for subsequent users
to verify those. As a result, agents with various types of information
can participate and inject their knowledge into the proof construction
and verification process; this allows one to strike a balance between
the `short collection of insightful statements' vs `list of all the
statements needed to establish perfect trust' tradeoff in mathematics
writing. 

In our claim of proof format, mathematical proofs can be viewed as
trees, in which claimers and skeptics can expand branches containing
the relevant level of detail for the agents in the community: branches
only grow in places where there is uncertainty, until either that
uncertainty is cleared or a specific problem is isolated. This resulting
subtree thus serves as a proof that is useful to the community, as
it makes the consensus-building process transparent and can help agents
build their own credence in the validity of the proof. 

Our analysis of SPRIG and its robustness is based on game-theoretic
considerations that take into account the various incentives of the
agents, address possible attacks, and leading up to a detailed equilibrium
analysis of a simplified protocol. While the complete SPRIG protocol
is very complex to study analytically, our results give a clear insight
into a number of qualitative aspects of its strategic features. 

We also present a number of variants and applications of SPRIG, allowing
it to be useful in numerous contexts, and demonstrating its versatility. 

\section*{Acknowledgements}

The authors would like to thank Tarun Chitra for enlightening and
inspiring explanations about blockchains and many other topics, Thibaut
Horel for numerous useful suggestions about the present manuscript,
as well as Juhan Aru, Dmitry Chelkak, Fedor Doval, Julien Fageot,
Patrick Gabriel, Max Hongler, Kalle Kytölä, Reda Messikh, Justin Neumann,
Christophe Nussbaumer, Daniele Ongari, Victor Panaretos, Stanislav
Smirnov, Fredrik Viklund, Jérémie Wenger, and Matthieu Wyart for interesting
conversations. We also thank Kevin Buzzard for a very much appreciated
feedback on an earlier version of this manuscript. 

C.H. acknowledges support from the Blavatnik Family Foundation and
the Latsis Foundation.

\section*{E-mail Addresses}

Sylvain Carré: \texttt{sylvain@sprig.ch}

Franck Gabriel: \texttt{franck@sprig.ch}

Clément Hongler:\texttt{ clement@sprig.ch}

Gustavo Lacerda: \texttt{gustavo@sprig.ch}

Gloria Capano: \texttt{gloria@sprig.ch}

\newpage{}

\newpage{}

\specialsection{\label{sec:appendix-claim-of-proof-examples}Appendix: Claim of Proof
Examples}

In this appendix, we give a number of examples of proofs written in
the claim of proof structure described in Section \ref{subsec:structured-proofs}
(we naturally only give subtrees of the entire trees). We denote by
$\mathbf{A}_{0}$ the standard background assumptions (including basic
axioms and whatever statement is taken for granted). 

\subsection{\label{subsec:a-simple-proof}A simple proof}

A simple proof concerns the existence of an infinite numer of primes.

In this case, we have $\gamma=\left\{ \text{there exists an infinite number of primes}\right\} $
and the root of the claim of proof is
\begin{itemize}
\item $\mathbf{S}_{*}:\mathbf{A}_{*}\implies\mathbf{C}_{*}$, with $\mathbf{A}_{*}=\mathbf{A}_{0}$
and $\mathbf{C}_{*}=\gamma$. 
\end{itemize}
The nodes at distance $1$ from the root are:
\begin{itemize}
\item $\mathbf{S}_{1}:\mathbf{A}_{1}\implies\mathbf{C}_{1}$, where $\mathbf{A}_{1}=\mathbf{A}_{0}$
and $\mathbf{C}_{1}$ corresponds to `For any $N\geq2$, $N!+1$ is
not divisible by any $k\in\mathbb{N}$ with $2\leq k\leq N$';
\item $\mathbf{S}_{2}:\mathbf{A}_{2}\implies\mathbf{C}_{2}$, where $\mathbf{A}_{2}=\mathbf{A}_{0}\cup\left\{ \mathbf{C}_{1}\right\} $
and $\mathbf{C}_{2}$ corresponds to `For any $N\geq2$, there exists
a prime number $p>N$'.
\item $\mathbf{S}_{3}:\mathbf{A}_{3}\implies\mathbf{C}_{3}$, where $\mathbf{A}_{3}=\mathbf{A}_{0}\cup\left\{ \mathbf{C}_{2}\right\} $
and $\mathbf{C}_{3}=\mathbf{C}_{*}$.
\end{itemize}
If we expand the proof of $\mathbf{S}_{2}$ (at distance $2$ from
the root), we find: 
\begin{itemize}
\item $\mathbf{S}_{2,1}:\mathbf{A}_{2,1}\implies\mathbf{C}_{2,1}$, where
$\mathbf{A}_{2,1}=\mathbf{A}_{2}$ and $\mathbf{C}_{2,1}$ corresponds
to `For any $N\geq2$, any prime factor of $N!+1$ is larger than
$N$'.
\item $\mathbf{S}_{2,2}:\mathbf{A}_{2,2}\implies\mathbf{C}_{2,2}$, where
$\mathbf{A}_{2,2}=\mathbf{A}_{2}\cup\left\{ \mathbf{C}_{2,1}\right\} $
and $\mathbf{C}_{2,2}=\mathbf{C}_{2}$. 
\end{itemize}

\subsection{\label{subsec:a-proof-by-contradiction}A proof by contradiction}

Proofs by contradiction can be formulated naturally in our framework.
A classical proof by contradiction is that of the fundamental theorem
of algebra $\alpha\implies\gamma$, where
\begin{itemize}
\item $\alpha$ corresponds to `$P$ is a complex polynomial of degree $\geq1$'.
\item $\gamma$ corresponds to `there exists $z\in\mathbb{C}$ such that
$P\left(z\right)=0$'.
\end{itemize}
In this case, the root of the claim of proof is
\begin{itemize}
\item $\mathbf{S}_{*}:\mathbf{A}_{*}\implies\mathbf{C}_{*}$, where $\mathbf{A}_{*}=\mathbf{A}_{0}\cup\left\{ \alpha\right\} $,
and $\mathbf{C}_{*}=\gamma$. 
\end{itemize}
The nodes at distance $1$ from the root are:
\begin{itemize}
\item $\mathbf{S}_{1}:\mathbf{A}_{1}\implies\mathbf{C}_{1}$, where $\mathbf{A}_{1}=\mathbf{A}_{*}$
and $\mathbf{C}_{1}$ corresponds to `There exists $M,R>0$ such that
$\left|P\left(z\right)\right|\geq M$ for all $\left|z\right|\geq R$';
\item $\mathbf{S}_{2}:\mathbf{A}_{2}\implies\mathbf{C}_{2}$, where $\mathbf{A}_{2}=\mathbf{A}_{*}$
and $\mathbf{C}_{2}$ corresponds to `$\left|P\right|$ cannot have
a nonzero minimum on $\mathbb{C}$'; 
\item $\mathbf{S}_{3}:\mathbf{A}_{3}\implies\mathbf{C}_{3}$, where $\mathbf{A}_{3}=\mathbf{A}_{*}\cup\left\{ \mathbf{C}_{1},\mathbf{C}_{2}\right\} $
and $\mathbf{C}_{3}=\mathbf{C}_{*}$.
\end{itemize}
If we go further into the details of the proof of $\mathbf{S}_{2}:\mathbf{A}_{2}\implies\mathbf{C}_{2}$
(the heart of the proof by contradiction), we have (at distance $2$
from the root):
\begin{itemize}
\item $\mathbf{S}_{2,1}:\mathbf{A}_{2,1}\implies\mathbf{C}_{2,1}$ where
$\mathbf{A}_{2,1}=\mathbf{A}_{2}$ and $\mathbf{C}_{2,1}$ corresponds
to `If $P\left(z\right)\neq0$, then there exists $z'\in\mathbb{C}$
such that $\left|P\left(z'\right)\right|<\left|P\left(z\right)\right|$'; 
\item $\mathbf{S}_{2,2}:\mathbf{A}_{2,2}\implies\mathbf{C}_{2,2}$, where
$\mathbf{A}_{2,2}=\mathbf{A}_{2,1}\cup\left\{ \mathbf{C}_{2,1}\right\} $
and $\mathbf{C}_{2,2}$ corresponds to `If $\left|P\right|$ has a
nonzero minimum on $\mathbb{C}$, then we have a contradiction';
\item $\mathbf{S}_{2,3}:\mathbf{A}_{2,3}\implies\mathbf{C}_{2,3}$, where
$\mathbf{A}_{2,3}=\mathbf{A}_{2}\cup\left\{ \mathbf{C}_{2,2}\right\} $
and $\mathbf{C}_{2,3}=\mathbf{C}_{2}$.
\end{itemize}
\begin{figure}
\begin{center}
\resizebox{0.9\textwidth}{!}
{\begin{tikzpicture}[nodes={draw, rectangle}, ->]   \node [align=left,label=left:$T$] (a) at (0,0) {     \hspace{28mm} $\alpha \implies \gamma$\\     $A_* = \alpha$ : $P$ is a complex polynomial of degree $\ge$ 1\\     $C_* = \gamma$ : $\exists z \in \mathbb{C}$ s.t. $P(z) = 0$};   \node [align=left,label=left:$S_1$] (b) at (-6,-4) {     $A_1 = A_*$\\     $C_1$ : $\exists M,R > 0$ s.t.\\     \hspace{5mm} $\left|P(z)\right| \ge M$ for all $\left| z \right| \ge R$   };   \node [align=left,label=left:$S_2$] (c) at (-0,-4) {     $A_2 = A_*$\\     $C_2$ : $\left|P\right|$ does not have\\     \hspace{5mm} a nonzero minimum on $\mathbb{C}$.   };     \node [align=left,label=left:$S_3$] (d) at (6,-4) {       $A_3 = \{A_*, C_1, C_2 \}$\\       $C_3 = C_*$     };     \node [align=left,label=left:$S_{2,1}$] (e) at (-6,-8) {       $A_{2,1} = A_*$\\       $C_{2,1}$ : If $\left|P(z) \ne 0\right|$\\       \hspace{2mm} $\exists z'$ s.t. $\left|P(z')\right| < \left|P(z)\right|$\\            };     \node [align=left,label=left:$S_{2,2}$] (f) at (-0,-8) {       $A_{2,2} = \{A_2, C_{2,1}\}$\\       $C_{2,2}$ : if $\left|P\right|$ has a\\       \hspace{2mm} non-zero minimum,\\       \hspace{2mm} we have a contradiction.     };     \node [align=left,label=left:$S_{2,3}$] (g) at (6,-8) {       $A_{2,3} = \{A_2,C_{2,2}\}$\\       $C_{2,3} = C_2$ : $\left|P\right|$ does not have\\       \hspace{2mm} a non-zero minimum\\       \hspace{2mm} on $\mathbb{C}$.     };     \draw[->] (a) to (b);     \draw[->] (a) to (c);     \draw[->] (a) to (d);     \draw [->] (d.west) to [out=120,in=30] (b.north east);     \draw [->] (d.west) to [out=120,in=30] (c.east);     \draw[->] (c) to (e);     \draw[->] (c) to (f);     \draw[->] (c) to (g);     \draw [->] (g.west) to [out=120,in=30] (f.east);     \draw [->] (f.west) to [out=120,in=30] (e.east); \end{tikzpicture}}\end{center}

\caption{Proof by contradiction of the Fundamental Theorem of Algebra. Straight
arrows denote importation of assumptions, while curved arrows denote
importation of conclusions.}
\end{figure}
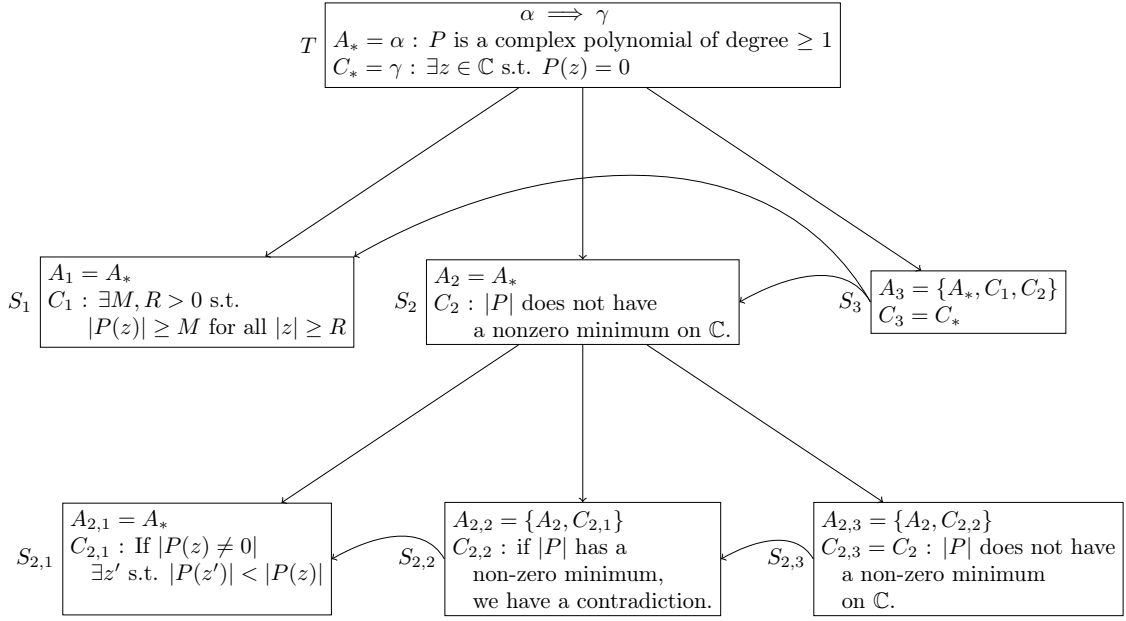

Ultimately, in the above format, proofs by contradictions must explicitly
carry their `wrong assumption' (i.e. the negation of the conclusion)
in the conclusion part of the statements: if we wish to assume the
negation $\neg\mathbf{C}_{*}$ of the conclusion $\mathbf{C}_{*}$
to arrive at a contradiction, this will involve substatements $\mathbf{S}:\mathbf{A}\implies\mathbf{C}$,
where $\mathbf{C}$ will be of the form `if $\neg\mathbf{C}_{*}$
then ...'. While this makes the proofs by contradiction heavier in
notation, this makes individual statements easier to verify: if the
proof is correct, the conclusion of each statement is correct (and
not contingent on an assumption which itself is wrong).

\subsection{\label{subsec:inverse-function-theorem-proof}Inverse Function Theorem
Proof}

A richer example of proof is that of the inverse function theorem. 

In this case, we have:
\begin{itemize}
\item $\alpha$: Let $U\subset\mathbb{R}^{n}$ be an open set and let $f:U\to\mathbb{R}^{n}$
be a function that is $\mathcal{C}^{1}$ with derivative $x\mapsto Df|_{x}$.
Let $x_{*}\in U$ be a point such that $Df|_{x_{*}}$ is an invertible
matrix. 
\item $\gamma$: There exists an open neighborhood $V\subset U$ of $x_{*}$
and an open neighborhood $W$ of $f\left(x_{*}\right)$ such that
$f|_{V}:V\to W$ is a bijection from $V$ to $W$, and an inverse
$\left(f|_{V}\right)^{-1}:W\to U$ that is differentiable at $f\left(x_{*}\right)$
with derivative $\left(Df|_{x_{*}}\right)^{-1}$.
\end{itemize}
In this case, the root of the claim of proof is:
\begin{itemize}
\item $\mathbf{S}_{*}:\mathbf{A}_{*}\implies\mathbf{C}_{*}$, where $\mathbf{A}_{*}=\mathbf{A}_{0}\cup\left\{ \alpha\right\} $
and $\mathbf{C}_{*}=\gamma$. 
\end{itemize}
Informally, we first argue that without loss of generality, we may
assume that $\text{ }x_{*}=0,f\left(x_{*}\right)=0,Df\big|_{x_{*}}=\mathrm{Id}_{n}$.
In this case, the nodes at distance $1$ from the root are: 
\begin{itemize}
\item \textbf{$\mathbf{S}_{1}:\mathbf{A}_{1}\implies\mathbf{C}_{1}$, }where\textbf{
$\mathbf{A}_{1}=\mathbf{A}_{*}$ }and\textbf{ $\mathbf{C}_{1}=\left(\alpha_{1}\implies\gamma\right)$}
with 
\[
\alpha_{1}=\left\{ x_{*}=0,f\left(x_{*}\right)=0,Df\big|_{x_{*}}=\mathrm{Id}_{n}\right\} .
\]
\item $\mathbf{S}_{2}:\mathbf{A}_{2}\implies\mathbf{C}_{2}$, where $\mathbf{A}_{2}:\mathbf{A}_{1}\cup\left\{ \mathbf{C}_{1}\right\} $
and $\mathbf{C}_{2}=\mathbf{C}_{*}$. 
\end{itemize}
If we go further into the details of why $\mathbf{S}_{1}$ holds true,
we find (at distance $2$ from the root):
\begin{itemize}
\item $\mathbf{S}_{1,1}:\mathbf{A}_{1,1}\implies\mathbf{C}_{1,1}$, where
$\mathbf{A}_{1,1}=\mathbf{A}_{1}$, and, denoting by $\|\cdot\|_{\mathcal{M}_{n}}$
the operator norm on $n\times n$ matrices, $\mathbf{C}_{1,1}=\left(\alpha_{1}\implies\gamma_{1,1}\right),$
with $\gamma_{1,1}$ 
\[
\gamma_{1,1}=\left\{ \exists r>0\text{ with \ensuremath{\|}Df\ensuremath{\big|_{x}}-\ensuremath{\mathrm{Id}_{n\times n}\|_{\mathcal{M}_{n}}\leq\frac{1}{2}} \ensuremath{\forall x}\ensuremath{\in}B\ensuremath{\left(0,r\right)}}\right\} .
\]
\item $\mathbf{S}_{1,2}:\mathbf{A}_{1,2}\implies\mathbf{C}_{1,2}$, where
$\mathbf{A}_{1,2}=\mathbf{A}_{1}\cup\left\{ \mathbf{C}_{1,1}\right\} $
and $\mathbf{C}_{1,2}=\left(\alpha_{1}\implies\gamma_{1,2}\right)$,
with
\[
\gamma_{1,2}=\left\{ \exists r>0\text{ such that }\forall y\in\mathbb{R}^{n},\text{ the function }x\mapsto x+y-f\left(x\right)\text{ is }\frac{1}{2}-\mathrm{Lipschitz\text{ on }}B\left(0,r\right)\right\} ,
\]
\item $\mathbf{S}_{1,3}:\mathbf{A}_{1,3}\implies\mathbf{C}_{1,3}$, where
$\mathbf{A}_{1,3}=\mathbf{A}_{1}\cup\left\{ \mathbf{C}_{1,2}\right\} $
and $\mathbf{C}_{1,3}=\left(\alpha_{1}\implies\gamma_{1,3}\right)$,
with
\[
\gamma_{1,3}=\left\{ \exists r>0:\forall y\in B\left(0,r/2\right):\exists!x\in B\left(0,r\right)\text{ such that }f\left(x\right)=y\right\} .
\]
\item $\mathbf{S}_{1,4}:\mathbf{A}_{1,4}\implies\mathbf{C}_{1,4}$, where
$\mathbf{A}_{1,4}=\mathbf{A}_{1}\cup\left\{ \mathbf{C}_{1,3}\right\} $
and $\mathbf{C}_{1,4}=\left(\alpha_{1}\implies\gamma_{1,4}\right)$,
with 
\[
\gamma_{1,4}=\left\{ \exists U,V\text{ open neighborhood of }0\text{ such that }f\text{ is a bijection }U\to V\right\} .
\]
\item $\mathbf{S}_{1,5}:\mathbf{A}_{1,5}\implies\mathbf{C}_{1,5}$, where
$\mathbf{A}_{1,5}=\mathbf{A}_{1}\cup\left\{ \mathbf{C}_{1,4}\right\} $
and $\mathbf{C}_{1,5}=\left(\alpha_{1}\implies\gamma_{1,5}\right)$,
with 
\[
\gamma_{1,5}=\gamma_{1,4}\cap\left\{ \exists f^{-1}:V\to U,\text{ }f^{-1}\text{ is the inverse of }f\text{ and }f^{-1}\text{ is differentiable at }0\text{ with differential }\mathrm{Id}_{n\times n}\right\} .
\]
 
\item $\mathbf{S}_{1,6}:\mathbf{A}_{1,6}\implies\mathbf{C}_{1,6}$, where
$\mathbf{A}_{1,6}=\mathbf{A}_{1}\cup\left\{ \mathbf{C}_{1,5}\right\} $
and $\mathbf{C}_{1,6}=\mathbf{C}_{1}$.
\end{itemize}
If we go further into the details of why $\mathbf{S}_{1,5}$ holds
true, we find (at distance $3$ from the root):
\begin{itemize}
\item $\mathbf{S}_{1,5,1}:\mathbf{A}_{1,5,1}\implies\mathbf{C}_{1,5,1}$,
where $\mathbf{A}_{1,5,1}=\mathbf{A}_{1,5}$ and $\mathbf{C}_{1,5,1}=\left(\alpha_{1}\implies\gamma_{1,5,1}\right)$,
with 
\[
\gamma_{1,5,1}=\gamma_{1,4}\cap\left\{ \text{if \ensuremath{h_{n}} is a seq. in \ensuremath{V\setminus\left\{ 0\right\} } with \ensuremath{h_{n}\to0}, we have }\|h_{n}\|/\|f\left(h_{n}\right)\|\to1\right\} .
\]
\item $\mathbf{S}_{1,5,2}:\mathbf{A}_{1,5,2}\implies\mathbf{C}_{1,5,2}$,
where $\mathbf{A}_{1,5,2}=\mathbf{A}_{1,5}\cup\left\{ \mathbf{C}_{1,5,1}\right\} $
and $\mathbf{C}_{1,5,2}=\left(\alpha_{1}\implies\gamma_{1,5,2}\right)$,
with 
\[
\gamma_{1,5,2}=\gamma_{1,5,1}\cap\left\{ \text{if \ensuremath{h_{n}} is a seq. in \ensuremath{V\setminus\left\{ 0\right\} } with \ensuremath{h_{n}\to0}, we have }\|f\left(h_{n}\right)-h_{n}\|/\|f\left(h_{n}\right)\|\to0\right\} .
\]
\item $\mathbf{S}_{1,5,3}:\mathbf{A}_{1,5,3}\implies\mathbf{C}_{1,5,3}$,
where $\mathbf{A}_{1,5,3}=\mathbf{A}_{1,5}\cup\left\{ \mathbf{C}_{1,5,2}\right\} $
and $\mathbf{C}_{1,5,3}=\left(\alpha_{1}\implies\gamma_{1,5,3}\right)$,
with 
\begin{align*}
\gamma_{1,5,3}= & \gamma_{1,5,2}\\
 & \cap\left\{ \exists f^{-1}:V\to U\text{ }f^{-1}\text{ is the inverse of }f_{\mid U}\right\} \\
 & \cap\left\{ \text{if }\text{\ensuremath{k_{n}} is a seq. in \ensuremath{V\setminus\left\{ 0\right\} } with \ensuremath{k_{n}\to0} we have \ensuremath{\|k_{n}-f{}^{-1}\left(k_{n}\right)\|/\|k_{n}\|\to0}}\right\} .
\end{align*}
\item $\mathbf{S}_{1,5,4}:\mathbf{A}_{1,5,4}\implies\mathbf{C}_{1,5,4}$,
where $\mathbf{A}_{1,5,4}=\mathbf{A}_{1,5}\cup\left\{ \mathbf{C}_{1,5,3}\right\} $
and $\mathbf{C}_{1,5,4}=\mathbf{C}_{1,5}$. 
\end{itemize}
\begin{rem}
As the above example reveals, the context of a proof needs to be explicitly
carried from statement to statement; a good concrete implementation
of the claim of proof format should facilitate this operation in the
writing of proofs.
\end{rem}

\newpage{}

\specialsection{\label{sec:appendix-game-theoretic-analysis}Appendix: Game-Theoretic
Analysis}

In this appendix, we give the proofs of the statements of Section
\ref{sec:a-simplified-equilibrium-analysis}.

\subsection{Proof of Proposition \ref{prop:game}}

\label{eco:proof:prop1} First, note that the expected payoff of Claimer
is increasing in $P$. Hence, for any given anticipated actions of
Skeptic, if initially posting at some $P$ is optimal, then it is
also the case for any $P'>P$. Hence, the entry decision of the claimer
must be of the threshold form given in the Proposition.

Let $h_{e}$ and $h_{1}$ be the histories $\text{(Post, Challenge)}$
and $\text{(Post, Challenge, Reply)}$ respectively, and $\pi_{e}=\mathbb{P}_{h_{e}}(X=1)$,
$\pi_{1}=\mathbb{P}_{h_{1}}(X=1)$ Skeptic's beliefs for these histories.
Note that $\pi_{e}=\frac{1}{2}\left(1+\pi^{*}\right)$.

\subsubsection{Subgame equilibria at $h_{e}$}

\label{eco:sec:subgame} Let us check when (Reply, No Challenge) is
an equilibrium of the subgame at $h_{e}$. In this scenario, Claimer
always replies so observing Reply has no informational content: $\pi_{e}=\pi_{1}$.
For No Challenge to be the best response, we need:

\begin{eqnarray}
-\beta_{1} & \geq & \pi_{1}(-\beta_{1}-\beta_{0})+(1-\pi_{1})(\sigma_{2}+\sigma_{1})\nonumber \\
\text{i.e.}\ \pi_{1} & \geq & \pi_{1}^{*}\equiv\frac{\sigma_{2}+\sigma_{1}+\beta_{1}}{\sigma_{2}+\sigma_{1}+\beta_{1}+\beta_{0}}.\label{eq:eco:eq:condpi1s}
\end{eqnarray}
Because Reply is trivially the best response to No Challenge, we have
constructed an equilibrium of the subgame at $h_{e}$ as soon as $\pi_{1}=\pi_{e}=\frac{1}{2}\left(1+\pi^{*}\right)$
satisfies (\ref{eq:eco:eq:condpi1s}).

Now check when (Reply if $X=1$, Reply if $X=0$ w.p $p$, Challenge
w.p. $q_{1}$) is an equilibrium of the subgame at $h_{e}$. In this
scenario, Bayes' rule indicates that 
\begin{equation}
\pi_{1}=\frac{\pi_{e}}{\pi_{e}+(1-\pi_{e})p}.\label{eco:eq:pi1p}
\end{equation}
For Skeptic to be indifferent between challenging or not, we must
have equality in (\ref{eq:eco:eq:condpi1s}). Using (\ref{eco:eq:pi1p}),
we see that this implies: 
\begin{equation}
p=p(\pi_{e})\equiv\frac{\pi_{e}(1-\pi_{1}^{*})}{\pi_{1}^{*}(1-\pi_{e})}.\label{eco:eq:p}
\end{equation}

Since $p<1$, $\pi_{e}<\pi_{1}^{*}$. When $X=1$, it is trivially
optimal for  Claimer to reply. When $X=0$, she must be indifferent.
Not replying gives payoff $-\sigma_{2}$, while replying gives the
expected payoff $(1-q_{1})(B_{1}+\beta_{1})-q_{1}(\sigma_{2}+\sigma_{1})$.
This pins down the equilibrium value of $q_{1}$: 
\begin{equation}
q_{1}=\frac{B_{1}+\sigma_{2}+\beta_{1}}{B_{1}+\sigma_{2}+\beta_{1}+\sigma_{1}}.\label{eco:eq:q1}
\end{equation}
This completes the description of the equilibria of the subgame at
$h_{e}$. Indeed,  Skeptic cannot be expected to challenge with certainty,
for the best response of Claimer would be to never reply when $X=0$,
which in turn would make systematic challenge suboptimal.

Hence, we have characterized equilibrium expected profits at $h_{e}$: 
\begin{itemize}
\item[-] if $\pi_{e}\geq\pi_{1}^{*}$: $(B_{1}+\beta_{1},-\beta_{1})$ 
\item[-] if $\pi_{e}<\pi_{1}^{*}$: 
\begin{itemize}
\item[-] $X=1$: $(q_{1}(B_{0}+\beta_{1}+\beta_{0})+(1-q_{1})(B_{1}+\beta_{1}),-q_{1}(\beta_{1}+\beta_{0})-(1-q_{1})\beta_{1})$ 
\item[-] $X=0$, $(-\sigma_{2},(1-p)\sigma_{2}+p(-\beta_{1}))$. 
\end{itemize}
\end{itemize}

\subsubsection{Type 1 Equilibria}

For such an equilibrium to exist, we must have 
\begin{eqnarray}
\pi_{e} & \equiv & \frac{1}{2}\left(1+\pi^{*}\right)<\pi_{1}^{*}\label{eco:eq:condpie}\\
0 & < & \phi(p_{e})\equiv\pi_{e}(-q_{1}(\beta_{1}+\beta_{0})-(1-q_{1})\beta_{1})+(1-\pi_{e})((1-p(\pi_{e}))\sigma_{2}-p(\pi_{e})\beta_{1}).\label{eco:eq:condtype1}
\end{eqnarray}

These conditions, obtained from the results of Section \ref{eco:sec:subgame},
ensure that Skeptic has a positive continuation value after Claimer
posts $C$. Indeed, his expected payoff at node $h_{e}$ is positive.
They are sufficient to guarantee an equilibrium of the Type 1 exists,
as soon as Claimer is indeed willing to post if and only if $P\geq\pi^{*}$.
That is, we have an indifference condition at $P=\pi^{*}$, where
the expected profit of posting, $\pi^{*}(q_{1}(B_{0}+\beta_{1}+\beta_{0})+(1-q_{1})(B_{1}+\beta_{1}))-(1-\pi^{*})\sigma_{2}$,
must equate 0, the profit of not posting. Hence 
\begin{equation}
\pi^{*}=\frac{\sigma_{2}}{q_{1}(B_{0}+\beta_{1}+\beta_{0})+(1-q_{1})(B_{1}+\beta_{1})+\sigma_{2}}.\label{eco:eq:condpistareq1}
\end{equation}

The ``$1-\pi_{e}$'' term in the denominator of $p(\pi_{e})$ cancels
out with the ``$1-\pi_{e}$'' term in (\ref{eco:eq:condtype1}),
so that the function $\phi$ is in fact linear. Moreover $\phi(\pi_{1}^{*})<0$
as the only positive term of $\phi$, $(1-p(\pi_{e}))\sigma_{2}$,
vanishes. In particular, if $\phi(0)<0$, a Type 1 equilibrium cannot
exist. The properties of $\phi$ will also be important to characterize
Type 2 equilibria.

A Type 1 equilibrium exists if and only if conditions (\ref{eco:eq:condpie}),
(\ref{eco:eq:condtype1}) and (\ref{eco:eq:condpistareq1}) are simultaneously
satisfied.

\subsubsection{Type 2 Equilibria}

For such an equilibrium to exist, Skeptic must be indifferent between
challenging the initial claim or not. Hence, we must have 
\begin{eqnarray}
\pi_{e} & \equiv & \frac{1}{2}\left(1+\pi^{*}\right)<\pi_{1}^{*}\label{eco:eq:condpieeq2}\\
0 & = & \phi(\pi_{e}).\label{eco:eq:skepticindiff}
\end{eqnarray}
If (\ref{eco:eq:condpieeq2}) is not satisfied, Skeptic makes a negative
profit by continuing because he cannot challenge back upon reply of
Claimer. Hence, he cannot be indifferent between challenging the initial
claim or not. 
Equation (\ref{eco:eq:skepticindiff}) writes down explicitly the
payoff of replying when (\ref{eco:eq:condpieeq2}) holds. (\ref{eco:eq:skepticindiff})
has a valid root if and only if 
\begin{equation}
\phi(0)\geq0.\label{eco:eq:condphi0eq2}
\end{equation}

Claimer should also be indifferent between posting $C$ and not posting
when $P=\pi^{*}$. That is, her expected payoff of posting, $(1-q_{2})B_{2}+q_{2}(\pi^{*}(q_{1}(B_{0}+\beta_{1}+\beta_{0})+(1-q_{1})(B_{1}+\beta_{1}))-(1-\pi^{*})\sigma_{2})$,
should be 0. This gives: 
\begin{equation}
q_{2}=\frac{B_{2}}{B_{2}-\pi^{*}(q_{1}(B_{0}+\beta_{1}+\beta_{0})+(1-q_{1})(B_{1}+\beta_{1}))+(1-\pi^{*})\sigma_{2}}.\label{eco:eq:condq0eq2}
\end{equation}

A Type 2 equilibrium exists if and only if conditions (\ref{eco:eq:condpieeq2}),
(\ref{eco:eq:condphi0eq2}) and (\ref{eco:eq:condq0eq2}) are simultaneously
satisfied, with $0<q_{2}<1$.

\subsubsection{Type 3 Equilibria}

From the previous analysis, it is now clear that if $\frac{1}{2}=\pi_{e}(\pi^{*}=0)\geq\pi_{1}^{*}$
or $\frac{1}{2}<\pi_{1}^{*}$ but $\phi(0)<0$ then we have a type
3 equilibrium: Claimer always enters the game and Skeptic never challenges.

\subsubsection{Existence and Uniqueness}

If $\frac{1}{2}\geq\pi_{1}^{*}$, a Type 3 equilibrium exists and
no equilibrium of Type 1 or Type 2 can exist. From now on, assume
$\frac{1}{2}<\pi_{1}^{*}$, and successively (i) $\phi(0)<0$ and
(ii) $\phi(0)\geq0$.

Case (i): we know that a Type 3 equilibrium exists and we have seen
that no equilibrium of Type 1 or Type 2 can exist. (That $\phi(0)<0$
indicates that Skeptic does not want to challenge even under the worst
possible belief about the correctness of $C$. Hence, for any belief
$\pi^{*}$ about the posting threshold, Skeptic would also find it
optimal not to challenge.)

Case (ii): we know that no equilibrium of Type 3 exists. Assume a
Type 2 equilibrium exists, characterized by, say, $\pi_{T2}^{*}$
and $q_{2,T2}$. Recall that we have 
\begin{equation}
0=(1-q_{2,T2})\underbrace{B_{2}}_{>0}+q_{2,T2}\underbrace{(\pi_{T2}^{*}(q_{1}(B_{0}+\beta_{1}+\beta_{0})+(1-q_{1})(B_{1}+\beta_{1}))-(1-\pi_{T2}^{*})\sigma_{2})}_{\text{payoff of Claimer if she is challenged}}
\end{equation}
with $0<q_{2,T2}<1$. This implies that Claimer expects a negative
profit conditional on being challenged at $\pi_{T2}^{*}$. In particular,
if a Type 1 equilibrium were to exist, it would need to feature $\pi^{*}=\pi_{T1}^{*}>\pi_{T2}^{*}$.
But $\phi$ decreases (the incentives of Skeptic to challenge decrease
with the probability that Claimer is right), so $\phi\left(\pi_{T1}^{*}\right)<\phi\left(\pi_{T2}^{*}\right)=0$
and Skeptic has no incentive to challenge, so that one cannot construct
a Type 1 equilibrium.

At this stage, we have seen that the different types of equilibria
are mutually exclusive. To show that there is always one, remark that
if Type 2 and Type 3 equilibria do not exist, then $\frac{1}{2}<\pi_{1}^{*}$,
$\phi(0)\geq0$ but there is no value of $q_{2}\in(0,1)$ such that
(\ref{eco:eq:condq0eq2}) holds. This means that at the unique root
$\pi$ of $\phi$ over $[0,\pi_{1}^{*}]$, for all $q_{2}\in[0,1]$,
the expected payoff of posting is non-negative. In particular, this
holds at $q_{2}=1$: at $\pi$ (meaning: when $P=\pi$ and under the
belief that Claimer posts if and only if $P\geq\pi$), Claimer is
willing to post even conditional on Skeptic always challenging, and
Skeptic is indifferent between challenging or not. As the candidate
$\pi^{*}$ decreases away from $\pi$, the incentives to Challenge
increase, and the expected payoff of Claimer decrease at $\pi^{*}$.
Hence, if we define $\pi^{*}$ as the infimum of the $\pi$ such that
Claimer is willing to post even conditional on  Skeptic always challenging
(a bounded, non-empty set from what we saw above), we have at $\pi^{*}$
that  Claimer is indifferent between posting or not, and that Skeptic
will always challenge: we have constructed a Type 1 equilibrium.

\subsection{Proof of Proposition \ref{prop:prob-extraction}}

We will need the following:
\begin{lem}
\label{eco:lemma:temp} The probabilities that Claimer enters the
game conditional on the claim of proof being correct (resp. incorrect)
are 
\begin{eqnarray}
\mathbb{P}(P\geq\pi^{*}\lvert X=1) & = & \left(1+\pi^{*}\right)\left(1-\pi^{*}\right)\\
\mathbb{P}(P\geq\pi^{*}\lvert X=0) & = & \left(1-\pi^{*}\right)^{2}.
\end{eqnarray}
\end{lem}

\begin{proof}
From Bayes' formula, 
\begin{equation}
\mathbb{P}(P\geq\pi^{*}\lvert X=1)=\frac{\mathbb{E}[X\lvert P\geq\pi^{*}]\mathbb{P}(P\geq\pi^{*})}{\mathbb{P}(X=1)}.
\end{equation}
Since $\mathbb{E}[X\lvert P]=P$ and $P$ is uniformly distributed,
\begin{equation}
\mathbb{E}[X\lvert P\geq\pi^{*}]=\mathbb{E}\left[\mathbb{E}[X\lvert P]\lvert P\geq\pi^{*}\right]=\mathbb{E}[P\lvert P\geq\pi^{*}]=\frac{1+\pi^{*}}{2},
\end{equation}
and $\mathbb{P}(P\geq\pi^{*})=1-\pi^{*}$, $\mathbb{P}(X=1)=\frac{1}{2}$,
which yields the first equality. The second one is obtained using
similar arguments. 
\end{proof}
We are now in a position to prove the results relative to Type 1 equilibria.
We first compute the probabilities that a claim of proof is accepted
given that it is correct/incorrect. Since a correct claim of proof
is accepted if and only if Claimer posts, $\mathbb{P}(\mathbf{A}\lvert X=1)=\mathbb{P}(P\geq\pi^{*}\lvert X=1)$,
the value of which is given in Lemma \ref{eco:lemma:temp}. An incorrect
claim of proof is accepted if and only if Claimer posts, then replies
and Skeptic does not challenge at the last step. Hence this has probability
$\mathbb{P}(\mathbf{A}\lvert X=0)=\mathbb{P}(P\geq\pi^{*},\mathbf{R},\mathbf{Q_{0}}^{c}\lvert X=0)$.
Using the fact that $\mathbb{P}(P\geq\pi^{*},\mathbf{R},\mathbf{Q_{0}}^{c}\lvert X=0)=\mathbb{P}(P\geq\pi^{*}\lvert X=0)\mathbb{P}(\mathbf{R})\mathbb{P}(\mathbf{Q_{0}}^{c})$
and using Lemma \ref{eco:lemma:temp}, we obtain the formula for $\mathbb{P}(\mathbf{A}\lvert X=0)$.

Note that the probability that a proof is true is $\frac{1}{2}$.
Hence 
\begin{eqnarray}
\mathbb{P}(\mathbf{A},X=1) & = & \frac{1}{2}\mathbb{P}(\mathbf{A}\lvert X=1)\\
\mathbb{P}(\mathbf{A}) & = & \frac{1}{2}\left(\mathbb{P}(\mathbf{A}\lvert X=1)+\mathbb{P}(\mathbf{A}\lvert X=0)\right).
\end{eqnarray}
This yields the result for the probability that a claim of proof is
accepted and true as well as accepted.

The probabilities that a proof is true given that it is accepted/rejected
are obtained by applying 
\begin{eqnarray}
\mathbb{P}(X=1\lvert\mathbf{A}) & = & \frac{\mathbb{P}(\mathbf{A},X=1)}{\mathbb{P}(\mathbf{A})}\\
\mathbb{P}(X=1\lvert\mathbf{A}^{c}) & = & \frac{\left(1-\mathbb{P}(\mathbf{A}\lvert X=1)\right)\mathbb{P}(X=1)}{1-\mathbb{P}(\mathbf{A})}.
\end{eqnarray}

Finally, in a Type 1 equilibrium, Skeptic always challenges at the
first step, so that $\mathbb{P}(\mathbf{A}_{2}\lvert\mathbf{A},X=1)=0$.
Moreover, if $X=1$, the proof finishes after the Reply of Claimer
if and only if Skeptic renounces to challenge. This occurs with probability
$1-q_{1}$, hence $\mathbb{P}(\mathbf{A}_{1}\lvert\mathbf{A},X=1)=1-q_{1}$.

Using similar computations, one can obtain these event probabilities
in the case of a Type 2 equilibrium. Specifically, in a Type 2 equilibrium: 
\begin{itemize}
\item The probabilities that a claim of proof is accepted (resp. accepted
and true) are 
\begin{eqnarray}
\mathbb{P}(\mathbf{A}) & = & \frac{1}{2}\left(1+\pi^{*}\right)\left(1-\pi^{*}\right)+\frac{1}{2}\left(1-\pi^{*}\right)^{2}\left(1-q_{2}+q_{2}p(1-q_{1})\right)\\
\mathbb{P}(\mathbf{A},X=1) & = & \frac{1}{2}\left(1+\pi^{*}\right)\left(1-\pi^{*}\right).
\end{eqnarray}
\item The probabilities that a claim of proof is accepted given that it
is true (resp. false) are 
\begin{eqnarray}
\mathbb{P}(\mathbf{A}\lvert X=1) & = & \left(1+\pi^{*}\right)\left(1-\pi^{*}\right)\\
\mathbb{P}(\mathbf{A}\lvert X=0) & = & \left(1-\pi^{*}\right)^{2}(1-q_{2}+q_{2}p(1-q_{1})).
\end{eqnarray}
\item The probabilities that a claim of proof is true given that it is accepted
(resp. rejected) are 
\begin{eqnarray}
\mathbb{P}(X=1\lvert\mathbf{A}) & = & \frac{1+\pi^{*}}{1+\pi^{*}+\left(1-\pi^{*}\right)(1-q_{2}+q_{2}p(1-q_{1}))}\\
\mathbb{P}(X=1\lvert\mathbf{A}^{c}) & = & \frac{\pi^{*2}}{\pi^{*2}+1-\left(1-\pi^{*}\right)^{2}(1-q_{2}+q_{2}p(1-q_{1}))}.
\end{eqnarray}
\item The probabilities that a claim of proof is accepted at level 0 (resp.
1) given that it is accepted and true are: 
\begin{eqnarray}
\mathbb{P}(\mathbf{A}_{2}\lvert\mathbf{A},X=1) & = & 1-q_{2}\\
\mathbb{P}(\mathbf{A}_{1}\lvert\mathbf{A},X=1) & = & q_{2}(1-q_{1}).
\end{eqnarray}
The computations of the probabilities for an equilibrium of Type 3
are trivial and omitted.
\end{itemize}

\end{document}